\documentclass[12pt,reqno]{amsart}
\usepackage{amsfonts}
\usepackage{amssymb}
\usepackage{amsbsy}
\usepackage{amsthm}
\usepackage{amscd}
\usepackage{color}
\usepackage{amsmath}
\usepackage{graphicx}

\usepackage{layout}

\newtheorem{thm}{Theorem}[section]
\newtheorem{cor}[thm]{Corollary}
\newtheorem{lem}[thm]{Lemma}
\newtheorem{prop}[thm]{Proposition}
\theoremstyle{definition}
\newtheorem{defn}[thm]{Definition}
\theoremstyle{remark}
\newtheorem{rem}[thm]{Remark}
\newtheorem{ex}[thm]{\bf Example}

\newcommand{\bt}{\begin{thm}}
\newcommand{\et}{\end{thm}}
\newcommand{\bc}{\begin{cor}}
\newcommand{\ec}{\end{cor}}
\newcommand{\bl}{\begin{lem}}
\newcommand{\el}{\end{lem}}
\newcommand{\bp}{\begin{prop}}
\newcommand{\ep}{\end{prop}}
\newcommand{\bd}{\begin{defn}}
\newcommand{\ed}{\end{defn}}
\newcommand{\br}{\begin{rem}}
\newcommand{\er}{\end{rem}}
\newcommand{\bpr}{\begin{proof}}
\newcommand{\epr}{\end{proof}}
\newcommand{\bex}{\begin{ex}}
\newcommand{\eex}{\end{ex}}
\newcommand{\bcd}{\begin{CD}}
\newcommand{\ecd}{\end{CD}}

\newcommand{\bi}{\begin{itemize}}
\newcommand{\ei}{\end{itemize}}
\newcommand{\be}{\begin{enumerate}}
\newcommand{\ee}{\end{enumerate}}
\newcommand{\ba}{\begin{array}}
\newcommand{\ea}{\end{array}}
\newcommand{\beq}{\begin{equation}}
\newcommand{\eeq}{\end{equation}}
\newcommand{\beqa}{\begin{eqnarray}}
\newcommand{\eeqa}{\end{eqnarray}}
\newcommand{\bca}{\begin{cases}}
\newcommand{\eca}{\end{cases}}
\newcommand{\bal}{\begin{aligned}}
\newcommand{\eal}{\end{aligned}}

\newcommand{\ds}{\displaystyle}
\newcommand{\ts}{\textstyle}

\newcommand{\N}{{\mathbb N}}
\newcommand{\Z}{{\mathbb Z}}
\newcommand{\R}{{\mathbb R}}
\newcommand{\C}{{\mathbb C}}
\newcommand{\T}{{\mathbb T}}
\newcommand{\D}{{\mathbb D}}

\newcommand{\cC}{{\mathcal  C}}
\newcommand{\cD}{{\mathcal  D}}

\newcommand{\frU}{{\mathfrak U}}

\newcommand{\bs}{\boldsymbol}

\newcommand{\bsalpha}{{\boldsymbol \alpha}}

\newcommand{\supp}{\mathrm{supp}}

\newcommand{\1}{{\bf 1}}
\newcommand{\0}{{\bf 0}}

\newcommand{\re}{\mathrm{Re}}
\newcommand{\im}{\mathrm{Im}}

\newcommand{\sgn}{\mathrm{sgn}}

\newcommand{\x}x
\newcommand{\<}{\langle}
\renewcommand{\>}{\rangle}
\newcommand{\uk}{|\kern-3.3pt\uparrow\>}
\newcommand{\dk}{|\kern-3.3pt\downarrow\>}
\newcommand{\ub}{\<\uparrow\kern-3.3pt|}
\newcommand{\db}{\<\downarrow\kern-3.3pt|}
\newcommand{\ua}{\uparrow}
\newcommand{\da}{\downarrow}

\begin{document}


\title{\bf One-dimensional quantum walks \break with one defect}

\author[CGMV]{M. J. Cantero, F. A. Gr\"unbaum, L. Moral, L. Vel\'azquez}

\address[F. A. Gr\"unbaum]{Department of Mathematics \\ University of California \\ Berkeley \\ CA \\ 94720}
\email[F. A. Gr\"unbaum]{grunbaum@math.berkeley.edu}

\address[M. J. Cantero, L. Moral, L. Vel\'azquez]{Departamento de Matemática Aplicada \\ Universidad de Zaragoza \\ Zaragoza \\ Spain}
\email[M. J. Cantero]{mjcante@unizar.es}
\email[L. Moral]{lmoral@unizar.es}
\email[L. Vel\'azquez]{velazque@unizar.es}

\thanks{The research of the second author was supported in part by the Applied Math. Sciences subprogram of the Office of Energy Research, USDOE,
        under Contract DE-AC03-76SF00098.}
\thanks{The research of the rest of the authors was partly supported by the Spanish grant from the Ministry of Science and Innovation, project code
        MTM2008-06689-C02-01, and by Project E-64 of Diputaci\'on General de Arag\'on (Spain).}

\subjclass[2000]{81P68, 47B36, 42C05}

\keywords{Quantum walks, localization, CGMV method, CMV matrices, scalar and matrix Laurent orthogonal polynomials on the unit circle}

\date{}

\begin{abstract}
The CGMV method allows for the general discussion of localization properties for the states of a one-dimensional quantum walk, both in the case of
the integers and in the case of the non negative integers. Using this method we classify, according to such localization properties, all the quantum
walks with one defect at the origin, providing explicit expressions for the asymptotic return probabilities to the origin.
\end{abstract}

\maketitle

\section{Introduction} \label{S:Int}

A quantum random walk can be considered as a quantum analog of the more familiar classical random walk on a lattice. In this much simpler case the
study of interesting ``return properties" can be said to have started with G. Pólya (1921), \cite{Po}, who proved that the simplest unbiased walk
will eventually return to the origin with probability one in dimension not greater than two. This holds in spite of the fact that the probability of
returning to the origin in $n$ steps, denoted by $p(n)$, converges to zero as $n$ tends to infinity.

In this paper we consider aspects of this problem in the context of quantum random walks (QWs). We give a method that allows us to analyze the
asymptotic behaviour of the quantity $p(n)$ for two-state one-dimensional QWs, leading to the discovery of general features of this asymptotics in
the case of distinct coins. We also apply this method to the QWs that are given by one arbitrary common coin at each site except for an arbitrary
``defective" coin at the origin. Before giving a summary of the results in the paper we give a brief review of the more standard case of a classical
random walk.

In this more traditional case, and for the simplest unbiased walk, one can obtain an expression for $p(n)$ in many different ways. This is basically
true since one is dealing with a translation invariant evolution. As soon as this condition is relaxed things become much harder. One of the methods
that can (at least in theory) give an expression for $p(n)$ goes back to the work of S. Karlin and J. McGregor (1959), \cite{KMcG}. Their method
applies to a birth-and-death process on the nonnegative integers but they themselves already contemplated extending their method to such processes on
the integers by using matrix valued objects. This has been implemented recently and independently by H. Dette et al. (2006), \cite{DRSZ},  as well as
by F.A. Gr\"unbaum (2007), \cite{G1}. In both cases one makes crucial use of the matrix valued orthogonal polynomials introduced by M.G. Krein
(1949), \cite{K1,K2}.

The Karlin-McGregor (KMcG) method alluded to above proceeds by starting from a very simple case, say a situation with no defects, considering its
orthogonality measure and related function theoretical objects such as its Stieltjes transform. One then introduces one defect and studies the effect
that this has on the orthogonality measure. In this way one can get an expression for the new value of $p(n)$. Clearly this process can be iterated a
finite number of times to obtain situations that are far from the initial basic case. The KMcG method proceeds by studying the effect of the defect
on the Stieltjes transform of the measure and it produces much more than an expression for the new value of $p(n)$. The interested reader can find a
host of examples treated in the fashion described above in the papers \cite{C,DRSZ,G1,G2,G3,G4,G5,GdI}.

In a previous paper, \cite{CGMV}, we found an analog of the KMcG method that can be used in the case of QWs on the integers and the non negative
integers. This method has recently been used by other authors, see \cite{KS}. We use their terminology and call this the CGMV method. In the CGMV
method the central role is played by a convenient spectral representation of the unitary operator that describes the dynamics of the walk. In that
sense this is very close in spirit to the mathematical foundations of Quantum Mechanics laid down by people like E. Schröedinger, W. Heisenberg, M.
Born and J. von Neumann.

The connection with QWs is embodied in the fact that a unitary operator has a nice five-diagonal representation, see \cite{CMV}. While this is true
for any unitary operator, the search for such a five-diagonal representation can be a difficult task in the infinite dimensional case. However, for a
standard two-state one-dimensional QW the five-diagonal representation comes from an appropriate reordering of the usual basis of states.

While the method itself is fairly similar to the one used by Karlin and McGregor, the results that one obtains in the quantum case are, regardless of
the method, often intrinsically different from the classical ones. For instance, the application of CGMV tools yields in \cite{CGMV} a number of
situations dealing with the notion of recurrence, where one sees that classical and quantum walks have little in common. Some of these differences
were previously discovered in \cite{SJK1,SJK2,SJK3} using a Fourier approach to quantum recurrence.

It is possible that many more such differences will become apparent as one develops good tools to study these kinds of walks. For instance,
appropriate mathematical techniques such as
the fractional moment method, \cite{AM}, have been key to prove for QWs with random coins a peculiarity of quantum systems in random environments
called ``Anderson localization": the wave packets stay trapped in a finite region of space for all time. There is a huge literature on this subject
that has its origin in a discrete model in solid state physics, see \cite{And}. The interested reader may get a guide to the literature as well as
very nice discussion of this notion by consulting \cite{HJS}. Concerning Anderson localization in QWs see \cite{JM}.

For an environment that is not too disordered, weaker ``localization" properties could distinguish quantum from classical random walks too. A
candidate for such a property is related to the asymptotics of the return probability to a given site. To be more precise we start at a qubit state
$(\alpha,\beta)$ at a site $k$ on the lattice and look at the probability $p_{\alpha,\beta}^{(k)}(n)$ that after time $n$ the state will again be
localized at site $k$ but with unspecified spin orientation. We adopt the terminology of \cite{K} and say that the qubit $(\alpha,\beta)$ at a site
$k$ exhibits localization when $p_{\alpha,\beta}^{(k)}(n)$ does not converge to zero as $n$ tends to infinity. We will deal with this notion of
localization of single qubit states, which should not be confused with such a global property of disordered quantum systems like Anderson
localization, and sometimes we will refer to it as ``single state localization".

Irreducible two-state QWs on the integers with a constant coin can not exhibit localization, see \cite{ABNVW} for instance. Nevertheless,
localization can appear for homogeneous QWs if we increase the internal degrees of freedom, \cite{IKS,IK}, or the dimension of the lattice,
\cite{IKK,WKKK}. A way to get localization in one dimension keeping a constant coin is to destroy the translation invariance by considering the
lattice of the non negative integers, see \cite{CGMV,KS}.

Another way of breaking the translation invariance is to introduce defects on the integers. In this context an important difference between classical
and quantum random walks has already been recognized concerning localization of single states. In \cite{K}, N. Konno has studied a perturbation of
the Hadamard QW on the integers and proved that in the case of a particular defect at the origin, and starting from a particular qubit too, the
quantity $p_{\alpha,\beta}^{(0)}(n)$ fails to converge to zero. The method used in \cite{K} is a path counting argument. One can think of our paper
as an effort to explore the phenomenon uncovered in \cite{K} in a more general setting by using the CGMV method.

Localization can depend on the initial state as well as on the coins of the QW. Our aim is to get a better understanding of these dependencies. This
appears to be a difficult task if one resorts to the usual approaches, like path counting or Fourier transform, which usually allow for the
computation of the asymptotics of $p_{\alpha,\beta}^{(k)}(n)$ at specific initial qubits or for very limited coin models.

The state and coin dependencies of localization seem to be more tractable from the CGMV point of view. This is specially true for the state
dependence, and we will be able to have a picture of it for a large class of QWs. The coin dependence, which is much more involved, will be discussed
in some explicit examples. More precisely, the CGMV approach will allow us to perform a complete classification, according to the localization
behaviour, of all the QWs with a coin which is constant except for a defect at the origin.

In addressing the study of the dynamics of any quantum system one should keep in mind that there is a large literature on the subject. The classical
book ``Non-relativistic quantum dynamics" by W. Amrein (1981) gives a good treatment. For a more recent account see the book ``Hilbert space methods
in quantum mechanics" (2009) by the same author.

Most of the work deals with the spectral properties of the unitary group that implements the quantum evolution and its dynamical consequences. There
are also several papers dealing with these issues of which we just mention two: at a fairly technical level one can consult \cite {L}, or later work
of this same author, and at a more approachable level see \cite{Kn}.

In a number of ways one can say that what we do in great detail is related to the bread and butter of Quantum Mechanics. What we call the return
probability $p_{\alpha,\beta}^{(k)}(n)$ is related (but not identical) to what Y. Last, \cite{L}, calls the ``survival probability", see (1.3) of his
paper. The main point of our paper is that we manage to compute these quantities explicitly and then study their asymptotic values as $n$ goes to
infinity.

We now try to give an account of the contents of the present paper, which deals with the localization properties of two-state QWs on the integers and
the non negative integers.

Just as in the KMcG method used for classical random walks one studies the effect that introducing defects on a simpler walk has on the Stieltjes
transform of the orthogonality measure, in the quantum case we need to study the so called Schur function of the measure. The analysis of the special
features of the Schur functions related to QWs with distinct coins, with special emphasis in the case of the integers, takes up a good part of
section \ref{S:CMV} in the present paper. The results of this section will be key for the rest of the paper.

The general problem of the single state localization within the CGMV language is studied in section \ref{S:Loc}. The main result is Theorem
\ref{T:Loc}, which establishes a connection between localization in a QW and the singular part of the corresponding orthogonality measure. In
particular, the absence of such a singular part leads to the absence of localized states, while the presence of mass points implies the existence of
states which exhibit localization.

Among other consequences the CGMV method shows that, despite its name, single state localization is in fact a quasi global property for a large class
of QWs so that a localization dichotomy holds in many cases: either no state exhibits localization or at most one state per site is localization
free. Such a dichotomy is ensured when the singular part of the measure is not purely continuous. That is, the state dependence of localization is
quite regular for a wide class of QWs. In this case we will refer to QWs with or without localization omitting any mention to the initial qubit
state.

In section \ref{S:Loc} we see that the localization dichotomy holds in particular for any QW with periodic coins up to a finite number of defective
coins. It is also shown that, among these QWs, the case of strictly periodic coins on the integers is somewhat special because it never gives
localization. QWs on the integers whose coins have period $P$ are the only one-dimensional QWs which are invariant with respect to right and left
translations of $P$ sites. Thus, we can state that the requirement of a (right and left) translation invariance for a one-dimensional QW forces the
absence of localization.

In sections \ref{S:defect} to \ref{S:Ret-defect-Z+} these results are made more specific, both for the non negative integers as well as for the
integers, in the case of the simplest perturbation of the constant coin model: the QWs with a coin which is constant except for one defect at the
origin. These QWs will serve as a laboratory to study the coin dependence of localization in QWs.

As we pointed out, localization already appears for two-state QWs with a constant coin on the non negative integers but not on the integers where the
related measure is absolutely continuous (see \cite{CGMV}). Therefore, the study of QWs with one defect acquires a special relevance in the case of
the integers because they are the nicest laboratory in which one can study localization on such a lattice. Nevertheless, we will perform also the
analysis of one defect on the non negative integers, which will allow us to compare with the case of the integers, thus showing the effects of the
boundary conditions on the localization behaviour.

As a particular case of the periodic QWs with a finite number of defects, the localization dichotomy holds for QWs with one defect. The analysis of
localization in these models becomes the study of the presence of mass points in the corresponding measure.

Section \ref{S:defect} deals with the general features of the orthogonality measure for QWs with one defect at the origin. It is shown that they fall
into groups with the same measure up to rotations, thus with the same localization behaviour. These groups are labelled by two parameters $a,b$ in
the unit disk for the non negative integers, while an additional labelling parameter $\omega$ in the unit circle appears in the case of the integers.

Section \ref{S:Loc-defect-Z} shows that $\omega$ actually plays no role in the presence or absence of localization, hence localization for one defect
at the origin only depends on two parameters $a,b$ which are defined by the coins of the QW (in a different fashion for the integers and the non
negative integers, see (\ref{E:ab-Z+}) and (\ref{E:ab-Z})). The parameter $a$ depends only on the unperturbed coin and the phases of the perturbed
one, while $b$ depends on the perturbed coin and the phases of the unperturbed one.

Sections \ref{S:Loc-defect-Z} and \ref{S:Ret-defect-Z} give a very exhaustive analysis of localization for one defect on the integers, and the same
in depth analysis is carried out in sections \ref{S:Loc-defect-Z+} and \ref{S:Ret-defect-Z+} for the case of the non negative integers. Sections
\ref{S:Loc-defect-Z} and \ref{S:Loc-defect-Z+} discuss the coin dependence, providing a characterization of localization in terms of the parameters
$a,b$. Sections \ref{S:Ret-defect-Z} and \ref{S:Ret-defect-Z+} yield explicit results for the asymptotic return probability to the origin for any
defect and any initial qubit state.

Different localization figures in the space of parameters $a,b$ are presented in sections \ref{S:Loc-defect-Z} and \ref{S:Loc-defect-Z+}. They
demonstrate that, in contrast to the classical case (see \cite[Section 6]{K}), localization is dominant under the presence of a defect. Nevertheless,
these figures also show that, at the same time, there are situations where the absence of localization is stable under small perturbations of $a$ and
$b$, i.e, under small perturbations of the coins. In particular, given $|a|$, the largest regions for the parameter $b$ without localization appear
when $a$ is imaginary, both for the integers and the non negative integers. Then there is no localization if $\im b \geq \im a > 0$ or $\im b \leq
\im a < 0$.

Concerning the return probabilities $p_{\alpha,\beta}^{(k)}(n)$, the CGMV method not only shows the dependence of its asymptotics on the initial
qubit $(\alpha,\beta)$, but also explains the reason for its possible oscillatory asymptotic behaviour: the presence of the factors $z^n$ in
(\ref{E:ARP}), where $z$ are the mass points of the measure. The return probabilities turn out to be convergent when the singular part of the measure
is a unique mass point or, in the case of several mass points, when some symmetries force the mutual cancellation of the cross terms in
(\ref{E:ARP}).

General reasons imply that the return probabilities $p_{\alpha,\beta}^{(k)}(2n-1)$ at odd time vanish for any QW on the integers, which is related to
the fact that the mass points on the integers always appear in pairs which are symmetric with respect to the origin. Therefore, in the presence of
localization on the integers we can not expect the convergence of $p_{\alpha,\beta}^{(k)}(n)$ but at most of $p_{\alpha,\beta}^{(k)}(2n)$. This
convergence takes place for sure when the singular part of the measure is a single pair of symmetric mass points.

QWs on the integers with one defect at the origin have a symmetry under reflection of the sites with respect to the origin which causes the
convergence of $p_{\alpha,\beta}^{(0)}(2n)$ regardless of the number of mass points. However, for one defect on the non negative integers with more
than one mass point, as well as when considering the return probability $p_{\alpha,\beta}^{(k)}(2n)$ to a site $k \neq 0$ on the integers with more
than two mass points, the asymptotic behaviour is in general oscillatory.

The state dependence can also disappear in some special situations like, for instance, one defect at the origin on the integers with an imaginary
value of $a$. In this case section \ref{S:Ret-defect-Z} proves that $p_{\alpha,\beta}^{(0)}(2n)$ actually converges to the same limit for any initial
qubit. This covers as a special case the result obtained in \cite{K} for a concrete perturbation of the Hadamard model and a specific initial state.
We not only prove that the result in \cite{K} is state independent, but we also extend this to a more general model with one defect since we find
that the state independence holds whenever the products of the diagonal coefficients of the perturbed and the unperturbed coins have the same phase.

Section \ref{S:Ret-defect-Z} gives explicitly $p_{\alpha,\beta}^{(0)}=\lim_{n\to\infty}p_{\alpha,\beta}^{(0)}(2n)$ for any QW on the integers with
one defect at the origin. The convergence of $p_{\alpha,\beta}^{(0)}(2n)$ allows for the analysis of the maximum asymptotic return probabilities to
the origin, both when running over the qubits $(\alpha,\beta)$ and also when running over the parameters $a,b$, i.e., over the coins of the model. We
find that $\max_{\alpha,\beta}p_{\alpha,\beta}^{(0)}$ approaches one when $|a|\to1$ provided that $|\im a - \im b|$ is bounded from below. The
consequence is that, given a defective coin, for most of the choices of the non defective coin there exist qubits which asymptotically return to the
origin with probability almost one, as long as the non defective coin is close enough to an anti-diagonal one.

In the special case of an imaginary $a$, the asymptotic return probability $p_{\alpha,\beta}^{(0)}$ does not depend on the state and approaches one
when $|a|\to1$ if $|a-b|$ is bounded from below. This implies that, when the products of the diagonal coefficients of both coins have similar phases,
all the qubits asymptotically return to the origin with probability almost one, provided that the non defective coin is close enough to an
anti-diagonal one.

The results above not only show the strength of the CGMV method for the analysis of localization in QWs, but they address new research lines with
could lead to new and surprising quantum effects. For instance, it could be very interesting to analyze the localization behaviour of QWs where the
localization dichotomy is not ensured, i.e., those whose measure has a singular part which is strictly continuous, and specially those with a purely
singular continuous measure. The physical consequences of singular continuous spectra in Quantum Mechanics is an active field of research, see for
instance \cite{L} and the references therein. The study of this problem in those models which can be considered the simplest realization of a
dynamical quantum system, i.e., the QWs on a lattice, could make it easier to understand the quantum meaning of a singular continuous spectrum and
its dynamical implications.

\section{QWs, CMV matrices and Schur functions} \label{S:CMV}

Throughout the paper we will deal with QWs on a state space with a basis $\{|k\!\ua\>,|k\!\da\>\}_{k\in\Z\text{ or }\Z_+}$, where
$\Z_+=\{0,1,2,\dots\}$. The quantum dynamics will be governed by unitary coins
$$
C_k = \begin{pmatrix} c_{11}^{(k)} & c_{12}^{(k)} \\ c_{21}^{(k)} & c_{22}^{(k)} \end{pmatrix},
\qquad c_{jj}^{(k)} \neq 0, \qquad j=1,2, \qquad k\in\Z \text{ or } \Z_+,
$$
so that the one-step transitions are given by the operator $\frU$ defined by
$$
\frU |k\!\ua\> = c_{11}^{(k)} |k+1\!\ua\> + c_{21}^{(k)} |k-1\!\da\>, \quad \frU |k\!\da\> = c_{12}^{(k)} |k+1\!\ua\> + c_{22}^{(k)} |k-1\!\da\>,
$$
except when $k=0$ for $\Z_+$, in which case the unitarity forces a transition
$$
\frU |0\!\ua\> = c_{11}^{(0)} |1\!\ua\> + c_{21}^{(0)} |0\!\ua\>, \qquad \frU |0\!\da\> = c_{12}^{(0)} |1\!\ua\> + c_{22}^{(0)} |0\!\ua\>.
$$
If a diagonal element of a coin were null, the QW would decouple into independent ones, so the requirement $c_{jj}^{(k)}\neq0$ is not really a
restriction, but it simply means that we are considering only irreducible QWs.

Once an order is chosen for the basis $\{|k\!\ua\>,|k\!\da\>\}_{k\in\Z\text{ or }\Z_+}$, it gives a matrix representation $\bs{U}$ of the transition
operator $\frU$, which we will call the transition matrix of the QW. If $|i\>$ denotes the $i$-th vector of such an ordered basis, we will use the
convention that $\bs{U}=(\bs{U}_{i,j})$ is defined by $\frU |i\> = \sum_j \bs{U}_{i,j} |j\>$, so that the one-step evolution $|\Psi\> \to \frU
|\Psi\>$ reads as $\psi \to \psi \bs{U}$ using the coordinates $\psi=(\psi_0,\psi_1,\dots)$ of $|\Psi\>=\sum_i\psi_i|i\>$.

It was shown in \cite{CGMV} that the order
\begin{equation} \label{E:order}
\renewcommand{\arraystretch}{1.5}
\begin{tabular}{|c|c|}
\hline
$\Z_+$ & $|0\!\ua\>,|0\!\da\>,|1\!\ua\>,|1\!\da\>,|2\!\ua\>,|2\!\da\>,\dots$
\\ \hline
$\Z$ & $|0\!\ua\>,|\!-\!1\!\da\>,|\!-\!1\!\ua\>,|0\!\da\>,|1\!\ua\>,|\!-\!2\!\da\>,|\!-\!2\!\ua\>,|1\kern-2pt\da\>,\dots$
\\ \hline
\end{tabular}
\end{equation}
gives a transition matrix $\bs{U} = \bs{\Lambda} \bs{\cC} \bs{\Lambda}^\dag$, with $\bs{\Lambda}=\text{diag}(\bs{\lambda}_0,\bs{\lambda}_1,\dots)$
diagonal unitary and
\begin{equation} \label{E:CMV}
\begin{gathered}
\bs{\mathcal{C}} = \text{\small$\begin{pmatrix}
\bs{\alpha}_0^\dag & \bs{0} & \bs{\rho}_0^L \\
\bs{\rho}_0^R & \bs{0} & -\bs{\alpha}_0 & \bs{0} \\
\bs{0} & \bs{\alpha}_2^\dag & \bs{0} & \bs{0} & \bs{\rho}_2^L \\
& \bs{\rho}_2^R & \bs{0} & \bs{0} & -\bs{\alpha}_2 & \bs{0} \\
& & \bs{0} & \bs{\alpha}_4^\dag & \bs{0} & \bs{0} & \bs{\rho}_4^L \\
& & & \bs{\rho}_4^R & \bs{0} & \bs{0} & -\bs{\alpha}_4 & \bs{0} \\
& & & & \ddots & \ddots & \ddots & \ddots & \ddots
\end{pmatrix}$},
\\
\|\bs{\alpha}_k\|<1,
\qquad
\bs{\rho}_k^L=(1-\bs{\alpha}_k^\dag\bs{\alpha}_k)^{1/2},
\qquad
\bs{\rho}_k^R=(1-\bs{\alpha}_k\bs{\alpha}_k^\dag)^{1/2},
\end{gathered}
\end{equation}
$\bs{\lambda}_k$, $\bs{\alpha}_k$ and $\bs{\rho}_k^{L,R}$ being scalars for $\Z_+$ and $2\times2$-matrices for $\Z$.

If we refer at the same time to both, scalar and matrix objects, we will use the boldface notation. However, whenever we wish to distinguish between
scalars and matrices we will reserve for the last ones the boldface notation, so in such a case we will denote by $\lambda_k$, $\alpha_k$ and
$\rho_k=\rho_k^R=\rho_k^L$ the scalars of a QW on $\Z_+$.

Denoting by $e^{i\sigma_j^{(k)}}$ the phase of $c_{jj}^{(k)}$, the coefficients $\bs{\lambda}_k$, $\bs{\alpha}_k$ and $\bs{\rho}_k$ are obtained from
the coins by means of
\begin{center}
\renewcommand{\arraystretch}{1}
{\Small
\begin{tabular}{|c|c|c|c|}
\hline
$\Z_+$
& $\begin{aligned}
   & \lambda_{-1}=\lambda_0=1 \\ & \lambda_{2k+1} = e^{i\sigma_2^{(k)}} \lambda_{2k-1} \\ & \lambda_{2k+2} = e^{-i\sigma_1^{(k)}} \lambda_{2k}
  \end{aligned}$
& $\ds \alpha_{2k} = \overline{c}_{21}^{(k)} \frac{\lambda_{2k}}{\lambda_{2k-1}}$
& $\rho_{2k} = \sqrt{1-|\alpha_{2k}|^2}$
\\ \hline
$\Z$
& $\begin{aligned}
   & \bs{\lambda}_{2k-1} = \begin{pmatrix} \lambda_{-2k} & \kern-7pt 0 \\ 0 & \kern-7pt \lambda_{2k-1} \end{pmatrix}
   \\
   & \bs{\lambda}_{2k} = \begin{pmatrix} \lambda_{2k} & \kern-7pt 0 \\ 0 & \kern-7pt \lambda_{-2k-1} \end{pmatrix}
  \end{aligned}$
& $\bs{\alpha}_{2k} = \begin{pmatrix} 0 & \kern-11pt -\overline{\alpha}_{-2k-2} \\ \alpha_{2k} & \kern-11pt 0 \end{pmatrix}$
& $\begin{aligned}
   & \bs{\rho}_{2k}^R = \begin{pmatrix} \rho_{-2k-2} & \kern-7pt 0 \\ 0 & \kern-7pt \rho_{2k} \end{pmatrix}
   \\
   & \bs{\rho}_{2k}^L = \begin{pmatrix} \rho_{2k} & \kern-7pt 0 \\ 0 & \kern-7pt \rho_{-2k-2} \end{pmatrix}
  \end{aligned}$
\\ \hline
\end{tabular}
}
\end{center}

The matrix $\bs{\cC}$ is a special case of a kind of matrices that play a central role in the theory of orthogonal polynomials (OP) on the unit
circle $\T=\{z\in\C : |z|=1\}$, the so called CMV matrices
$$
\text{\Small$\begin{pmatrix}
\bs{\alpha}_0^\dagger & \kern-3pt \bs{\rho}_0^L\bs{\alpha}_1^\dagger & \kern-3pt \bs{\rho}_0^L\bs{\rho}_1^L &
\kern-3pt \0 & \kern-3pt \0 &  \kern-3pt \0 & \kern-3pt \0 & \kern-3pt\dots
\\
\bs{\rho}_0^R & \kern-3pt -\bs{\alpha}_0\bs{\alpha}_1^\dagger & \kern-3pt -\bs{\alpha}_0\bs{\rho}_1^L &
\kern-3pt \0 & \kern-3pt \0 & \kern-3pt \0& \kern-3pt \0 & \kern-3pt\dots
\\
\0 & \kern-3pt \bs{\alpha}_2^\dagger\bs{\rho}_1^R & \kern-3pt -\bs{\alpha}_2^\dagger\bs{\alpha}_1 &
\kern-3pt \bs{\rho}_2^L\bs{\alpha}_3^\dagger & \kern-3pt \bs{\rho}_2^L\bs{\rho}_3^L & \kern-3pt \0 & \kern-3pt \0 & \kern-3pt\dots
\\
\0 & \kern-3pt \bs{\rho}_2^R\bs{\rho}_1^R & \kern-3pt -\bs{\rho}_2^R\bs{\alpha}_1 & \kern-3pt -\bs{\alpha}_2\bs{\alpha}_3^\dagger & \kern-3pt
-\bs{\alpha}_2\bs{\rho}_3^L & \kern-3pt \0 & \kern-3pt \0 & \kern-3pt\dots
\\
\0 & \kern-3pt \0 & \kern-3pt \0 & \kern-3pt \bs{\alpha}_4^\dagger\bs{\rho}_3^R & \kern-3pt -\bs{\alpha}_4^\dagger\bs{\alpha}_3 &
\kern-3pt \bs{\rho}_4^L\bs{\alpha}_5^\dagger & \kern-3pt \bs{\rho}_4^L \bs{\rho}_5^L & \kern-3pt\dots
\\
\0 & \kern-3pt \0& \kern-3pt \0 & \kern-3pt \bs{\rho}_4^R\bs{\rho}_3^R & \kern-3pt -\bs{\rho}_4^R\bs{\alpha}_3 &
\kern-3pt -\bs{\alpha}_4\bs{\alpha}_5^\dagger & \kern-3pt -\bs{\alpha}_4\bs{\rho}_5^L & \kern-3pt\dots
\\
\dots & \kern-3pt\dots & \kern-3pt\dots & \kern-3pt\dots & \kern-3pt\dots & \kern-3pt\dots & \kern-3pt\dots & \kern-3pt\dots
\end{pmatrix}$}.
$$
The case of a QW corresponds to $\bs{\alpha}_{2k+1}=\0$ so that $\bs{\rho}_{2k+1}^L=\bs{\rho}_{2k+1}^R=\1$ stands for the number 1 or the $2\times2$
identity matrix for $\Z_+$ and $\Z$ respectively.

The connection between the transition matrix $\bs{U}$ of a QW and the CMV matrices ensures that the Laurent polynomials $\bs{X}_k$ defined by
\begin{equation} \label{E:eig}
\bs{U}\bs{X}(z) = z \bs{X}(z), \quad \bs{X}=(\bs{X}_0,\bs{X}_1,\dots)^T, \quad \bs{X}_0(z)=\1,
\end{equation}
constitute a sequence of orthogonal Laurent polynomials (OLP) with respect to a probability measure $\bs{\mu}$ supported on the unit circle (see
\cite{CGMV}), i.e.,
$$
\int_\T \bs{X}_j(z) d\bs{\mu}(z) \bs{X}_k(z)^\dag = \1\delta_{j,k}.
$$
Such OLP and measures are scalar or $2\times2$-matrix valued for $\Z_+$ and $\Z$ respectively. The coefficients $\bs{\alpha}_k$ are known as the
Verblunsky or reflection coefficients of the measure $\bs{\mu}$.

The orthogonality and the ``eigenvalue" equation (\ref{E:eig}) yield a KMcG formula for the QW, i.e., an OLP representation of the $n$-step
transition amplitudes (see \cite[pages 479 and 483]{CGMV})
$$
(\bs{U}^n)_{j,k} = \int_\T z^n \bs{X}_j(z) d\bs{\mu}(z) \bs{X}_k(z)^\dag.
$$
Here $(\bullet)_{j,k}$ stands for the $(j,k)$-th element in $\Z_+$ and for the $(j,k)$-th $2\times2$-block in $\Z$. The KMcG formula is the
cornerstone of the CGMV method, which takes advantage of the OP techniques for the analysis of QWs.

Useful tools in the theory of OP on the unit circle are the so called Carathéodory and Schur functions related to $\bs{\mu}$, defined respectively by
$$
\bs{F}(z) = \int_\T \frac{t+z}{t-z} d\bs{\mu}(t), \quad \bs{f}(z) = z^{-1} (\bs{F}(z)-\1)(\bs{F}(z)+\1)^{-1}, \quad |z|<1.
$$
The Carathéodory and Schur functions of a probability measure on $\T$ can be characterized as the analytic functions on the unit disk
$\D=\{z\in\C:|z|<1\}$ such that $\bs{F}(0)=\1$, $\re\bs{F}(z)>0$ and $\|\bs{f}(z)\|<1$ for $z\in\D$ respectively, where $\re A =
\frac{1}{2}(A+A^\dag)$ and $\im A = \frac{1}{2i}(A-A^\dag)$ for any square matrix $A$. We will assume that $\bs{F}$ and $\bs{f}$ are extended to the
unit circle by their radial boundary values, which exist for almost every point of $\T$, so $\re\bs{F}(e^{i\theta})\geq0$ and
$\|\bs{f}(e^{i\theta})\|\leq1$ for a.e. $\theta$.

These analytic functions provide a shortcut between the measure $\bs{\mu}$ and the Verblunsky coefficients $\bs{\alpha}_k$ because both of them can
be recovered from $\bs{F}$ or $\bs{f}$. We will now point out the connections needed for the rest of the paper (see \cite{DaPuSi} for the general
matrix case).

If $d\bs{\mu}(e^{i\theta})=\bs{w}(\theta)\frac{d\theta}{2\pi}+d\bs{\mu}_s(e^{i\theta})$ is the Lebesgue decomposition of $\bs\mu$ into an absolutely
continuous and a singular part,
\begin{equation} \label{E:SC-measure}
\begin{aligned}
& \bs{w}(\theta) = \re\bs{F}(e^{i\theta}), \quad \text{ for a.e. } \theta,
\\
& \det\bs{w}(\theta) \neq 0 \Leftrightarrow \|\bs{f}(e^{i\theta})\| < 1, \quad \text{ for a.e. } \theta,
\\
& \supp\bs{\mu}_s \subset \{z\in\T : \ts \lim_{r\ua1}\text{tr}(\re\bs{F}(rz))=\infty\},
\\
& \bs{\mu}(\{z\}) = \ts \lim_{r\ua1} \frac{1-r}{2}\bs{F}(rz), \qquad z\in\T.
\end{aligned}
\end{equation}
In particular, if $z\in\T$ is a pole of an analytic extension of $\bs{F}$, such a pole must be of order one and $z$ is an isolated mass point of
$\bs{\mu}$ with mass $\bs{\mu}(\{z\})=-(2z)^{-1}\text{Res}(\bs{F};z)$.

These properties are well known in the case of scalar measures. Concerning matrix measures, the first two properties can be found in \cite{DaPuSi},
while the remaining ones can be reduced to the scalar case by noticing that $\bs{\mu}$ is absolutely continuous with respect to the scalar trace
measure tr$\bs{\mu}$.

On the other hand, starting at $\bs{f}_0=\bs{f}$, the Verblunsky coefficients $\bs{\alpha}_k=\bs{f}_k(0)$ can be recovered through the Schur
algorithm
$$
\begin{aligned}
\bs{f}_{k+1}(z) & = z^{-1} (\bs{\rho}_k^R)^{-1} (\bs{f}_k(z)-\bs{\alpha}_k)(\1-\bs{\alpha}_k^\dag\bs{f}_k(z))^{-1} \bs{\rho}_k^L
\\
& = z^{-1} \bs{\rho}_k^R \, (\1-\bs{f}_k(z)\bs{\alpha}_k^\dag)^{-1} (\bs{f}_k(z)-\bs{\alpha}_k) \, (\bs{\rho}_k^L)^{-1},
\end{aligned}
$$
which assigns to $\bs{f}$ a sequence of Schur functions $\bs{f}_k$. For this reason, the Verblunsky coefficients of $\bs{\mu}$ are also known as the
Schur parameters of $\bs{f}$. Obviously, the Schur parameters of each Schur iterate $\bs{f}_k$ are obtained deleting the first $k$ parameters of
$\bs{f}$: $\bs{\alpha}_k,\bs{\alpha}_{k+1},\dots$. The Schur algorithm can be inverted to give
$$
\begin{aligned}
\bs{f}_k(z) & = (\bs{\rho}_k^R)^{-1} (z\bs{f}_{k+1}(z)+\bs{\alpha}_k)(\1+\bs{\alpha}_k^\dag z\bs{f}_{k+1}(z))^{-1} \bs{\rho}_k^L
\\
& = \bs{\rho}_k^R \, (\1+z\bs{f}_{k+1}(z)\bs{\alpha}_k^\dag)^{-1} (z\bs{f}_{k+1}(z)+\bs{\alpha}_k) \, (\bs{\rho}_k^L)^{-1}.
\end{aligned}
$$
In the scalar case the factors $\bs{\rho}_k^{L,R}$ cancel each other in the above formulas due to the commutativity, giving
\begin{equation} \label{E:SchurAlg-scalar}
f_{k+1}(z) = \frac{1}{z} \frac{f_k(z)-\alpha_k}{1-\overline{\alpha}_kf_k(z)},
\qquad
f_k(z) = \frac{zf_{k+1}(z)+\alpha_k}{1+\overline{\alpha}_kzf_{k+1}(z)}.
\end{equation}

Both, the measure and the Schur function, are univocally determined by the Verblunsky coefficients. Some results relating Schur functions and Schur
parameters will be of interest for us.

\begin{prop} \label{P:Schur-1}
If $\bsalpha_0,\bsalpha_2,\bsalpha_4,\dots$ are the Schur parameters of the Schur function $\bs{f}(z)$, then
$\bsalpha_0,\0,\bsalpha_2,\0,\bsalpha_4,\0,\dots$ are the Schur parameters of $\bs{f}(z^2)$.
\end{prop}

\begin{proof}
Denote $\bs{g}_{2k-1}(z)=z\bs{f}_k(z^2)$ and $\bs{g}_{2k}(z)=\bs{f}_k(z^2)$. The substitution $z \to z^2$ in the Schur algorithm for $\bs{f}(z)$,
together with the relation $\bs{g}_{2k}(z)=z^{-1}\bs{g}_{2k-1}(z)$, gives the Schur algorithm for $\bs{g}(z)=\bs{f}(z^2)$.
\end{proof}

\begin{prop} \label{P:Schur-2}
If a Schur function $\bs{f}$ has Schur parameters $\bs{\alpha}_k$, then for any unitary matrices $V_1,V_2$ the Schur function $V_1\bs{f}V_2$ has
Schur parameters $V_1\bs{\alpha}_kV_2$.
\end{prop}

\begin{proof}
Simply check that the Schur algorithm is invariant under the transformation $\bs{f} \to V_1\bs{f}V_2$, $\bs{\alpha}_k \to V_1\bs{\alpha}_kV_2$,
bearing in mind that it maps $\bs{\rho}_k^L \to V_2^\dag\bs{\rho}_k^LV_2$ and $\bs{\rho}_k^R \to V_1\bs{\rho}_k^RV_1^\dag$.
\end{proof}

Proposition \ref{P:Schur-1} states that the Schur functions whose odd Schur parameters vanish are the even Schur functions, i.e., those satisfying
$\bs{f}(-z)=\bs{f}(z)$. In terms of the Carathéodory function this condition reads as $\bs{F}(-z)\bs{F}(z)=\1$, as follows from the inverse relation
\begin{equation} \label{E:Ff}
\bs{F}(z)=(1+z\bs{f}(z))(1-z\bs{f}(z))^{-1}.
\end{equation}

Hence, the Schur function of any QW on $\Z$ or $\Z_+$ must be even. On the other hand, Proposition \ref{P:Schur-2} has the following consequences of
interest for QWs on $\Z$.

\begin{prop} \label{P:antidiag}
Given a sequence of $2\times2$ Schur parameters
$$
\bs{\alpha}_k = {\small \begin{pmatrix} 0 & \alpha_k^- \\ \alpha_k^+ & 0 \end{pmatrix}},
$$
the corresponding Schur and Carathéodory functions are
$$
\begin{aligned}
\bs{f}(z) & = {\small \begin{pmatrix} 0 & f_-(z) \\ f_+(z) & 0 \end{pmatrix}},
\\
\bs{F}(z) & = \frac{1}{1-g(z)} {\small \begin{pmatrix} 1+g(z) & 2zf_-(z) \\ 2zf_+(z) & 1+g(z) \end{pmatrix}}, \qquad g(z) = z^2f_+(z)f_-(z),
\end{aligned}
$$
where $f_\pm$ is the Schur function with Schur parameters $\alpha_k^\pm$.

If $d\bs{\mu}(e^{i\theta}) = \bs{w}(\theta)\frac{d\theta}{2\pi} + d\bs{\mu}_s(e^{i\theta})$ is the Lebesgue decomposition of the related measure,
$$
\det\bs{w}(e^{i\theta})\neq0 \Leftrightarrow |f_\pm(e^{i\theta})|<1, \quad \text{ for a.e. } \theta,
$$
and the singular part $\bs{\mu}_s$ is supported on the roots $z\in\T$ of $g(z)=1$. The mass points are those roots such that
$$
m(z) = \lim_{r\ua1}\frac{1-r}{1-g(rz)} \neq 0,
$$
and the corresponding mass is the singular matrix
$$
\bs{\mu}(\{z\}) = m(z) \begin{pmatrix} 1 & \eta(z) \\ \overline{\eta(z)} & 1 \end{pmatrix}, \qquad \eta(z) = zf_-(z) \in \T.
$$
In particular, if $g$ extends analytically to a neighbourhood of a root $z\in\T$ of $g(z)=1$, then $z$ is a simple isolated root\footnote{If the root
$z$ were not isolated, general principles would imply that $g=1$ on $\D$, which is not possible since $g$ must be a Schur function. Similar comments
for the roots of $h=1$ in the paragraph below Corollary \ref{C:sym}.} and also an isolated mass point with $m(z)=1/zg'(z)$.
\end{prop}

\begin{proof}
The result for $\bs{f}$ is a direct consequence of Proposition \ref{P:Schur-2} and
$$
\bs{f} = V \begin{pmatrix}f_+&0\\0&f_-\end{pmatrix},
\qquad
\bs{\alpha}_k = V \begin{pmatrix}\alpha_k^+&0\\0&\alpha_k^-\end{pmatrix},
\qquad
V = \begin{pmatrix}0&1\\1&0\end{pmatrix}.
$$
Then, the expression of $\bs{F}$ follows from (\ref{E:Ff}).

The rest of the results are obtained from (\ref{E:SC-measure}) by using the actual form of the Schur and Carathéodory functions. Concerning the
singularity of the masses, simply take into account that the roots $z\in\T$ of $g(z)=1$ must satisfy $|f_\pm(z)|=1$ because $|f_\pm|\leq1$ in $\T$.
Hence, the non diagonal elements of the mass must be proportional to $\eta(z)=zf_-(z)\in\T$ and $zf_+(z)=(zf_-(z))^{-1}=\overline{\eta(z)}$.
\end{proof}

Proposition \ref{P:antidiag} holds for QWs on $\Z$ with $\alpha_{2k}^+=\alpha_{2k}$, $\alpha_{2k}^-=-\overline\alpha_{-2k-2}$ and
$\alpha_{2k+1}^+=\alpha_{2k+1}^-=0$, where $\alpha_{2k}=\overline{c}_{21}^{(k)}\lambda_{2k}/\lambda_{2k-1}$. Then, $f_+$ is the Schur function
associated with the Schur parameters with non negative indices, while $f_-$ corresponds to the Schur parameters with negative indices. Furthermore,
$f_\pm$ are even functions, so $g$ is even too, and this has the following consequence.

\begin{cor} \label{C:sym}
For any QW on $\Z$, the mass points of the corresponding measure appear in pairs $\pm z$ which are symmetric with respect to the origin, and the mass
is given by Proposition \ref{P:antidiag} with $m(-z)=m(z)$ and $\eta(-z)=-\eta(z)$.
\end{cor}

The mass points of a QW on $\Z_+$ have no such a symmetry, despite the fact that the corresponding Schur function $f$ is even too. The reason is that
the Carathéodory function is given by
$$
F(z) = \frac{1+h(z)}{1-h(z)}, \qquad h(z)=zf(z),
$$
so the singular part of the measure is supported on the roots $z\in\T$ of the equation $h(z)=1$, which is not invariant under $z \to -z$ because $h$
is odd. The mass points are those roots such that
$$
\mu(\{z\}) = \lim_{r\ua1} \frac{1-r}{1-h(rz)} \neq 0.
$$
When $h$ has an analytic extension to a neighbourhood of a root $z\in\T$ of $h(z)=1$, such a root is simple and isolated, and is also an isolated
mass point with mass $\mu(\{z\})=1/zh'(z)$.

\section{Single state localization in QWs} \label{S:Loc}

Following \cite{K}, we will adopt the definition below for the localization of a state in a QW. It applies only to qubit states
$\alpha|k\!\ua\>+\beta|k\!\da\>$ at a given site $k$, and characterizes those states which have a non null probability of asymptotic return to the
same site where they are placed originally.

\begin{defn}
Given a QW on $\Z$ or $\Z_+$, let $p_{\alpha,\beta}^{(k)}(n)$ be the probability that the walker returns to the site $k$ in $n$ steps having started
at the qubit state $|\Psi_{\alpha,\beta}^{(k)}\> = \alpha|k\!\ua\>+\beta|k\!\da\>$ at the initial time. We will say that such a state exhibits localization
if $\limsup_{n\to\infty}p_{\alpha,\beta}^{(k)}(n)\neq0$.
\end{defn}

It is known that the structure of the transition matrix for the QWs on $\Z$ we are discussing always gives a null return probability for an odd number
of steps, i.e., $p_{\alpha,\beta}^{(k)}(2n-1)=0$. Therefore, the only quantity of interest in the case $\Z$ is the asymptotics of
$p_{\alpha,\beta}^{(k)}(2n)$.


If $\frU$ is the transition operator of the QW,
$$
p_{\alpha,\beta}^{(k)}(n) =
|\< \Psi_{1,0}^{(k)} | \frU^n | \Psi_{\alpha,\beta}^{(k)} \>|^2 + |\< \Psi_{0,1}^{(k)} | \frU^n | \Psi_{\alpha,\beta}^{(k)} \>|^2.
$$
The KMcG formula provides an alternative expression for this probability which is nicely adapted to study its asymptotics. Indeed, a simple extension
of the formula in \cite[page 497]{CGMV} to the case of two arbitrary states $| \Psi \>$, $| \tilde\Psi \>$ gives
\begin{equation} \label{E:KMcG-gen}
\< \tilde\Psi | \frU^n | \Psi \> = \psi \bs{U}^n \tilde\psi^\dag = \int_\T z^n \bs{\psi}(z) d\bs{\mu}(z) \tilde{\bs{\psi}}(z)^\dag,
\end{equation}
where $\bs{\psi}(z)$ is an $L^2_{\bs{\mu}}(\T)$ function associated with the state $|\Psi\>=\sum_i\psi_i|i\>$  ($|i\>$ is the $i$-th vector of the
ordered basis), which is a scalar function for $\Z_+$ and a 2-vector function for $\Z$. The general form of $\bs{\psi}(z)$ is given in the first
column of the following table, while the second column shows the particular case $\bs{\psi}_{\alpha,\beta}^{(k)}(z)$ for the qubit state
$|\Psi_{\alpha,\beta}^{(k)}\>$.
\begin{equation} \label{E:psi}
\renewcommand{\arraystretch}{2}
\kern-7pt {\Small
\begin{tabular}{|c|c|c|}
\hline
& $\bs{\psi}(z)$ & $\bs{\psi}^{(k)}_{\alpha,\beta}(z)$
\\ \hline
$\Z_+$ & $\ds \sum_k \psi_k X_k$ & $\alpha X_{2k} + \beta X_{2k+1}$
\\ \hline
$\Z$ & $\ds \sum_k (\psi_{2k},\psi_{2k+1}) \bs{X}_k$
& $\begin{cases} (\alpha,0)\bs{X}_{2j}+(0,\beta)\bs{X}_{2j+1} & k=j \\ (0,\beta)\bs{X}_{2j}+(\alpha,0)\bs{X}_{2j+1} & k=-j-1 \end{cases}
\quad j\geq0$
\\ \hline
\end{tabular}
}
\end{equation}
Relation (\ref{E:KMcG-gen}) gives the identity
$$
p_{\alpha,\beta}^{(k)}(n) = {\small \left| \int_\T z^n \bs{\psi}_{\alpha,\beta}^{(k)}(z) d\bs{\mu}(z) \bs{\psi}_{1,0}^{(k)}(z)^\dag \right|^2 +
\left| \int_\T z^n \bs{\psi}_{\alpha,\beta}^{(k)}(z) d\bs{\mu}(z) \bs{\psi}_{0,1}^{(k)}(z)^\dag \right|^2}.
$$

The importance of the KMcG formula in the study of the localization of the states in a QW was first pointed out by N. Konno et al in \cite{KS}.
There the authors use the Riemann-Lebesgue lemma to obtain the asymptotics of $p_{\alpha,\beta}^{(k)}(n)$ for the case of a constant coin on $\Z_+$.
The method can handle other QWs on $\Z_+$, as well as QWs on $\Z$.

For convenience, in what follows we will use the notation
$$
a_n\underset{n}{\sim}b_n \Leftrightarrow \lim_{n\to\infty}(a_n-b_n)=0.
$$

\begin{lem} \label{L:R-L}
If $\frU$ is the transition operator of a QW on $\Z$ or $\Z_+$ with measure $\bs{\mu}$,
$$
\< \tilde\Psi | \frU^n | \Psi \> \underset{n}{\sim} \int_\T z^n \bs{\psi}(z) d\bs{\mu}_s(z) \tilde{\bs{\psi}}(z)^\dag,
$$
where $\bs{\mu}_s$ is the singular part of $\bs{\mu}$.
\end{lem}

\begin{proof}
Let $d\bs{\mu}(e^{i\theta}) = \bs{w}(\theta)\frac{d\theta}{2\pi} + d\bs{\mu}_s(e^{i\theta})$ be the Lebesgue decomposition of the measure $\bs{\mu}$.
The Riemann-Lebesgue lemma implies that
$$
\lim_{n\to\infty} \int_0^{2\pi} e^{in\theta} \bs{\psi}(e^{i\theta}) \bs{w}(\theta) \tilde{\bs{\psi}}(e^{i\theta})^\dag \frac{d\theta}{2\pi} = 0
$$
because $\bs{\psi}(e^{i\theta}) \bs{w}(\theta) \tilde{\bs{\psi}}(e^{i\theta})^\dag$ is integrable with respect to the Lebesgue measure. This gives
the result, bearing in mind the KMcG formula.
\end{proof}

As a consequence, no state can exhibit localization in a QW with an absolutely continuous measure. As for the singular part, it always can be
decomposed into mass points and a singular continuous part. As we will see, due to Wiener's theorem, the presence of mass points will always give
localized states, regardless of the presence of a singular continuous part. However, if the singular part is exclusively continuous the situation is
more involved because the Riemann-Lebesgue lemma holds for some singular continuous measures, but not for all of them.

To obtain the strongest results about localization for QWs on $\Z$ one is greatly aided by using the freedom in renumbering the sites $k \to k+k_0$,
$k_0\in\Z$. The consequence of this freedom is that, for any QW on $\Z$, there are infinitely many orders of the basis giving a CMV-shape transition
matrix. Just as good as the initial order would be to take
$$
|k_0\!\ua\>,|k_0-1\!\da\>,|k_0-1\!\ua\>,|k_0\!\da\>,|k_0+1\!\ua\>,|k_0-2\!\da\>,|k_0-2\!\ua\>,|k_0+1\!\da\>,\dots
$$
where $k_0$ is an arbitrary integer. The order of the basis given originally in (\ref{E:order}) for QWs on $\Z$ can be understood as a folding of
$\Z$ at site 0, so these other possibilities correspond to foldings at an arbitrary site $k_0$.

These new foldings lead to different CMV matrices, measures and OLP, any of them could be used to study a QW on $\Z$. Since the presence of
localization in a QW has to do with the Lebesgue decomposition of the measure, it is important to know how the measure changes with the renumbering
of the sites. This is answered by the following result.

\begin{lem} \label{L:folding}
Given a QW on $\Z$, the measures $\bs{\mu}$, $\tilde{\bs{\mu}}$ corresponding to different foldings are related by
$$
d\tilde{\bs{\mu}}(z) = \bs{A}(z) d\bs{\mu}(z) \bs{A}(z)^\dag,
$$
where $\bs{A}$ is a $2\times2$-matrix polynomial.
\end{lem}

\begin{proof}
The transition matrices $\bs{U}$, $\tilde{\bs{U}}$ related to different foldings are representations of the same transition operator with respect to
basis which only differ in the order. Thus they are related by conjugation with a permutation matrix $\Pi$, i.e., $\tilde{\bs{U}} = \Pi^\dag \bs{U}
\Pi$.

On the other hand, the KMcG formula ensures that
$$
((\bs{U}+z\1)(\bs{U}-z\1)^{-1})_{j,k} = \int_\T \frac{t+z}{t-z} \, \bs{X}_j(t) d\bs{\mu}(t) \bs{X}_k(t)^\dag,
$$
where $(\bullet)_{j,k}$ stands for the $(j,k)$-th $2\times2$-block. Thus, the Carathéodory function of the measure $\bs{\mu}$ is given by
$$
\bs{F}(z) = \int_\T \frac{t+z}{t-z}\,d\bs{\mu}(t) = ((\bs{U}+z\1)(\bs{U}-z\1)^{-1})_{0,0},
$$
and similarly for the Carathéodory function $\tilde{\bs{F}}$ of $\tilde{\bs{\mu}}$.

Each subindex $k$ stands for a pair of indices which we will denote by $k_s$, $s=+,-$. If $\Pi$ transforms the indices $0_+$ and $0_-$ into $j_r$ and
$k_s$ respectively, then $\Pi_{i,0_+}=\delta_{i,j_r}$, $\Pi_{i,0_-}=\delta_{i,k_s}$, and
$$
\begin{aligned}
\tilde{\bs{F}}(z) & = (\Pi^\dag(\bs{U}+z\1)(\bs{U}-z\1)^{-1}\Pi)_{0,0}
\\
& = \begin{pmatrix}
((\bs{U}+z\1)(\bs{U}-z\1)^{-1})_{j_r,j_r} & ((\bs{U}+z\1)(\bs{U}-z\1)^{-1})_{j_r,k_s}
\\
((\bs{U}+z\1)(\bs{U}-z\1)^{-1})_{k_s,j_r} & ((\bs{U}+z\1)(\bs{U}-z\1)^{-1})_{k_s,k_s}
\end{pmatrix} =
\\
& = \int_\T \frac{t+z}{t-z} \,
\begin{pmatrix} \bs{X}_j^r(t) \\ \bs{X}_k^s(t) \end{pmatrix}
d\bs{\mu}(t)
\begin{pmatrix} \bs{X}_j^r(t)^\dag & \bs{X}_k^s(t)^\dag \end{pmatrix},
\end{aligned}
$$
where $\bs{X}_k^+$ and $\bs{X}_k^-$ stand for the upper and lower row of $\bs{X}_k$ respectively. This proves the proposition with
$$
\bs{A}(z) = z^l \begin{pmatrix} \bs{X}_j^r(z) \\ \bs{X}_k^s(z) \end{pmatrix}, \quad \text{ for some } l\geq0,
$$
since $\bs{X}_j^r$ and $\bs{X}_k^s$ are 2-vector Laurent polynomials.
\end{proof}

We are interested in the following consequence of the lemma above.

\begin{cor} \label{C:folding}
Two measures of the same QW on $\Z$ with respect to different foldings have the same mass points, and the support of their absolutely continuous and
singular continuous parts coincide.
\end{cor}

\begin{proof}
Let $\bs{\mu}$, $\tilde{\bs{\mu}}$ be such measures. Then, $d\tilde{\bs{\mu}} = \bs{A} d\bs{\mu} \bs{A}^\dag$ for some matrix polynomial $\bs{A}$,
which implies that $\tilde{\bs{\mu}}(\{z\})=0$ whenever $\bs{\mu}(\{z\})=0$. Since there must be another polynomial $\tilde{\bs{A}}$ such that
$d\bs{\mu} = \tilde{\bs{A}} d\tilde{\bs{\mu}} \tilde{\bs{A}}^\dag$, we conclude that the mass points of $\bs{\mu}$ and $\tilde{\bs{\mu}}$ coincide.
The rest of the assertions follow similarly from Proposition \ref{L:folding} and the invariance of the absolutely continuous and singular character
of a matrix measure on $\T$ under the transformation $d{\bs\nu} \to \bs{A} d\bs{\nu} \bs{A}^\dag$ for any matrix polynomial $\bs{A}$.
\end{proof}

Unless we state specifically a different folding, the measure of a QW on $\Z$ means for us that one related to the folding at site 0 given in
(\ref{E:order}). Nevertheless, the previous corollary ensures that we can refer to some characteristics of the measure without indicating any folding
because they are common for all of them.

Besides exploiting different foldings, the general results about localization for QWs on $\Z$ also require the use of Wiener's theorem on the unit
circle (see \cite[Theorem 12.4.7]{Si2}): for any scalar measure $\mu$ on $\T$
$$
\lim_{N\to\infty} \frac{1}{2N+1} \sum_{n=-N}^N |\mu_n|^2 = \sum_{z\in\T} |\mu(\{z\})|^2, \qquad \mu_k = \int_\T z^k d\mu(z).
$$
Thus, $\mu$ has no mass points if and only if $\lim_{N\to\infty} \frac{1}{2N+1} \sum_{n=-N}^N |\mu_n|^2 = 0$, which is satisfied in particular when
$\lim_{n\to\infty}\mu_n=0$. The complex numbers $\mu_n$ are known as the moments of the measure $\mu$.

The following theorem is the main result of this section. It gives an interpretation of the single state localization in terms of the measure of the
QW.

\begin{thm} \label{T:Loc}

Let a QW on $\Z$ or $\Z_+$ with transition matrix $\bs{U}$ and measure $\bs{\mu}$.

\begin{itemize}

\item[(a)] If $\bs{\mu}$ is absolutely continuous, no state exhibits localization.

\item[(b)] If $\bs{\mu}$ has a mass point, all the states $|\Psi_{\alpha,\beta}^{(k)}\>$ exhibit localization except at most one state at each site
$k$ which must have $\alpha,\beta\neq0$. The existence of such a non localized state is mandatory when the singular part of the measure is a single
mass point in $\Z_+$ or a single pair of opposite mass points in $\Z$.

\item[(c)] The states $|\Psi\>$ which do not exhibit localization must satisfy
\begin{equation} \label{E:NoLoc}
\bs{\psi}(z) \bs{\mu}(\{z\}) = 0, \quad  \forall z\in\T.
\end{equation}

\item[(d)] If $\bs{\mu}$ has no singular continuous part:

\begin{itemize}

\item[(i)] No state exhibits localization $\Leftrightarrow$ $\bs{\mu}$ has no mass points
$\Leftrightarrow$
\newline
$\Leftrightarrow \ds\lim_{n\to\infty}\psi\bs{U}^n\tilde{\psi}^\dag=0$ for all $\psi,\tilde{\psi} \in L^2(\Z)$.

\item[(ii)] $|\Psi\>$ does not exhibit localization
$\Leftrightarrow (\ref{E:NoLoc}) \Leftrightarrow \ds\lim_{n\to\infty} \psi \bs{U}^n \psi^\dag = 0$.

\end{itemize}

\end{itemize}

\end{thm}

\begin{proof}
Statement (a) follows directly from Lemma \ref{L:R-L}.

$|\Psi\>=|\Psi_{\alpha,\beta}^{(k)}\>$ does not exhibit localization if and only if
$\lim_{n\to\infty}\psi\bs{U}^n(\psi_{1,0}^{(k)})^\dag=\lim_{n\to\infty}\psi\bs{U}^n(\psi_{0,1}^{(k)})^\dag=0$, which obviously implies that
$\lim_{n\to\infty}\psi\bs{U}^n\psi^\dag=0$, i.e,
$$
\lim_{n\to\infty} \int_\T z^n \bs{\psi}(z) d\bs{\mu}(z) \bs{\psi}(z)^\dag = 0.
$$
In other words, the $n$-th moment of the scalar measure $d\mu_\psi = \bs{\psi} d\bs{\mu} \bs{\psi}^\dag$ converges to zero as $n\to\infty$. Wiener's
theorem ensures that $\mu_\psi$ has no mass points, which means that $\bs{\psi}(z) \bs{\mu}(\{z\}) \bs{\psi}(z)^\dag = 0$ for any $z \in \T$. Bearing
in mind that $\bs{\mu}(\{z\})$ is positive semidefinite, we get (c).

According to (\ref{E:psi}), given a QW on $\Z_+$, condition (\ref{E:NoLoc}) becomes $\bs{\psi}(z) = \alpha X_{2k}(z) + \beta X_{2k+1}(z) = 0$ for any
mass point $z$ of $\mu$. Since the OLP have no zeros on $\T$, when $\mu$ has a mass point this equation has a one-dimensional subspace of solutions
$(\alpha,\beta)$, $\alpha,\beta\neq0$, which represent the same quantum state. The existence of more than one mass point or a singular continuous
part of the measure can give incompatible equations for $\alpha,\beta$, so the presence of states which do not exhibit localization is only ensured
in the case of at most one mass point. This proves (b) for $\Z_+$.

For a QW on $\Z$, (\ref{E:psi}) implies that the no localization condition (\ref{E:NoLoc}) at site $k=0$ reads as $\bs{\psi}(z)\bs{\mu}(\{z\}) =
(\alpha\bs{X}_0^++\beta\bs{X}_1^-) \bs{\mu}(\{z\}) = 0$ for any mass point $z$ of $\bs{\mu}$, with $\bs{X}_k^\pm$ the upper and lower rows of
$\bs{X}_k$. We know that $\bs{X}_0=\1$, while $\bs{X}_1$ can be calculated from the first two $2\times2$-block equations of (\ref{E:eig}),
$$
\begin{aligned}
& \begin{pmatrix} -z & c_{21}^{(0)} \\ c_{12}^{(-1)} & -z \end{pmatrix} \bs{X}_0(z) +
\begin{pmatrix} c_{11}^{(0)} & 0 \\ 0 & c_{22}^{(-1)} \end{pmatrix} \bs{X}_2(z) = 0,
\\
& \begin{pmatrix} c_{11}^{(-1)} & 0 \\ 0 & c_{22}^{(0)} \end{pmatrix} \bs{X}_0(z) - z \bs{X}_1(z) +
\begin{pmatrix} 0 & c_{21}^{(-1)} \\ c_{12}^{(0)} & 0 \end{pmatrix} \bs{X}_2(z) = 0,
\end{aligned}
$$
giving
\begin{equation} \label{E:X1-Z}
\bs{X}_1(z) = \begin{pmatrix}
z^{-1}(\det C_{-1})/c_{22}^{(-1)} & c_{21}^{(-1)}/c_{22}^{(-1)}
\\
c_{12}^{(0)}/c_{11}^{(0)} & z^{-1}(\det C_0)/c_{11}^{(0)}
\end{pmatrix}.
\end{equation}

On the other hand, the mass of any mass point $z$ is the singular matrix given in Proposition \ref{P:antidiag}. Combining these results we find that
$$
\bs{\psi}(z)\bs{\mu}(\{z\})=0 \kern7pt \Leftrightarrow \kern7pt \alpha+\frac{\beta}{c_{11}^{(0)}}(c_{12}^{(0)}+\overline{z\eta(z)}\det C_0)=0.
$$
The coins are unitary, so $|\det C_0|=1$. Besides, the assumption of the irreducibility for the QW implies that $c_{jj}^{(0)}\neq0$, so
$|c_{12}^{(0)}|^2=1-|c_{jj}^{(0)}|^2 < 1$. Since $|z\eta(z)|=1$ for any mass point $z$, the above equation becomes $\alpha+\beta\kappa(z)=0$ with
$\kappa(z)\neq0$. This equation is invariant under the reflection $z \to -z$ due to Corollary \ref{C:sym}. Therefore, if there is a single pair of
opposite mass points, such an equation has a one-dimensional subspace of solutions $(\alpha,\beta)$, $\alpha,\beta\neq0$, which represent the same
quantum state. The presence at site $k=0$ of non localized states is ensured only when the singular part consists in at most a single pair of mass
points, otherwise incompatibilities can appear between different equations for $(\alpha,\beta)$.

To generalize this results for $k\neq0$ is enough to use the folding at site $k$. Corollary \ref{C:folding} states that the mass points will not
change when choosing this new folding. Thus, the previous discussion remains unchanged, but the conclusions are now about the states at the site $k$,
which plays the role of the origin with this new folding. Therefore (b) is proved for $\Z$ too.

Let us return to the general case of a QW on $\Z$ or $\Z_+$, and assume that the measure has no singular continuous part, which will be true no
matter which folding we choose in $\Z$, again due to Corollary \ref{C:folding}. Then (a) and (b) are the only options. In case (a), Lemma \ref{L:R-L}
states that $\lim_{n\to\infty}\psi\bs{U}^n\tilde{\psi}^\dag=0$ for all $\psi,\tilde{\psi}$. Also, as we pointed out at the beginning of the proof,
the condition $\lim_{n\to\infty} \psi \bs{U}^n \psi^\dag = 0$ is not only a consequence of the fact that $|\Psi\>$ does not exhibit localization, but
also yields (\ref{E:NoLoc}). On the other hand, (\ref{E:NoLoc}) gives $\int_\T\bs{\psi}(z)d\bs{\mu}_s(z)\tilde{\bs{\psi}}(z)^\dag=0$ for any
$\tilde{\bs{\psi}}$ because the mass points constitute all the singular part of the measure, thus implying that $|\Psi\>$ does not exhibit
localization. This finishes the proof of (d).
\end{proof}

The previous theorem indicates that the more mass points the measure exhibits, the less possibilities for non localized states because, apart from
the conditions associated with the singular continuous part of the measure, there are as many no localization equations (\ref{E:NoLoc}) as mass
points, which makes it more difficult to have non localized states as the number of mass points increases.

The proof of Theorem \ref{T:Loc} shows that, in $\Z$, the no localization equation (\ref{E:NoLoc}) is invariant under the reflection $z \to -z$,
hence, only one of such equations must be taken into account for each pair of opposite mass points.

There are situations in which it is known that the singular part of the measure is not purely continuous (or even no singular continuous part
appears). Theorem \ref{T:Loc} provides in such a case a {\bf localization dichotomy}: either no mass points and no localized states exist, or there
are mass points and ``almost" any state (at most one exception per site) exhibits localization. When this dichotomy works we can talk about QWs with
or without localization because then localization becomes an ``almost" global property.

Moreover, Lemma \ref{L:R-L} shows that in the absence of a singular continuous part of the measure the asymptotic return probability can be computed
exactly through
\begin{equation} \label{E:ARP}
\kern-5pt p_{\alpha,\beta}^{(k)}(n) \underset{n}{\sim} \text{\SMALL
$\left|\sum_{z\in\T} z^n\bs{\psi}_{\alpha,\beta}^{(k)}(z)\bs{\mu}(\{z\})\bs{\psi}_{1,0}^{(k)}(z)^\dag\right|^2 +
\left|\sum_{z\in\T} z^n\bs{\psi}_{\alpha,\beta}^{(k)}(z)\bs{\mu}(\{z\})\bs{\psi}_{0,1}^{(k)}(z)^\dag\right|^2$},
\end{equation}
where the sums are in fact over the mass points $z$ of $\bs{\mu}$.

\subsection{Periodic QWs with finite defects} \label{SS:periodic}

Among the QWs where the localization dichotomy works are those with periodic coins, with or without a finite number of defects.

\begin{prop} \label{P:periodic-defects}

If the coins $C_k$ of a QW on $\Z$ or $\Z_+$ satisfy $C_{k+p}=C_k$, $p\in\N$, for all but a finite number of sites $k$, the corresponding measure has
no singular continuous part and thus the localization dichotomy holds.

\end{prop}

\begin{proof}
Consider first the case of strictly periodic coins on $\Z_+$ with period $p$, i.e., $C_{k+p}=C_k$ for all $k\in\Z_+$. The related measure $\mu$ has
Verblunsky coefficients
$$
\alpha_0=\overline{c}_{21}^{(0)}, \qquad \alpha_{2k}=\overline{c}_{21}^{(k)}e^{-i(\sigma^{(0)}+\cdots+\sigma^{(k-1)})},  \qquad \alpha_{2k-1} = 0,
\qquad k\geq1,
$$
where $\sigma^{(k)}=\sigma_1^{(k)}+\sigma_2^{(k)}$. The fact that $c_{21}^{(k)}$ and $\sigma^{(k)}$ have period $p$ ensures that the new Verblunsky
coefficients
\begin{equation} \label{E:rot-per}
\hat{\alpha}_k = \alpha_k e^{i(k+1)\vartheta}, \qquad \vartheta=\frac{1}{2p}(\sigma^{(0)}+\cdots+\sigma^{(p-1)}),
\end{equation}
have period $2p$.

Let $\hat\mu$ and $\hat{f}$ be the measure and Schur function associated with the Schur parameters $\hat\alpha_k$. As a consequence of the
periodicity of $\hat\alpha_k$, the Schur iterate $\hat{f}_{2p}$ has the same Schur parameters $\hat\alpha_k$ as $\hat{f}$, hence
$\hat{f}_{2p}=\hat{f}$. Bearing in mind that any step of the Schur algorithm (\ref{E:SchurAlg-scalar}) is a rational transformation, the relation
$\hat{f}_{2p}=\hat{f}$ can be written as a polynomial equation for $\hat{f}(z)$ with polynomial coefficients in $z$. Therefore $\hat{f}(z)$, and thus
$z\hat{f}(z)$, are algebraic functions of $z$, which implies that the equation $z\hat{f}(z)=1$ has a finite number of roots. This means that the
singular part of $\hat\mu$ has a finite support, so it can not have a continuous part.

Relation (\ref{E:rot-per}) between $\alpha_k$ and $\hat\alpha_k$ implies that the corresponding measures $\mu$, $\hat\mu$ are connected by a rotation
(see \cite[page 473]{CGMV}), $d\mu(z)=d\hat{\mu}(e^{-i\vartheta}z)$, thus $\mu$ has no singular continuous part neither. Besides, from the link
between the measure $\mu$ and its Schur function $f$ we find that $f(z)=e^{i\vartheta}\hat{f}(e^{-i\vartheta}z)$, thus $f$ is algebraic too.

Now suppose that we modify a finite number of coins $C_k$, so that $C_{k+p}=C_k$ only holds for $k \geq k_0$. Then, the sequence $(\hat\alpha_k)_{k
\geq k_0}$ is periodic with period $2p$, and the corresponding Schur function, which is $\hat{f}_{k_0}$, must be algebraic. The Schur function
$\hat{f}$ is obtained from $\hat{f}_{k_0}$ by $k_0$ steps of the inverse Schur algorithm (\ref{E:SchurAlg-scalar}), each of them preserving the
algebraic character. Hence, $\hat{f}$ and $f$ are algebraic too, and the measures $\hat\mu$ and $\mu$ have no singular continuous part, just as in
the strictly periodic case.

With regard to QWs on $\Z$, the periodicity of the coins $C_k$ for any $k\in\Z$ with $|k| \geq k_0$ implies again the periodicity of the Schur
parameters $\hat\alpha_k$ given in (\ref{E:rot-per}) for the same range of indices. Therefore, the previous arguments show that the Schur functions
$f_+$, $f_-$ associated respectively with the Schur parameters $\alpha_k^+=\alpha_k$, $\alpha_k^-=-\overline{\alpha}_{-k-2}$ are algebraic. Since the
singular part of the matrix measure $\bs\mu$ of the QW is supported on the roots of $z^2f_+(z)f_-(z)=1$, the result follows from the fact that
$z^2f_+(z)f_-(z)$ is algebraic.
\end{proof}

In the case of QWs on $\Z$ with strictly periodic coins, stronger results can be achieved.

\begin{prop} \label{P:pure-periodic}
Any QW on $\Z$ with strictly periodic coins is free of localized states.
\end{prop}

\begin{proof}
If a QW on $\Z$ has strictly periodic coins, the full sequence $(\hat\alpha_k)_{k\in\Z}$ appearing in the proof of the previous proposition is
periodic too. The matrix measure $\bs\mu$ of the QW is a rotation of the measure $\hat{\bs\mu}$ with Verblunsky coefficients
$$
\hat{\bs\alpha}_k = \begin{pmatrix} 0 & -\overline{\hat\alpha}_{-k-2} \\ \hat\alpha_k & 0 \end{pmatrix}.
$$

The block CMV matrix $\hat{\bs\cC}$ with Verblunsky coefficients $\hat{\bs\alpha}_k$ is obtained by folding a two-sided CMV matrix $\hat\cC$ with
scalar Verblunsky coefficients $\hat\alpha_k$ (see \cite{CGMV}). Like any two-sided CMV matrix with periodic Verblunsky coefficients, $\hat\cC$ has
an absolutely continuous spectrum (see \cite{DaKiSi}), and the same holds for $\hat{\bs\cC}$ because it is related to $\hat\cC$ by a simple
reordering of the basis. This means that the scalar measure $\hat{\mu}_\psi$ defined by $\psi \, \hat{\bs\cC}^n \! \psi^\dag = \int_\T z^n
d\hat{\mu}_\psi(z)$, $n\in\Z$, is absolutely continuous for all $\psi$. Then, $\lim_{n\to\infty}\psi \, \hat{\bs\cC}^n \! \psi^\dag=0$ for any $\psi$
and Theorem \ref{T:Loc}.d.ii implies that no state exhibits localization.
\end{proof}


\subsection{Quasi-deterministic QWs} \label{SS:quasi-det}

QWs with diagonal coins $C_k$ for any $k$ are deterministic because the one-step transitions
$$
\begin{aligned}
& \Z \kern30pt |k\!\ua\> \to |k+1\!\ua\> \qquad |k\!\da\> \to |k-1\!\da\>
\\
& \Z_+ \kern20pt \cdots \to |2\!\da\> \to |1\!\da\> \to |0\!\da\> \to |0\!\ua\> \to |1\!\ua\> \to |2\!\ua\> \to \cdots
\end{aligned}
$$
take place with probability one. These QWs exhibit no localization, even when a finite number of defects appear, regardless of the number and details
of the defective coins.

The presence of a finite number of defects means that $C_k$ is diagonal for all but a finite number of sites $k$. In such a case, the related measure
$\bs\mu$ has null Verblunsky coefficients $\bs\alpha_k$ except for a finite number of indices $k$. Hence, $\bs\alpha_k=0$ for $k \geq k_0$ and
$\bs\mu$ is a Bernstein-Szeg\H{o} measure which can be expressed using the OLP as (see \cite{DaPuSi} for the general matrix case)
$$
d\bs\mu(e^{i\theta}) = [\bs{X}_{k_0}(e^{i\theta})^\dag \bs{X}_{k_0}(e^{i\theta})]^{-1} \frac{d\theta}{2\pi}.
$$
Since $\bs\mu$ is absolutely continuous, no state exhibits localization.

\medskip

Concerning the possibility of having a singular continuous part in the measure, it is known that sparse sequences of Verblunsky coefficients on
$\Z_+$ can give a measure which is exclusively singular continuous (see \cite{Go} and \cite[Section 12.5]{Si2}). This shows that such a pathological
situation can appear surprisingly close to the deterministic case corresponding to diagonal coins.

Another source of singular continuous measures are those measures supported on a Cantor type set (see \cite{Fa} and \cite[Section 2.12]{Si1}) or
those given by appropriate infinite Riesz products (see \cite{RIESZ} and \cite[Section 2.11]{Si1}). It would be interesting to search for QWs
corresponding to these kinds of measures, as well as to study the localization properties of such rather pathological situations. This could shed
light on the general picture for the localization properties of QWs with a measure whose singular part is purely continuous.

The models which we will analyze in detail are the QWs with a constant coin up to one defect at the origin. They are a special case of periodic QWs
with finite defects, so the localization dichotomy works for them. We will make an exhaustive analysis of localization in these examples, both on
$\Z$ and $\Z_+$, to illustrate the effectiveness the CGMV method beyond the case of a constant coin, and to understand how the single state
localization depends on the parameters of the models.

\section{QWs with one defect} \label{S:defect}

We will consider a general QW with coins $C_k$ which are constant except for the site $k=0$, i.e.,
\begin{equation} \label{E:C-D}
C_k = C = \begin{pmatrix} c_{11} & c_{12} \\ c_{21} & c_{22} \end{pmatrix}, \quad k\neq0,
\qquad
C_0 = D = \begin{pmatrix} d_{11} & d_{12} \\ d_{21} & d_{22} \end{pmatrix}.
\end{equation}
We will refer to this as a QW with one defect at the origin. To place the defect at the origin is simply a convention for the numbering of the sites
in $\Z$, but it is real restriction in $\Z_+$.

Remember that we only need to consider irreducible QWs, which means that we can assume without loss of generality that $c_{jj},d_{jj}\neq0$. We will
use the notation
$$
\sigma = \sigma_1+\sigma_2, \qquad \tau = \tau_1+\tau_2, \qquad \vartheta=\frac{\sigma}{2},
$$
where $e^{i\sigma_j}$ and $e^{i\tau_j}$ are the phases of $c_{jj}$ and $d_{jj}$ respectively.

The case of a diagonal coin $C$ is somewhat special. We know that it leads to an absolutely continuous Bernstein-Szeg\H{o} measure which, therefore,
yields a QW with no localization. Thus, for the general discussion we will suppose that $c_{21}\neq0$.

\subsection{QWs with one defect on the non negative integers} \label{SS:defect-Z+}

Let us consider first the coins (\ref{E:C-D}) on $\Z_+$. According to the results described in Section \ref{S:CMV}, the order indicated in
(\ref{E:order}) gives a transition matrix $U=\Lambda\cC\Lambda^\dag$, $\Lambda=\text{diag}(1,\lambda_1,\lambda_2,\dots)$, with
$$
\lambda_{2k-1} = e^{i(\tau_2+(k-1)\sigma_2)},
\qquad
\lambda_{2k} = e^{-i(\tau_1+(k-1)\sigma_1)},
\qquad
k \geq 1,
$$
and $\cC=\cC(\alpha_k)$ a CMV matrix with Verblunsky coefficients
$$
\alpha_{2k} = \begin{cases} \overline{d}_{21}, & \text{ if } k=0, \\ \overline{c}_{21}e^{-i(\tau+(k-1)\sigma)}, & \text{ if } k>0, \end{cases}
\qquad
\alpha_{2k+1}=0,
\qquad
k \geq0.
$$

The Verblunsky coefficients can be written as $\alpha_k = \hat{\alpha}_k e^{-i(k+1)\vartheta}$ with
\begin{equation} \label{E:ab-Z+}
(\hat{\alpha}_k) = (b,0,a,0,a,0,a,0,\dots), \quad a=\overline{c}_{21}e^{i(\frac{3}{2}\sigma-\tau)}, \quad b=\overline{d}_{21}e^{i\frac{\sigma}{2}}.
\end{equation}
This means that the measure $\mu$, the Carathéodory function $F$ and the OLP $X_k$ of the QW are related to those ones of
$\hat{\cC}=\cC(\hat{\alpha}_k)$ by (see \cite[page 473]{CGMV})
\begin{equation} \label{E:rot-Z+}
\kern-5pt d\mu(z)=d\hat{\mu}(e^{-i\vartheta}z),
\kern9pt
F(z)=\hat{F}(e^{-i\vartheta}z),
\kern9pt
X_k(z)=\hat{\lambda}_k\hat{X}_k(e^{-i\vartheta}z),
\end{equation}
$$
\hat{\lambda}_0=1, \qquad
\begin{cases}
\hat{\lambda}_{2k-1} = \lambda_{2k-1} e^{-ik\vartheta} = e^{i(k\frac{\sigma_2-\sigma_1}{2}+\tau_2-\sigma_2)},
\\
\hat{\lambda}_{2k} = \lambda_{2k} e^{ik\vartheta} = e^{i(k\frac{\sigma_2-\sigma_1}{2}+\sigma_1-\tau_1)},
\end{cases}
k\geq1,
$$
with an obvious notation for the elements corresponding to $\hat{\cC}$. In other words, the OLP of the QW are, up to a change of phases, a rotation
by an angle $\vartheta$ of those corresponding to a CMV matrix with Verblunsky coefficients $(b,0,a,0,a,0,a,0,\dots)$.

The related Schur function $\hat{f}=f_{a,b}$ has Schur parameters $(b,0,a,0,a,0,a,0,\dots)$. Its second Schur iterate $\hat{f}_2$ is the Schur
function $f_a$ whose Schur parameters are $(a,0,a,0,a,0,\dots)$, so from (\ref{E:SchurAlg-scalar}) we find the relations
\begin{equation} \label{E:fa-fab}
f_a(z) = \frac{1}{z^2} \frac{f_{a,b}(z)-b}{1-\overline{b}f_{a,b}(z)}, \qquad f_{a,b}(z) = \frac{z^2f_a(z)+b}{1+\overline{b}z^2f_a(z)}.
\end{equation}
In particular, setting $b=a$, $f_{a,b}$ becomes $f_a$. This leads to the quadratic equation $\overline{a}z^2f_a(z)^2+(1-z^2)f_a(z)-a=0$ for $f_a$
which yields the expression
\begin{equation} \label{E:fa}
\qquad f_a(z) = \frac{z^2-1 + \sqrt{\Delta_a(z)}}{2\overline{a}z^2}, \qquad \Delta_a(z) = (z^2-1)^2+4|a|^2z^2.
\end{equation}
Since $f_a$ is analytic in $\D$, the choice for the square root must result in the branch such that $\sqrt{\Delta_a(z)}\xrightarrow{z\to0}1$. Such a
choice implies that the boundary values of $f_a$ on the unit circle are\footnote{See \cite[Appendix]{CGMV} for a discussion about the boundary values
of $\sqrt{\Delta_a}$ on $\T$.}
\begin{equation} \label{E:fa-T}
\begin{gathered}
f_a(e^{i\theta}) =  \frac{e^{-i\theta}}{\overline{a}} (R_a(\theta)+i\sin\theta),
\\
R_a(\theta) =
\begin{cases}
\sgn(\cos\theta) \sqrt{|a|^2-\sin^2\theta}, & \text{ if } |\sin\theta|\leq|a|,
\\
-i \, \sgn(\sin\theta) \sqrt{\sin^2\theta-|a|^2}, & \text{ if } |\sin\theta|>|a|.
\end{cases}
\end{gathered}
\end{equation}

Thus, $|f_a(e^{i\theta})|=1$ if  $|\sin\theta|\leq|a|$ and $|f_a(e^{i\theta})|<1$ if $|\sin\theta|>|a|$. This also holds for $\hat{f}=f_{a,b}$
because any step of the Schur algorithm preserves the relations $|f(z)|<1$ and $|f(z)|=1$ at any point $z\in\T$. Therefore, according to
(\ref{E:SC-measure}), the weight $\hat{w}(\theta)$ of the measure $\hat{\mu}$ lives on $|\sin\theta|>|a|$, which defines two arcs which are symmetric
with respect to the real axis. The singular part is supported on the finite number of roots $z\in\T$ of $z\hat{f}(z)=1$, so it can have only mass
points. When these mass points are present, they must lie on any of the two complementary arcs given by
$$
\Gamma_a=\{e^{i\theta}:|\sin\theta|\leq|a|\},
$$
which are symmetric with respect to the imaginary axis. This is because the equality $z\hat{f}(z)=1$ implies $|\hat{f}(z)|=1$ for any $z\in\T$. The
consequences of these conclusions for the measure of the QW are obvious because $\mu$ is obtained simply rotating $\hat{\mu}$ by an angle
$\vartheta$.

Different coins (\ref{E:C-D}) giving the same pair $a,b\in\D$ have measures which only differ in a rotation. Concerning the relative values of $a$
and $b$, when the defect disappears, i.e., $D=C$, we get $b=a=c_{21}e^{i\frac{\sigma}{2}}$. Nevertheless, the defect not only changes the first Schur
parameter from $a$ to $b$, but it affects also the value of $a$ which acquires an extra phase $e^{i(\sigma-\tau)}$. Due to this, the equality $b=a$
can happen even with $D \neq C$, indeed it is equivalent to $d_{21}=c_{21}e^{i(\tau-\sigma)}$. The restriction $c_{21}\neq0$ that excludes the
special case of a diagonal coin $C$ means that we are considering $a\neq0$.

\subsection{QWs with one defect on the integers} \label{SS:defect-Z}

Assume now that we have the coins (\ref{E:C-D}) in $\Z$. From the general results of Section \ref{S:CMV} we find that the order indicated in
(\ref{E:order}) gives a transition matrix $\bs{U}=\bs{\Lambda}\bs{\cC}\bs{\Lambda^\dag}$,
$\bs{\Lambda}=\text{diag}(\1,\bs{\lambda}_1,\bs{\lambda_2},\dots)$, where
$$
\bs{\lambda}_{2k-1} = {\small \begin{pmatrix} e^{ik\sigma_1} & \kern-9pt 0 \\ 0 & \kern-9pt e^{i(\tau_2+(k-1)\sigma_2)} \end{pmatrix}},
\quad
\bs{\lambda}_{2k} = {\small \begin{pmatrix} e^{-i(\tau_1+(k-1)\sigma_1)} & \kern-9pt 0 \\ 0 & \kern-9pt e^{-ik\sigma_2} \end{pmatrix}},
\quad
k\geq1.
$$
and $\bs{\cC}=\bs{\cC}(\bs{\alpha}_k)$ is the CMV matrix with Verblunsky coefficients
$$
\begin{gathered}
\bs\alpha_{2k} = \begin{pmatrix} 0 & -\overline\alpha_{-2k-2} \\ \alpha_{2k} & 0 \end{pmatrix}, \qquad \bs\alpha_{2k+1}=\0, \qquad k\geq0,
\\
\alpha_{2k} =
\begin{cases} \overline{d}_{21}, & \text{ if } k=0,
\\
\overline{c}_{21} e^{-i(\tau+(k-1)\sigma)}, & \text{ if } k>0,
\\
\overline{c}_{21} e^{-ik\sigma}, & \text{ if } k\leq0,
\end{cases}
\end{gathered}
$$

As in the case of $\Z_+$, a rotation plays a useful role in our analysis. Defining,
\begin{equation} \label{E:ab-Z}
a=i|c_{21}|e^{i\frac{\sigma-\tau}{2}},
\quad
b=i\frac{c_{21}}{|c_{21}|}e^{i\frac{\tau-\sigma}{2}}\overline{d}_{21},
\quad
\omega=i\frac{c_{21}}{|c_{21}|}e^{i(\frac{\tau}{2}-\sigma)},
\end{equation}
we can rewrite $\bs{\alpha}_k = e^{-i(k+1)\vartheta} \hat{\bs{\alpha}}_k$, where
$$
(\hat{\bs{\alpha}}_k) = (\bs{\beta},\0,\bs{\alpha},\0,\bs{\alpha},\0,\bs{\alpha},\0,\dots),
\kern7pt
\bs{\alpha}=\begin{pmatrix}0&\omega a\\\overline{\omega}a&0\end{pmatrix},
\kern7pt
\bs{\beta}=\begin{pmatrix}0&\omega a\\\overline{\omega}b&0\end{pmatrix}.
$$

As a consequence, the measure $\bs{\mu}$, the Carathéodory function $\bs{F}$ and the OLP $\bs{X}_k$ of the QW are given by
\begin{equation} \label{E:rot-Z}
\kern-9pt d\bs{\mu}(z)=d\hat{\bs{\mu}}(e^{-i\vartheta}z),
\kern7pt
\bs{F}(z)=\hat{\bs{F}}(e^{-i\vartheta}z),
\kern7pt
\bs{X}_k(z)=\hat{\bs{\lambda}}_k\hat{\bs{X}}_k(e^{-i\vartheta}z),
\end{equation}
$$
\hat{\bs{\lambda}}_0=\1, \quad
\begin{cases}
\hat{\bs{\lambda}}_{2k-1} = \bs{\lambda}_{2k-1} e^{-ik\vartheta} =
\begin{pmatrix} e^{ik\frac{\sigma_1-\sigma_2}{2}} & \kern-15pt 0 \\ 0 & \kern-15pt e^{i(k\frac{\sigma_2-\sigma_1}{2}+\tau_2-\sigma_2)} \end{pmatrix},
\\
\hat{\bs{\lambda}}_{2k} = \bs{\lambda}_{2k} e^{ik\vartheta} =
\begin{pmatrix} e^{i(k\frac{\sigma_2-\sigma_1}{2}+\sigma_1-\tau_1)} & \kern-15pt 0 \\ 0 & \kern-15pt e^{ik\frac{\sigma_1-\sigma_2}{2}} \end{pmatrix},
\end{cases}
k\geq1,
$$
where elements with a hat are related to $\hat{\bs{\cC}}=\bs{\cC}(\hat{\bs{\alpha}}_k)$. Therefore, just as in the case of $\Z_+$, up to phases, the
OLP of the QW are obtained rotating by an angle $\vartheta$ the OLP going along with a one defect sequence of Verblunsky coefficients
$(\bs{\beta},\0,\bs{\alpha},\0,\bs{\alpha},\0,\bs{\alpha},\0,\dots)$. We should remark that the defect is only at the (2,1)-th coefficient of
$\bs{\beta}$.

The corresponding Schur function $\hat{\bs{f}}$ has as Schur parameters the antidiagonal sequence
$(\bs{\beta},\0,\bs{\alpha},\0,\bs{\alpha},\0,\bs{\alpha},\0,\dots)$. Applying Propositions \ref{P:antidiag} and \ref{P:Schur-2} we conclude that
$$
\hat{\bs{f}} = \begin{pmatrix} 0 & \omega f_a \\ \overline{\omega} f_{a,b} & 0 \end{pmatrix}, \qquad \omega\in\T,
$$
where $f_a$ and $f_{a,b}$ are the scalar Schur functions introduced in the previous subsection, although here $a$ and $b$ bear a different relation
to the coefficients of the coins.

The measure $\bs{\mu}$ of the QW is obtained as a simple rotation of $\hat{\bs{\mu}}$, so we only need to discuss this last one. We know that
$|f_a|=|f_{a,b}|=1$ in the two closed arcs $\Gamma_a$ and $|f_a|,|f_{a,b}|<1$ in the two open arcs $\T\setminus\Gamma_a$. Therefore, the same result
holds for $\|\hat{\bs{f}}\|=\max\{|f_a|,|f_{a,b}|\}$. Using (\ref{E:SC-measure}), we find that, for a.e. $\theta$, the weight $\hat{\bs{w}}(\theta)$
of $\hat{\bs{\mu}}$ is singular if and only if $e^{i\theta}\in\Gamma_a$. Furthermore, (\ref{E:SC-measure}) also yields for a.e. $\theta$
$$
\begin{aligned}
\hat{\bs{w}}(\theta) & = \re[(\1+e^{i\theta}\hat{\bs{f}})(\1-e^{i\theta}\hat{\bs{f}})^{-1}]
\\
& = (\1-e^{-i\theta}\hat{\bs{f}}^\dag)^{-1} (\1-\hat{\bs{f}}(e^{i\theta})^\dag\hat{\bs{f}}(e^{i\theta})) (\1-e^{i\theta}\hat{\bs{f}})^{-1}.
\end{aligned}
$$
The equality
$$
\1-\hat{\bs{f}}^\dag\hat{\bs{f}} = \begin{pmatrix} 1-|f_{a,b}|^2 & 0 \\ 0 & 1-|f_a|^2 \end{pmatrix}
$$
shows that $\det\hat{\bs{w}}(\theta)=0$ implies $|f_{a,b}(e^{i\theta})|=1$ or $|f_{a,b}(e^{i\theta})|=1$. Since these two conditions hold
simultaneously, $\det\hat{\bs{w}}(\theta)=0$ necessarily gives $\hat{\bs{f}}(e^{i\theta})^\dag\hat{\bs{f}}(e^{i\theta})=\1$ and thus
$\hat{\bs{w}}(\theta)=0$. We conclude that $\hat{\bs{w}}$ is zero in $\Gamma_a$ and non singular in $\T\setminus\Gamma_a$.

As for the singular part of $\hat{\bs{\mu}}$, it is supported on a finite number of points, the roots $z\in\T$ of $z^2f_a(z)f_{a,b}(z)=1$, and any of
these roots must satisfy $|f_a(z)f_{a,b}(z)|=1$. Hence, the singular part only can have mass points located at $\Gamma_a$.

In contrast to the case of $\Z_+$, three parameters $a,b\in\D$, $\omega\in\T$ characterize now the coins (\ref{E:C-D}) with the same measure up to
rotations. On the other hand, just as in the case of $\Z_+$, the equality $b=a$ does not hold only for $D=C$ because, remarkably, it is equivalent to
the same condition $d_{21}=c_{21}e^{i(\tau-\sigma)}$ appearing for the non negative integers. Also, the consequences of the defect are not only
encoded in $b$, but the imaginary value $a=i|c_{21}|$ for a constant coin $C$ acquires with the defect an extra phase $e^{i\frac{\sigma-\tau}{2}}$
which is the square root of the similar extra phase for the case of $\Z_+$. As in $\Z_+$, we only need to consider $a\neq0$ because we know that
$a=0$ yields no localization.

\section{Localization: one defect on $\Z$} \label{S:Loc-defect-Z}

The previous discussions indicate that the study of localization in QWs leads to the analysis of the mass points (in general, the singular part) of
the corresponding measure. This analysis is worth doing specially for QWs where the localization dichotomy works, such as periodic QWs with a finite
number of perturbations. The QWs with one defect are just the simplest examples of this case. As we pointed out, in such situations we will talk
about QWs with or without localization because then localization can be viewed as a global property: it holds for no state or for almost any state.

Surprisingly, for QWs with one defect, the mass points analysis is simpler for $\Z$ than for $\Z_+$, among other reasons, due to the symmetry of the
mass points with respect to the origin, which is lost for $\Z_+$. Hence, we will study first localization in a QW on $\Z$ with coins (\ref{E:C-D}).


Concerning previous related results, N. Konno has proved in \cite{K} that the state $\frac{1}{\sqrt{2}}|0\!\ua\>+\frac{i}{\sqrt{2}}|0\!\da\>$
exhibits localization in $\Z$ for the perturbation of the constant Hadamard coin
$$
H = \frac{1}{\sqrt{2}}\begin{pmatrix}1&1\\1&-1\end{pmatrix}
$$
given by
\begin{equation} \label{E:Konno2}
C = H, \qquad D = \frac{1}{\sqrt{2}}\begin{pmatrix}1&e^{i\phi}\\e^{-i\phi}&-1\end{pmatrix},
\end{equation}
whenever $e^{i\phi}\neq1$. We will recover this result as a particular case of our analysis but, assuming it for the moment, observe that we have a
stronger result: the dichotomy implies that all the states should exhibit localization up to, at most, one state per site. Indeed, we will see that
in this model any state $\alpha|0\!\ua\>+\beta|0\!\da\>$ exhibits localization.

The aim of this section is to perform a systematic study of localization for any QW with one defect on $\Z$, which can reveal in the simplest
examples the coin dependence of localization properties.

Therefore, our purpose is to determine which coins (\ref{E:C-D}) give in $\Z$ a measure with mass points. Subsection \ref{SS:defect-Z} shows that
these models fall into groups with a common measure up to rotations, each such a group characterized by the three parameters $a,b\in\D$,
$\omega\in\T$ given in (\ref{E:ab-Z}). A canonical representative of the measures in a given group is that one $\hat{\bs{\mu}}=\bs{\mu}_{a,b}^\omega$
associated with the common CMV matrix $\hat{\bs{\cC}}=\bs{\cC}(\hat{\bs{\alpha}}_k)$ of the group given in Subsection \ref{SS:defect-Z}, whose weight
and mass points live in $\T\setminus\Gamma_a$ and $\Gamma_a$ respectively.

Coins (\ref{E:C-D}) with the same values of these parameters have the same mass points up to rotations and, therefore, the corresponding QWs have the
same localization character. For instance, in $\Z$, the constant Hadamard coin has the same values $a=b=\frac{i}{\sqrt{2}}$, $\omega=1$ as its
perturbation
$$
C = H, \qquad D = \frac{1}{\sqrt{2}}\begin{pmatrix}e^{i\phi}&1\\1&-e^{-i\phi}\end{pmatrix},
$$
which proves in a very simple way that no state in such a QW exhibits localization (see \cite{Ko} for a different proof in the particular case of the
state $\frac{1}{\sqrt{2}}|0\!\ua\>+\frac{i}{\sqrt{2}}|0\!\da\>$).

\subsection{Mass points of $\bs{\mu}_{a,b}^\omega$} \label{SS:roots-Z}

Bearing in mind that localization in a QW with one defect on $\Z$ only depends on the associated parameters $a,b,\omega$, we can restrict our
attention to the canonical representative $\bs{\mu}_{a,b}^\omega$. Proposition \ref{P:antidiag} states that the corresponding mass points are the
roots $z\in\T$ of $g_{a,b}(z)=z^2f_a(z)f_{a,b}(z)=1$ such that
\begin{equation} \label{E:mab}
m_{a,b}(z) = \lim_{r\ua1} \frac{1-r}{1-g_{a,b}(rz)} \neq 0.
\end{equation}
These conditions do not depend on $\omega$, so the mass points, and thus the localization behaviour, only depend on $a,b$ but not on $\omega$.

The choice of the square root $\sqrt{\Delta_a}$ makes $f_a$ analytic in $\T$ except at the branch points, i.e., the solutions of $\Delta_a=0$, which
are the four boundary points $\partial\Gamma_a$ of the two arcs $\Gamma_a$,
$$
\partial\Gamma_a = \{\pm z_a, \pm\overline{z}_a\}, \qquad z_a=\rho_a+i|a|, \qquad \rho_a=\sqrt{1-|a|^2}.
$$
The relation (\ref{E:fa-fab}) between $f_a$ and $f_{a,b}$ shows that $f_{a,b}$ is analytic in $\T\setminus\partial\Gamma_a$ too, and so the same is
true for $g_{a,b}$. In consequence, we find from Proposition \ref{P:antidiag} that the measure $\bs{\mu}_{a,b}^\omega$ has a mass point at any root
$z\in\T\setminus\partial\Gamma_a$ of $g_{a,b}(z)=1$. Indeed, we know that these roots must be in the interior $\Gamma_a^0$ of $\Gamma_a$ because
there is no root in $\T\setminus\Gamma_a$.

Therefore, the roots $z\in\T$ of $g_{a,b}(z)=1$ can lie on $\Gamma_a^0$, and then they are mass points for sure, or they can be on
$\partial\Gamma_a$, in which case we should check (\ref{E:mab}) to decide if they are mass points or not.

\subsubsection{Mass points on $\partial\Gamma_a$} \label{SSS:mass-boundary-Z}

We will prove that, although the points of $\partial\Gamma_a$ can be roots of $g_{a,b}(z)=1$, they are never mass points of $\bs{\mu}_{a,b}^\omega$
because condition (\ref{E:mab}) is not satisfied. We will only consider $z_a$, the analysis for the remaining three points is similar.

First, let us find the values of $b$ which make $z_a$ a root of $g_{a,b}(z)=1$. We know that $|f_a(z_a)|=1$ because $z_a\in\Gamma_a$, so from
(\ref{E:fa-fab}) we obtain
$$
g_{a,b}(z_a) = \frac{z_a^2f_a(z_a)+b}{\overline{z_a^2f_a(z_a)}+\overline{b}},
$$
and $g_{a,b}(z_a) = 1$ becomes equivalent to $\im (z_a^2f_a(z_a)+b)=0$. On the other hand, the expression (\ref{E:fa}) for $f_a$ gives $z_a^2f_a(z_a)
= i\frac{a}{|a|}z_a$, so that $z_a$ is a root of $g_a(z)=1$ if and only if $\im b = -\im(i\frac{a}{|a|}z_a)$.

Now, assume that $b$ satisfies the condition $\im b = -\im(i\frac{a}{|a|}z_a)$, and let us compute $\lim_{r\ua1}m(rz_a)$. The first order Taylor
expansion of $\Delta_a(rz_a)$ at $r=1$ gives
$$
\Delta_a(rz_a) = K_1(1-r)+O((1-r)^2), \qquad K_1\neq0,
$$
which, using (\ref{E:fa}), yields
\begin{equation} \label{E:z2fa-exp-r}
(rz_a)^2f_a(rz_a) = i\frac{a}{|a|}z_a + K_2\sqrt{1-r} + O(1-r), \qquad K_2\neq 0.
\end{equation}
Inserting this into the relation
$$
g_{a,b}(z)-1 = \frac{(z^2f_a(z)+1)(z^2f_a(z)-1)+2i(\im b)z^2f_a(z)}{1+\overline{b}z^2f_a(z)},
$$
obtained from (\ref{E:fa-fab}), leads to
$$
g_{a,b}(rz_a)-1 = K \ts \re(i\frac{a}{|a|}z_a) \sqrt{1-r} + O(1-r), \qquad K\neq0.
$$
This proves that $\lim_{r\ua1}m(rz_a)=0$ when $\re(i\frac{a}{|a|}z_a)\neq0$. The equality $\re(i\frac{a}{|a|}z_a)=0$ would imply $|\im
b|=|\re(i\frac{a}{|a|}z_a)|=1$, which is not possible, thus we conclude that there is no mass point at $z_a$, even if it is a root of $g_{a,b}(z)=1$.

\subsubsection{Mass points on $\Gamma_a^0$} \label{SSS:mass-interior-Z}

At this point we know that the mass points of $\bs{\mu}_{a,b}^\omega$ are the roots of $g_{a,b}(z)=1$ in $\Gamma_a^0$. We can restrict our analysis
to the right arc $\Gamma_a^+=\{e^{i\theta}\in\Gamma_a^0:\cos\theta\geq0\}$ of $\Gamma_a^0$ because the mass points appear in pairs $\pm z$, one
belonging to $\Gamma_a^+$ and the opposite one lying on the left arc $\Gamma_a^-=\{e^{i\theta}\in\Gamma_a^0:\cos\theta\leq0\}$.

The study of the roots in $\Gamma_a^+$ is simplified under the change of variables
$$
\zeta=\zeta(z)=-z^2f_a(z),
$$
which maps $\Gamma_a^+$ one to one onto the arc (see figure \ref{F:z-zeta})
$$
\Sigma_a = \ts \{\frac{a}{|a|}e^{it} : \cos t < |a|\},
$$
and $\partial\Gamma_a^+=\{z_a,\overline{z}_a\}$ onto $\partial\Sigma_a=\{\zeta_a^-,\zeta_a^+\}$, $\zeta_a^\pm=\frac{a}{|a|}(|a|\pm i\rho_a)$. These
mapping properties can be inferred from the expression
$$
\zeta(e^{i\theta}) = -\frac{e^{i\theta}}{\overline{a}} \left(\sqrt{|a|^2-|\sin^2\theta|}+i\sin\theta\right), \qquad
e^{i\theta}\in\overline{\Gamma_a^+},
$$
obtained from (\ref{E:fa-T}), which shows that $\zeta(z_a)=\zeta_a^-$, $\zeta(\overline{z}_a)=\zeta_a^+$ and the argument of $\zeta(e^{i\theta})$ is
increasing in $\theta$ for $e^{i\theta}\in\overline{\Gamma_a^+}$. The inverse mapping is
$$
z=z(\zeta) = \frac{1-\overline{a}\zeta}{|1-\overline{a}\zeta|}.
$$
The arc $\Sigma_a$ can be alternatively described as
$$
\ts \Sigma_a = \{\zeta\in\T : \re(\overline{a}\zeta) < |a|^2\} = \{\zeta\in\T : |a-\frac{\zeta}{2}| < \frac{1}{2}\},
$$
a result that will be of interest later on.

\begin{figure}
\includegraphics[width=12cm]{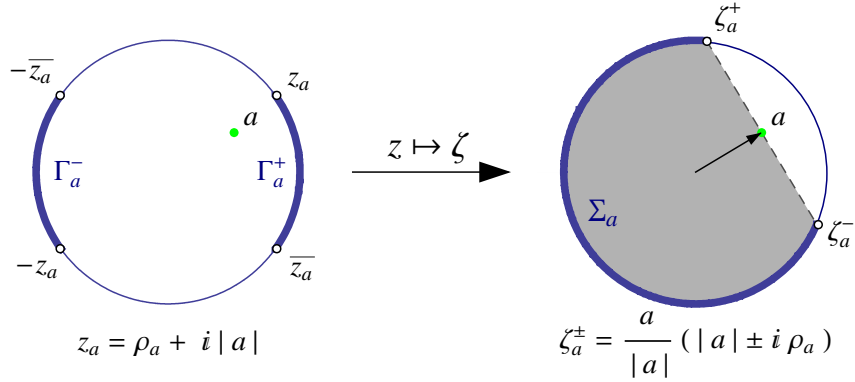}
\caption{The transformation $z \mapsto \zeta$ maps both $\Gamma_a^+$ and $\Gamma_a^-$ one to one onto $\Sigma_a$. $S(a)$ is the open set limited by
the arc $\Sigma_a$ and the straight line passing through $\zeta_a^+$ and $\zeta_a^-$, in grey color in the figure.} \label{F:z-zeta}
\end{figure}

Bearing in mind that $|f_a|=1$ in $\Gamma_a^+$, we find that
$$
g_{a,b}(z)=\frac{z_a^2f_a(z)+b}{\overline{z_a^2f_a(z)}+\overline{b}}, \qquad z\in\Gamma_a^+,
$$
so the translation of the equation for $z$ to the new variable $\zeta$ is
$$
g_{a,b}(z)=1, \quad z\in\Gamma_a^+ \quad \Leftrightarrow \quad \im b = \im \zeta, \quad \zeta\in\Sigma_a.
$$
Given $b\in\D$, the solutions $\zeta\in\T$ of the equation $\im b = \im \zeta$ are
$$
\zeta_\pm(b) = \pm\sqrt{1-\im^2 b}+i\im b.
$$
The values of $a$ which are compatible with $\zeta_\pm(b)$ are given respectively by any of the equivalent conditions
\begin{equation*} \label{E:mass+-}
\ts \zeta_\pm(b) \in \Sigma_a \kern5pt \Leftrightarrow \kern5pt \re(\overline{a}\zeta_\pm(b))<|a|^2 \kern5pt \Leftrightarrow \kern5pt
|a-\frac{1}{2}\zeta_\pm(b)|>\frac{1}{2}. \tag{$\mathbf{M}_\pm$}
\end{equation*}

Therefore, the measure $\bs{\mu}_{a,b}^\omega$ has mass points if and only if at least one of the conditions $\mathbf{M}_+$, $\mathbf{M}_-$ is
satisfied. For each of the conditions $\mathbf{M}_+$, $\mathbf{M}_-$ which is satisfied, there is a pair of mass points at $\pm z_+(a,b)$, $\pm
z_-(a,b)$ respectively, where
\begin{equation} \label{E:z+-ab}
z_\pm(a,b) = \frac{1-\overline{a}\zeta_\pm(b)}{|1-\overline{a}\zeta_\pm(b)|} \in \Gamma_a^+, \qquad -z_\pm(a,b)\in\Gamma_a^-.
\end{equation}
Hence, we have the following possibilities:
\begin{itemize}
\item[($\mathbf{M}^0$)] If none of $\mathbf{M}_\pm$ are satisfied, $\bs{\mu}_{a,b}^\omega$ has no mass point.
\item[($\mathbf{M}^2_+$)] If $\mathbf{M}_+$ is satisfied but $\mathbf{M}_-$ is not, $\bs{\mu}_{a,b}^\omega$ has 2 mass points:
\newline
$z_+(a,b)\in \Gamma_a^+$ and $-z_+(a,b)\in\Gamma_a^-$.
\item[($\mathbf{M}^2_-$)] If $\mathbf{M}_-$ is satisfied but $\mathbf{M}_+$ is not, $\bs{\mu}_{a,b}^\omega$ has 2 mass points:
\newline
$z_-(a,b)\in \Gamma_a^+$ and $-z_-(a,b)\in\Gamma_a^-$.
\item[($\mathbf{M}^4$)] If $\mathbf{M}_\pm$ are both satisfied, $\bs{\mu}_{a,b}^\omega$ has 4 mass points:
\newline
$z_\pm(a,b)\in \Gamma_a^+$ and $-z_\pm(a,b)\in\Gamma_a^-$.
\end{itemize}
The case of 2 mass points is characterized by $\mathbf{M}^2\equiv(\mathbf{M}^2_+\text{ or }\mathbf{M}^2_-)$, while
$\mathbf{M}\equiv(\mathbf{M}_+\text{ or }\mathbf{M}_-)$ is the condition for the existence of mass points.

\subsection{Localization pictures on $\Z$: dependence on $a$ and $b$} \label{SS:Loc-ab-Z}

The localization dichotomy for one defect on $\Z$ does not depend on $\omega$, but only on $a,b$. Hence, we can discuss two kind of problems: Given
$a$, which values of $b$ yield localization? Given $b$, which values of $a$ yield localization?

\subsubsection{From $b$ to $a$} \label{SSS:Loc-btoa-Z}

The last way of expressing $\mathbf{M}_\pm$ in ($\mathbf{M}_\pm$) above means that $a$ lies outside the closed disk $\cD_b^\pm$ of center
$\zeta_\pm(b)/2$ and radius $1/2$. Therefore, the different cases can be stated as (see figure \ref{F:Zba}):
\begin{itemize}
\item[($\mathbf{M}^0$)] $a\in\cD_b^+\cap\cD_b^- \Leftrightarrow \bs{\mu}_{a,b}^\omega$ has no mass point.
\item[($\mathbf{M}^2_+$)] $a\in\cD_b^-\setminus\cD_b^+ \Leftrightarrow \bs{\mu}_{a,b}^\omega$ has 2 mass points $\pm z_+(a,b)$.
\item[($\mathbf{M}^2_-$)] $a\in\cD_b^+\setminus\cD_b^- \Leftrightarrow \bs{\mu}_{a,b}^\omega$ has 2 mass points $\pm z_-(a,b)$.
\item[($\mathbf{M}^4$)] $a\notin\cD_b^+\cup\cD_b^- \Leftrightarrow \bs{\mu}_{a,b}^\omega$ has 4 mass points $\pm z_+(a,b),\pm z_-(a,b)$.
\end{itemize}
We conclude that, given a value of $b$, a QW with one defect on $\Z$ exhibits localization if and only if $a\notin\cD_b^+\cap\cD_b^-$.

\begin{figure}
\includegraphics[width=5cm]{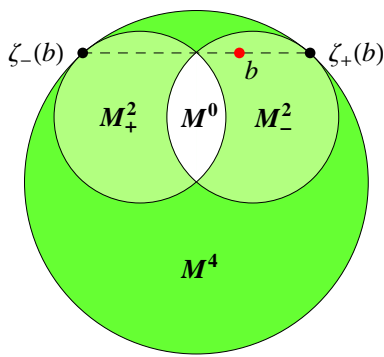}
\caption{{\bf Localization for one defect on $\Z$ (from $b$ to $a$).} In green color the values of $a$ giving localization for the choice of $b$ in
red. They only depend on $\im b$. In light green the values of $a$ with 2 mass points and in dark green those with 4 mass points.} \label{F:Zba}
\end{figure}

\subsubsection{From $a$ to $b$} \label{SSS:Loc-atob-Z}

Now we look at the first condition in ($\mathbf{M}_\pm$) above. To make it more explicit let us decompose $\Sigma_a=\Sigma_a^+\cup\Sigma_a^-$ into
its right and left parts $\Sigma_a^+=\{\zeta\in\Sigma_a:\re\zeta\geq0\}$, $\Sigma_a^-=\{\zeta\in\Sigma_a:\re\zeta\leq0\}$. Then, a choice of $a$
fixes $\Sigma_a^\pm$, and so the interval $\im\Sigma_a^\pm$ where $\im b$ must lie to fulfill $\mathbf{M}_\pm$. This means that localization for QWs
with one defect on $\Z$ only depends on $a$ and $\im b$, but not on $\re b$.

Taking into account that the angular amplitude of $\Sigma_a$ is bigger than $\pi$ (remember that we are considering $a\neq0$), we find three
possibilities for $\im\Sigma_a$ (see figure \ref{F:Zab}):
\begin{itemize}
\item[($\mathbf{A}$)] $(\re\zeta_a^+)(\re\zeta_a^-)\geq0 \kern7pt \Leftrightarrow \kern7pt (-1,1)\subset\im\Sigma_a$.
\item[($\mathbf{B}_+$)] $(\re\zeta_a^+)(\re\zeta_a^-)<0, \kern7pt \im a>0 \kern7pt \Leftrightarrow \kern7pt \im\Sigma_a=[-1,r_+), \kern7pt r_+<1$.
\item[($\mathbf{B}_-$)] $(\re\zeta_a^+)(\re\zeta_a^-)<0, \kern7pt \im a<0 \kern7pt \Leftrightarrow \kern7pt \im\Sigma_a=(r_-,1], \kern7pt r_->-1$.
\end{itemize}
The value $(\re\zeta_a^+)(\re\zeta_a^-) = \frac{|a|^4-\im^2 a}{|a|^2}$ turns these three cases into the following localization criteria:
\begin{itemize}
\item[($\mathbf{A}$)] Localization $\forall b \kern5pt \Leftrightarrow \kern5pt |\im a|\leq|a|^2 \kern5pt \Leftrightarrow \kern5pt
|a-\frac{i}{2}|\geq\frac{1}{2} \text{ or } |a+\frac{i}{2}|\geq\frac{1}{2}$.
\item[($\mathbf{B}_+$)] Localization for $\im b<r_+\in(0,1) \kern5pt \Leftrightarrow \kern5pt \im a>|a|^2 \kern5pt \Leftrightarrow \kern5pt
|a-\frac{i}{2}|<\frac{1}{2}$.
\item[($\mathbf{B}_-$)] Localization for $\im b>r_-\in(-1,0) \kern5pt \Leftrightarrow \kern5pt \im a<-|a|^2 \kern5pt \Leftrightarrow \kern5pt
|a+\frac{i}{2}|<\frac{1}{2}$.
\end{itemize}

\begin{figure}
\includegraphics[width=4.7cm]{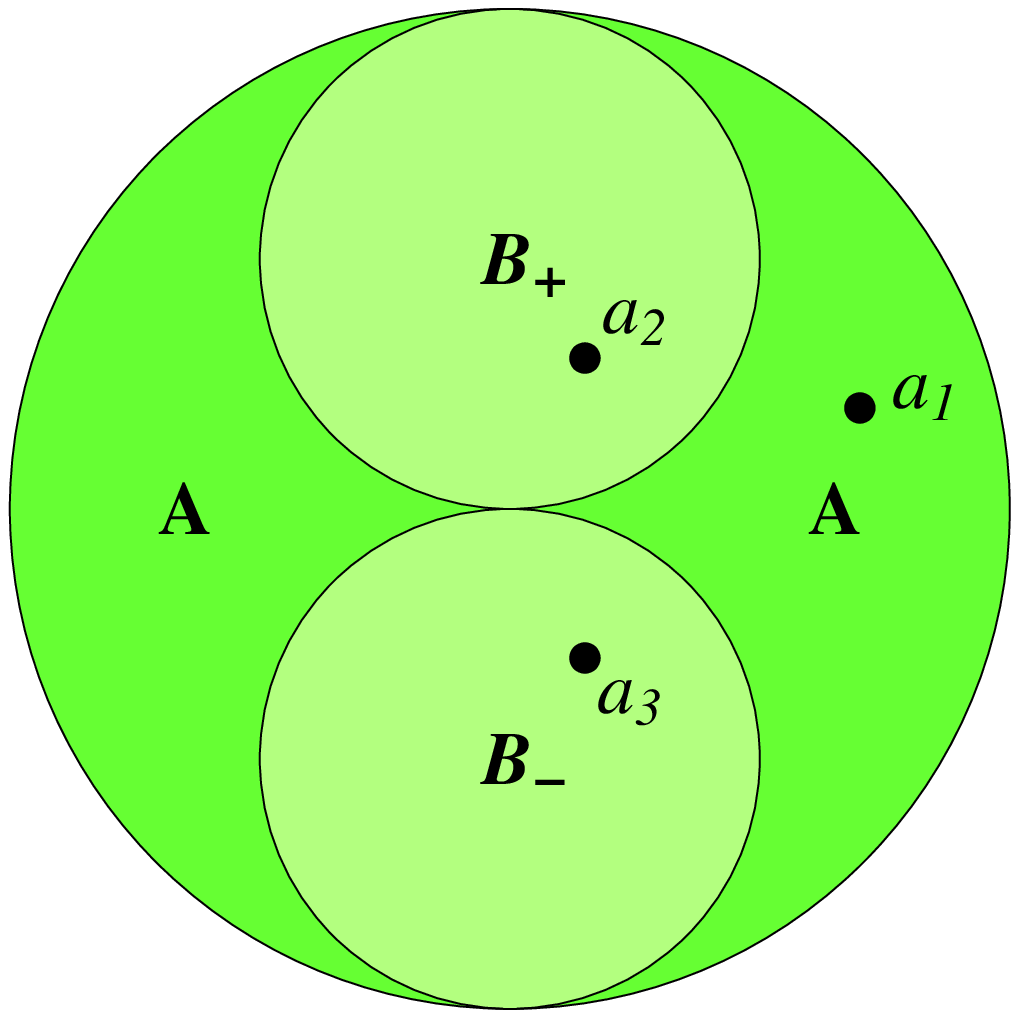}

\vspace*{.2cm}

\includegraphics[width=4.3cm]{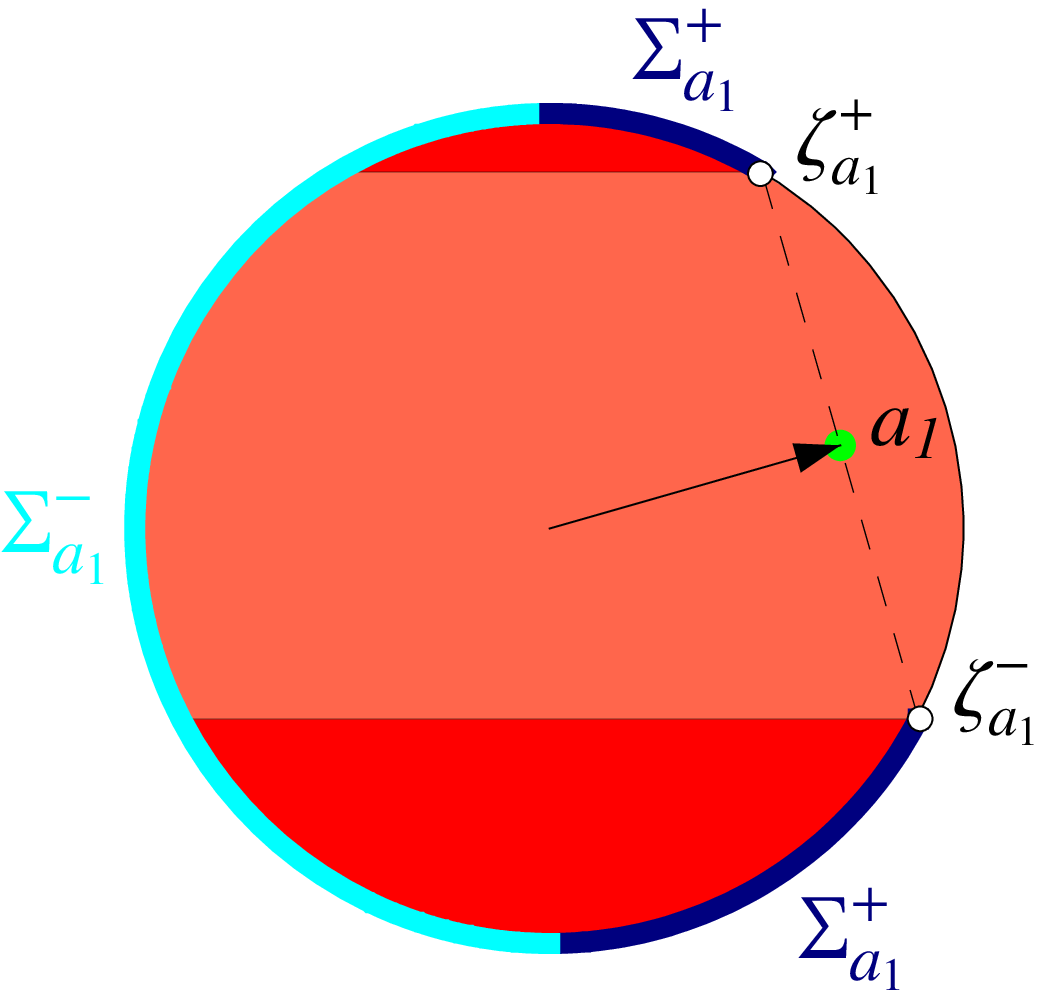}
\hspace{.3cm}
\includegraphics[width=4.3cm]{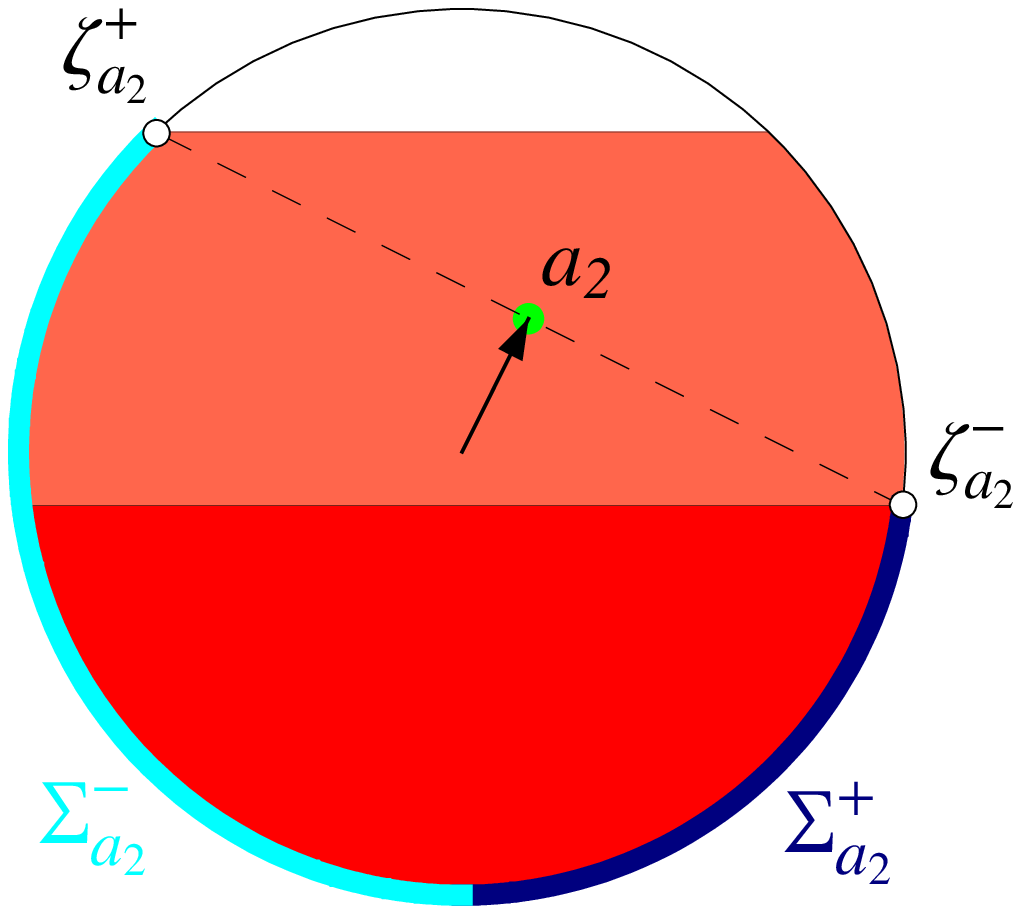}
\hspace{.18cm}
\includegraphics[width=4.3cm]{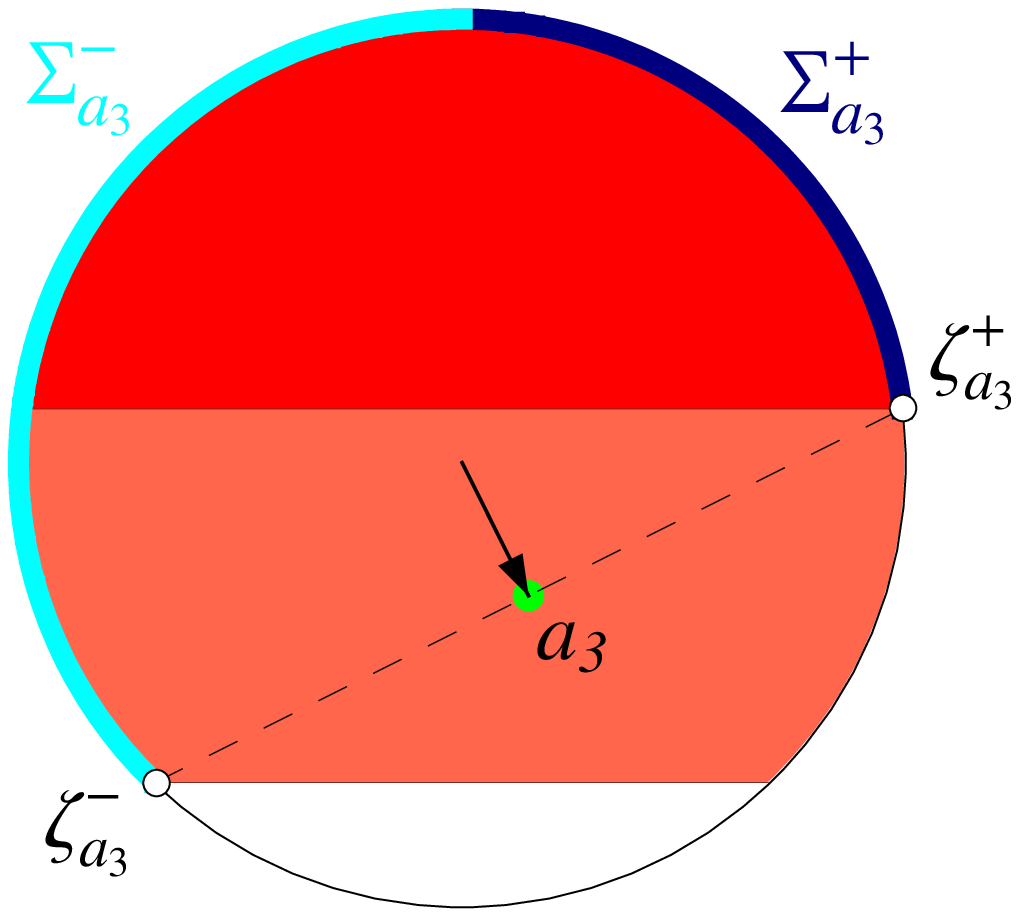}
\caption{{\bf Localization for one defect on $\Z$ (from $a$ to $b$).} The upper figure shows in dark green the values of $a$ giving localization for
any $b$. The upper (lower) circle in light green are the values of $a$ such that localization fails for $b$ lying on an upper (lower) band $\im b
\geq r_+$ ($\im b \leq r_-$). The lower figures represent in red color the values of $b$ giving localization for each of the three values of $a$
shown in the upper figure. They are characterized by $\im b\in\im\Sigma_a$. A pair of mass points appears for each of the conditions $\im
b\in\Sigma_a^\pm$ which is satisfied. Therefore, the dark red covers the values of $b$ with 4 mass points, while the light red covers those with 2
mass points.} \label{F:Zab}
\end{figure}

Roughly speaking, the values of $a$ split into three regions delimited by two circles with radius 1/2 centered at $\pm i/2$. Outside these circles
localization holds for any defect. Inside the upper or lower circle there is respectively an upper and lower bound for $\im b$ which delimits the
defects giving localization (see figure \ref{F:Zab}).

Notice that, for any $a\in\D$, localization holds at least for $b$ lying on the open set $S(a)$ limited by the arc $\Sigma_a$ and the straight line
joining $\zeta_a^+$ and $\zeta_a^-$ (see figure \ref{F:z-zeta}), that is,
\begin{equation} \label{E:S}
\ts S(a) = \{r\frac{a}{|a|}e^{it}:\cos t<|a|,r<1\}.
\end{equation}
Indeed, $S(a)$ yields exactly the values of $b$ giving localization when $a$ is imaginary because in that case $\im\zeta_a^+=\im\zeta_a^-$. Thus,
among the values of $a$ with the same modulus, the biggest region of values of $b$ without localization holds for $\re a = 0$. Since
$\im\Sigma_a^+=\im\Sigma_a^-$ for an imaginary value of $a$, it also ensures 4 mass points in case of localization.

For a fixed $a$, the bounds $r_\pm=r_\pm(a)$ are $r_+(a)=\max\im\{\zeta_a^+,\zeta_a^-\}$ and $r_-(a)=\min\im\{\zeta_a^+,\zeta_a^-\}$, i.e.,
$$
r_\pm(a)=\im a\pm\frac{\rho_a}{|a|}|\re a|.
$$
These bounds also permit to distinguish between the values of $b$ giving 2 or 4 mass points, once $a$ is chosen. There are 4 mass points when $\im b
\in \im\Sigma_a^+\cap\im\Sigma_a^-$, and only 2 mass points if $\im b\in\im\Sigma_a^+\setminus\im\Sigma_a^-$ or $\im
b\in\im\Sigma_a^-\setminus\im\Sigma_a^+$. Looking separately at the three previous possibilities we find that 2 mass points appear when $b$ lies on a
band limited by two horizontal lines passing through $\zeta_a^+$ and $\zeta_a^-$. Hence, the situation in the three cases above can be more precisely
described as follows (see figure \ref{F:Zab}):

\begin{itemize}

\item[($\mathbf{A}$)]
$\begin{cases}
\im b < r_-(a), & \text{ 4 mass points, }
\\
r_-(a) \leq \im b \leq r_+(a), & \text{ 2 mass points, }
\\
r_+(a) < \im b, & \text{ 4 mass points. }
\end{cases}$

\item[($\mathbf{B}_+$)]
$\begin{cases}
\im b < r_-(a), & \text{ 4 mass points, }
\\
r_-(a) \leq \im b < r_+(a), & \text{ 2 mass points, }
\\
r_+(a) \leq \im b, & \text{ no mass points. }
\end{cases}$

\item[($\mathbf{B}_-$)]
$\begin{cases}
\im b \leq r_-(a), & \text{ no mass points, }
\\
r_-(a) < \im b \leq r_+(a), & \text{ 2 mass points, }
\\
r_+(a) < \im b, & \text{ 4 mass points. }
\end{cases}$

\end{itemize}

\section{Asymptotic return probabilities: one defect on $\Z$} \label{S:Ret-defect-Z}

To compute the asymptotics of $p_{\alpha,\beta}^{(k)}(n)$ as in (\ref{E:ARP}) we need, not only the mass points (known from the previous results),
but also their masses and the OLP related to the site $k$. Let us see how to make the computations with the canonical representative
$d\hat{\bs{\mu}}=d\bs{\mu}_{a,b}^\omega$ instead of the actual measure $d\bs{\mu}(z)=d\hat{\bs{\mu}}(e^{-i\vartheta}z)$ of the QW.

Introducing (\ref{E:rot-Z}) in (\ref{E:psi}) we obtain the relation $\bs{\psi}_{\alpha,\beta}^{(k)}(z) =
\hat{\bs{\psi}}_{\hat\alpha,\hat\beta}^{(k)}(e^{-i\vartheta}z)$ between the corresponding functions for the state $|\Psi_{\alpha,\beta}^{(k)}\>$,
where
$$
\begin{cases}
\hat\alpha = \hat\lambda_{2j}^{(1)} \alpha, \quad \hat\beta = \hat\lambda_{2j+1}^{(2)} \beta, & \text{ if } k=j,
\\
\hat\alpha = \hat\lambda_{2j+1}^{(1)} \alpha, \quad \hat\beta = \hat\lambda_{2j}^{(2)} \beta, & \text{ if } k=-j-1,
\end{cases}
\quad j\geq0,
$$
and $\bs{\hat\lambda}_k=\text{diag}(\hat\lambda_k^{(1)},\hat\lambda_k^{(2)})$. In particular, $\bs{\psi}_{1,0}^{(k)}(z) = \kappa_j^{(1)}
\hat{\bs{\psi}}_{1,0}^{(k)}(e^{-i\vartheta}z)$ and $\bs{\psi}_{0,1}^{(k)}(z) = \kappa_j^{(2)} \hat{\bs{\psi}}_{1,0}^{(k)}(e^{-i\vartheta}z)$ with
$\kappa_j^{(l)}\in\T$. Hence, (\ref{E:ARP}) can be written as
\begin{equation} \label{E:ARP-hat}
p_{\alpha,\beta}^{(k)}(n) \underset{n}{\sim} \text{\SMALL
$\left|\sum_{z\in\T} z^n\hat{\bs{\psi}}_{\hat\alpha,\hat\beta}^{(k)}(z)\hat{\bs{\mu}}(\{z\})\hat{\bs{\psi}}_{1,0}^{(k)}(z)^\dag\right|^2 +
\left|\sum_{z\in\T} z^n\hat{\bs{\psi}}_{\hat\alpha,\hat\beta}^{(k)}(z)\hat{\bs{\mu}}(\{z\})\hat{\bs{\psi}}_{0,1}^{(k)}(z)^\dag\right|^2$}.
\end{equation}

For convenience, while performing the calculations we will omit the hat on $\hat\alpha,\hat\beta$ so that at the end of the computations we should
make the substitution $\alpha,\beta\to\hat\alpha,\hat\beta$. We also remember that $p_{\alpha,\beta}^{(k)}(2n-1)=0$, thus we only must consider
$p_{\alpha,\beta}^{(k)}(2n)$.

\subsection{Masses of $\hat{\bs{\mu}}=\bs{\mu}_{a,b}^\omega$} \label{SS:masses-Z}

There are 4 possible mass points: $\pm z_+(a,b)$, $\pm z_-(a,b)$. We only need to calculate the mass of the two points $z_\pm(a,b)$ given in
(\ref{E:z+-ab}) because the mass of the opposite points follow from Corollary \ref{C:sym}. We will make the calculations for a general point of the
form
\begin{equation} \label{E:z-gen}
z_0 = \frac{1-\overline{a}\zeta_0}{|1-\overline{a}\zeta_0|}, \qquad \zeta_0\in\Sigma_a.
\end{equation}
The mass of $z_\pm(a,b)$ is obtained setting $\zeta_0=\zeta_\pm(b)$.

Proposition \ref{P:antidiag} states that
\begin{equation} \label{E:mu-z0-Z}
\hat{\bs{\mu}}(\{z_0\}) = m(z_0) \begin{pmatrix} 1 & \eta(z_0) \\ \overline{\eta(z_0)} & 1 \end{pmatrix},
\qquad
\begin{aligned}
& \eta(z_0) = \omega z_0f_a(z_0),
\\
& m(z_0) = 1/z_0g'_{a,b}(z_0).
\end{aligned}
\end{equation}
In this case
\begin{equation} \label{E:eta-z0-Z}
\eta(z_0)=-\omega\overline{z_0}\zeta_0 = -\omega\frac{\zeta_0-a}{|\zeta_0-a|},
\end{equation}
because we know that inverting (\ref{E:z-gen}) yields $\zeta_0=-z_0^2f_a(z_0)$.

For the calculation of $g'_{a,b}(z_0)$, first perform the change $\zeta(z)=-z^2f_a(z)$ in $g_{a,b}(z)$,
$$
g_{a,b} = \zeta\frac{\zeta-b}{1-\overline{b}\zeta},
$$
so that
$$
\frac{g'_{a,b}}{g_{a,b}} = \frac{\zeta'}{\zeta} \left( 1 + \frac{\rho_b^2\zeta}{(1-\overline{b}\zeta)(\zeta-b)} \right),
$$
and
$$
g'_{a,b}(z_0) = \frac{\zeta'(z_0)}{\zeta_0} \left( 1 + \frac{\rho_b^2}{|\zeta_0-b|^2} \right)
= 2\frac{\zeta'(z_0)}{\zeta_0} \frac{1-\re(\overline{b}\zeta_0)}{|\zeta_0-b|^2}.
$$

It only remains to compute $\zeta'(z_0)$. From (\ref{E:fa}) we obtain
$$
\zeta'(z) = 2z\frac{\zeta(z)-a}{\sqrt{\Delta_a(z)}},
$$
hence
\begin{equation} \label{E:der-zeta}
\frac{\zeta'(z_0)}{\zeta_0} = \frac{1-a\overline{\zeta}_0}{\sqrt{|a|^2-\im^2 z_0}} =
\frac{|1-\overline{a}\zeta_0|}{|a|^2-\re(\overline{a}\zeta_0)}(1-a\overline{\zeta}_0).
\end{equation}

Combining the previous results we get
\begin{equation} \label{E:m-z0-Z}
m(z_0) = \frac{1}{2} \frac{|\zeta_0-b|^2}{|\zeta_0-a|^2} \frac{|a|^2-\re(\overline{a}\zeta_0)}{1-\re(\overline{b}\zeta_0)}
= \frac{1}{2} \frac{1-\frac{\rho_a^2}{|\zeta_0-a|^2}}{1+\frac{\rho_b^2}{|\zeta_0-b|^2}}.
\end{equation}

\subsection{Asymptotics of $p_{\alpha,\beta}^{(0)}(n)$ on $\Z$} \label{SS:ARP-0-Z}

The asymptotic return probability to the origin involves $\hat{\bs{\psi}}_{\alpha,\beta}^{(0)}=(\alpha,0)\hat{\bs{X}}_0+(0,\beta)\hat{\bs{X}}_1$. We
know that $\hat{\bs{X}}_0=1$, while $\hat{\bs{X}}_1$ can be obtained specializing the general expression (\ref{E:X1-Z}) for the coins
$$
\hat{C}_{-1} = \begin{pmatrix} \rho_a & \overline{\omega a} \\ -\omega a & \rho_a \end{pmatrix},
\qquad
\hat{C}_0 = \begin{pmatrix} \rho_b & -\overline\omega b \\ \omega\overline{b} & \rho_b \end{pmatrix},
\qquad
\rho_b=\sqrt{1-|b|^2},
$$
related to the CMV matrix $\hat{\bs{\cC}}$ of $\hat{\bs{\mu}}$. This gives
$$
\hat{\bs{X}}_1(z) = \begin{pmatrix} z^{-1}/\rho_a & -\omega a/\rho_a \\ -\overline\omega b/\rho_b & z^{-1}/\rho_b \end{pmatrix}.
$$
We finally find that
$$
\hat{\bs{\psi}}_{\alpha,\beta}^{(0)}(z) = \alpha(1,0) + \frac{\beta}{\rho_b}(-\overline\omega b,z^{-1}).
$$

\subsubsection{The case of 2 mass points} \label{SSS:ARP-0-Z-2mass}

Assume the case $\mathbf{M}^2\equiv(\mathbf{M}^2_+\text{ or }\mathbf{M}^2_-)$, so that there are exactly 2 mass points: $\pm z_+(a,b)$ for
$\mathbf{M}^2_+$, and $\pm z_-(a,b)$ for $\mathbf{M}^2_-$. For convenience, let us write $z_0=z_\pm(a,b)$ and $\zeta_0=\zeta(z_0)$ for
$\mathbf{M}^2_\pm$. Then, from (\ref{E:mu-z0-Z}) and (\ref{E:eta-z0-Z}) we find that
$$
\begin{aligned}
& \hat{\bs{\psi}}_{\alpha,\beta}^{(0)}(z_0) \hat{\bs{\mu}}(\{z_0\}) \hat{\bs{\psi}}_{1,0}^{(0)}(z_0)^\dag
= m(z_0) \left( \alpha - \beta \overline{\omega} \frac{\overline{\zeta}_0+b}{\rho_b} \right),
\\
& \hat{\bs{\psi}}_{\alpha,\beta}^{(0)}(z_0) \hat{\bs{\mu}}(\{z_0\}) \hat{\bs{\psi}}_{1,0}^{(0)}(z_0)^\dag
= - m(z_0) \, \omega \frac{\zeta_0+\overline{b}}{\rho_b} \left( \alpha - \beta \overline{\omega} \frac{\overline{\zeta}_0+b}{\rho_b} \right),
\end{aligned}
$$
where $m(z_0)$ is given in (\ref{E:m-z0-Z}).

The result for $-z_0$ is the same because $m(z)$ is even and $\eta(z)$ is odd, thus, according to (\ref{E:ARP-hat}),
$$
\lim_{n\to\infty} p_{\alpha,\beta}^{(0)}(2n) =
(2m(z_0))^2 \left( 1+\frac{|\zeta_0+\overline{b}|^2}{\rho_b^2} \right)
\left| \hat\alpha - \hat\beta \overline{\omega} \frac{\overline{\zeta}_0+b}{\rho_b} \right|^2.
$$

Setting $z_0=z_\pm(a,b)$, then $\zeta_0=\zeta_\pm(b)$ and
$$
1+\frac{\rho_b^2}{|\zeta_0-b|^2} = 1+\frac{|\zeta_0+\overline{b}|^2}{\rho_b^2} = \frac{2\sqrt{1-\im^2b}}{\sqrt{1-\im^2b}\mp\re b},
$$
which gives
$$
\lim_{n\to\infty} p_{\alpha,\beta}^{(0)}(2n) = p_{\alpha,\beta}^\pm(a,b,\omega) \quad \text{ for } \; \mathbf{M}^2_\pm,
$$
with
\begin{equation} \label{E:p+-}
\begin{aligned}
& \kern-6pt p_{\alpha,\beta}^\pm(a,b,\omega) = \frac{\left(1-\frac{\rho_a^2}{|\zeta_0-a|^2}\right)^2}{1+\frac{\rho_b^2}{|\zeta_0-b|^2}}
\left| \hat\alpha - \hat\beta \overline{\omega} \frac{\overline{\zeta}_0+b}{\rho_b} \right|^2
\\
& \kern-9pt = \frac{1}{2} \left(1-\frac{\rho_a^2}{|\zeta_\pm(b)-a|^2}\right)^2
\left\{ 1 \mp \frac{(|\hat\alpha|^2-|\hat\beta|^2) \re b + 2\rho_b \re(\overline{\omega\hat\alpha}\hat\beta)}{\sqrt{1-\im^2b}} \right\}\kern-2pt.
\end{aligned}
\end{equation}
Here we have used that $|\alpha|^2+|\beta|^2=1$.

We see that the asymptotic return probability $p_{\alpha,\beta}^\pm(a,b,\omega)$ to the origin depends on the coefficients $\alpha$ and $\beta$ of
the state, as well as on the parameters $a$, $b$, $\omega$ associated with the QW. Indeed, there is a state $\alpha|0\!\ua\>+\beta|0\!\da\>$ which
exhibits no localization, given by
\begin{equation} \label{E:state-noloc-2mass}
\hat\beta = \hat\alpha \, \omega \frac{\rho_b}{\re b \pm \sqrt{1-\im^2b}} \quad \text{ for } \; \mathbf{M}^2_\pm.
\end{equation}

\subsubsection{The case of 4 mass points} \label{SSS:ARP-0-Z-4mass}

In the case $\mathbf{M}^4$ there are 4 mass points: $\pm z_+$, where $z_+=z_+(a,b)$ is related to $\zeta_+=\zeta_+(b)$, and $\pm z_-$, where
$z_-=z_-(a,b)$ is related to $\zeta_-=\zeta_-(b)$. Therefore
$$
\begin{aligned}
p_{\alpha,\beta}^{(0)}(2n) & \underset{n}{\sim}
\left| \sum_{z=z_\pm} 2m(z) \left( \hat\alpha - \hat\beta \overline{\omega} \frac{\overline{\zeta}+b}{\rho_b} \right) z^{2n} \right|^2
\\
& + \left| \sum_{z=z_\pm} 2m(z) \frac{\zeta+\overline{b}}{\rho_b}
\left( \hat\alpha - \hat\beta \overline{\omega} \frac{\overline{\zeta}+b}{\rho_b} \right) z^{2n} \right|^2.
\end{aligned}
$$
The cross terms of both summands cancel each other because
$$
\frac{\overline{\zeta}_++b}{\rho_b} \frac{\zeta_-+\overline{b}}{\rho_b} = -1,
$$
hence,
$$
\lim_{n\to\infty} p_{\alpha,\beta}^{(0)}(2n) = p_{\alpha,\beta}^+(a,b,\omega) + p_{\alpha,\beta}^-(a,b,\omega) \quad \text{ for } \; \mathbf{M}^4,
$$
with $p_{\alpha,\beta}^\pm(a,b,\omega)$ given in (\ref{E:p+-}).

In other words, the 4 mass points $\pm z_+$, $\pm z_-$ contribute to the asymptotic return probability to the origin simply by adding the
contributions that they should have if considered as two independent cases with 2 mass points. As a consequence, the existence of 4 mass points
ensures that all the states at the origin exhibit localization because the two conditions in (\ref{E:state-noloc-2mass}) are incompatible for
$(\alpha,\beta)\neq(0,0)$.

Particularly simple is the case of an imaginary value of $a$ which, according to (\ref{E:ab-Z}), corresponds to a defect such that
$e^{i\tau}=e^{i\sigma}$. Then, localization appears if and only if $\im a>0,\im b$ or $\im a<0,\im b$, and in such a case there exist always 4 mass
points. The simplicity of the asymptotic return probability to the origin comes from the fact that, for an imaginary $a$, we have
$|\zeta_+-a|=|\zeta_--a|$ and thus
\begin{equation} \label{E:p-aim}
\lim_{n\to\infty} p_{\alpha,\beta}^{(0)}(2n) = \left(1-\frac{\rho_a^2}{|\zeta_\pm-a|^2}\right)^2
= \left(\frac{2\im a(\im a-\im b)}{1+\im^2a-2\,\im a\,\im b}\right)^2.
\end{equation}
In this case the asymptotic return probability to the origin does not depend on the state.

For instance, the model (\ref{E:Konno2}) gives $a=\frac{i}{\sqrt{2}}$, $b=\frac{ie^{i\phi}}{\sqrt{2}}$, $\omega=1$, which exhibits localization when
$\im b < \im a$, i.e., $e^{i\phi}\neq1$. The application of (\ref{E:p-aim}) to this model yields
$$
\lim_{n\to\infty} p_{\alpha,\beta}^{(0)}(2n) = \left(\frac{2(1-\cos\phi)}{3-2\cos\phi}\right)^2,
$$
which shows that the result obtained in \cite{K} for the special case $\alpha=\frac{1}{\sqrt{2}}$, $\beta=\frac{i}{\sqrt{2}}$ is indeed true for any
$\alpha,\beta$.

\subsection{Maximum asymptotic return probabilities on $\Z$} \label{SS:max-ARP-0-Z}

The previous results seem to indicate that the maximum values of
$$
p_{\alpha,\beta}^{(0)}=\lim_{n\to\infty}p_{\alpha,\beta}^{(0)}(2n)
$$
should be reached for $a$ close to the unit circle, which means that the non defective coin $C$ has an almost anti-diagonal shape. Let us analyze the
behaviour of $\max_{\alpha,\beta}p_{\alpha,\beta}^{(0)}$ when $|a|\to1$.

Given a value of $a$ with a fixed phase different from that of $\zeta_\pm(b)$, the localization pictures in subsection \ref{SS:Loc-ab-Z} show that
the measure $\bs{\mu}_{a,b}^\omega$ has 4 mass points as far as $|a|$ is close enough to 1. Take $\alpha_0$, $\beta_0$ such that
$\hat\beta_0=i\omega\hat\alpha_0$, so $|\alpha_0|=|\beta_0|=\frac{1}{\sqrt{2}}$. Then, we find from (\ref{E:p+-}) that, in the case of 4 mass points,
$$
\max_{\alpha,\beta} p_{\alpha,\beta}^{(0)} \geq p_{\alpha_0,\beta_0}^{(0)} =
\frac{1}{2} \left( 1-\frac{\rho_a^2}{|\zeta_+(b)-a|^2}\right) +  \frac{1}{2} \left( 1-\frac{\rho_a^2}{|\zeta_-(b)-a|^2}\right).
$$
Therefore,
$$
\lim_{a \to a_0} \max_{\alpha,\beta}p_{\alpha,\beta}^{(0)} = 1, \qquad a_0\in\T\setminus\{\zeta_\pm(b)\}.
$$
That is, if $\im a \neq \im b$ and $|a|$ is close enough to one, there exist qubits which asymptotically return to the origin with probability almost
one. According to (\ref{E:ab-Z}), given a defect $D$, this holds for almost any coin $C$ as long as its diagonal is close enough to zero.

These results become stronger when $a$ is imaginary. Then $p_{\alpha,\beta}^{(0)}$ is independent of $\alpha,\beta$ and (\ref{E:p-aim}) yields
$\lim_{a \to \pm i} p_{\alpha,\beta}^{(0)} = 1$ for any state. In other words, if $a$ is close enough to $i$ or $-i$, all the qubits asymptotically
return to the origin with probability almost one. Looking at (\ref{E:ab-Z}) we see that, given a defect $D$, this is the case of any coin $C$ which
is close enough to an anti-diagonal one provided that $e^{i\sigma}$ is close enough to $e^{i\tau}$.

\section{Localization: one defect on $\Z_+$} \label{S:Loc-defect-Z+}

We will study the localization for the coins (\ref{E:C-D}) in $\Z_+$. As in the case of $\Z$, this requires the analysis of the mass points of the
corresponding measure. Subsection \ref{SS:defect-Z+} shows that these models fall again into groups with the same localization behaviour because any
such a group has a unique measure up to rotations. Nevertheless, these groups are characterized now by only two parameters $a,b\in\D$ given in
(\ref{E:ab-Z+}). The measure $\hat{\mu}=\mu_{a,b}$ of the CMV matrix $\hat{\cC}=\cC(\hat{\alpha}_k)$ introduced in Subsection \ref{SS:defect-Z+}
serves as a canonical representative for the measures in a group. The corresponding weight and mass points are supported in $\T\setminus\Gamma_a$ and
$\Gamma_a$ respectively.

\subsection{Mass points of $\mu_{a,b}$} \label{SS:roots-Z+}

Concerning localization properties for one defect on $\Z_+$ we can restrict ourselves to the measures $\mu_{a,b}$ without loss of generality. The
corresponding mass points are the roots $z\in\T$ of $h_{a,b}(z)=zf_{a,b}(z)=1$ such that
\begin{equation} \label{E:muab}
\mu_{a,b}(\{z\}) = \lim_{r\ua1} \frac{1-r}{1-h_{a,b}(rz)} \neq 0.
\end{equation}
These roots must lie on $\Gamma_a$ and when they lie on $\Gamma_a^0$ condition (\ref{E:muab}) is always satisfied because $h_{a,b}$ is analytic in
$\Gamma_a^0$.

Moreover, the points of $\partial\Gamma_a$ can be roots of $h_{a,b}(z)=1$ but never mass points of $\mu_{a,b}$ because condition (\ref{E:muab}) is
not satisfied on $\partial\Gamma_a$. Consider for instance the point $z_a$ and assume that $h_{a,b}(z_a)=1$. Then, using (\ref{E:z2fa-exp-r}) we find
that
$$
h(rz_a)-1 = K'\sqrt{1-r}+O(1-r), \qquad K'\neq0,
$$
which, according to (\ref{E:muab}), implies that $\mu_{a,b}(\{z_a\})=0$. A similar proof works for the remaining points of $\partial\Gamma_a$.

Therefore, the mass points of $\mu_{a,b}$ are exactly the roots of $h_{a,b}(z)=1$ in $\Gamma_a^0$. To study these roots we will use the same chage of
variables as in the case of $\Z$. However, since the symmetry of the mass points with respect to the origin disappears in $\Z_+$, we must study
independently the roots in the right and left arcs $\Gamma_a^\pm$ of $\Gamma_a^0$.

The transformation
$$
\zeta = \zeta(z) = -z^2f_a(z)
$$
maps both arcs $\Gamma_a^\pm$ onto $\Sigma_a$ (see figure \ref{F:z-zeta}), with and inverse mapping given respectively by
$$
z = \pm z(\zeta)\in\Gamma_a^\pm, \qquad z(\zeta) = \frac{1-\overline{a}\zeta}{|1-\overline{a}\zeta|}.
$$
Therefore, the equation for $z\in\Gamma_a^\pm$ reads in terms of $\zeta\in\Sigma_a$ as
$$
\begin{gathered}
h_{a,b}(z)=1, \quad z\in\Gamma_a^\pm
\quad \Leftrightarrow \quad
\mp \frac{1-\overline{a}\zeta}{|1-\overline{a}\zeta|} \frac{\zeta-b}{1-\overline{b}\zeta} = 1, \quad \zeta\in\Sigma_a
\\
\Leftrightarrow \quad \frac{(\zeta-b)^2}{|\zeta-b|^2} = \mp \frac{\zeta-a}{|\zeta-a|}, \quad \zeta\in\Sigma_a
\quad \Leftrightarrow \quad \frac{(\zeta-b)^2}{\zeta-a}\in\R^\mp, \quad \zeta\in\Sigma_a.
\end{gathered}
$$

The last of the above equivalent conditions states that $b$ lies on a straight line passing through $\zeta$ in the direction given by
$i\sqrt{\zeta-a}$ or $\sqrt{\zeta-a}$ respectively. The first case is equivalent to the presence of a mass point at $z(\zeta)\in\Gamma_a^+$, while
the second case means that there is a mass point at $-z(\zeta)\in\Gamma_a^-$.

In other words, any $a\in\D$ defines two orthogonal one-parameter families of straight lines (see figure \ref{F:haz}): those $b_{a,\zeta}^-$ passing
through each $\zeta\in\Sigma_a$ in the direction of $\sqrt{\zeta-a}$, and those $b_{a,\zeta}^+$ passing through each $\zeta\in\Sigma_a$ in the
orthogonal direction $i\sqrt{\zeta-a}$. The points of $\D$ swept by the family $\{b_{a,\zeta}^\pm\}_{\zeta\in\Sigma_a}$ are the values of $b$ giving
mass points for $\mu_{a,b}$ at $\Gamma_a^\pm$ respectively. Hence, the values $b\in\D$ which yield mass points for $\mu_{a,b}$ are those swept by
$\{b_{a,\zeta}^+\}_{\zeta\in\Sigma_a}\cup\{b_{a,\zeta}^-\}_{\zeta\in\Sigma_a}$. Given $b$, the number of mass points of $\mu_{a,b}$ is equal to the
number of straight lines of both families $\{b_{a,\zeta}^\pm\}_{\zeta\in\Sigma_a}$ which pass through $b$. Moreover, each line $b_{a,\zeta}^\pm$
passing through $b$ provides the corresponding mass point $z = \pm z(\zeta)$ because it crosses the unit circle at $\zeta$ and $\overline{z}$. This
follows from
$$
w\in\T \quad \Rightarrow \quad \frac{(\zeta-w)^2}{\zeta-a} = -\frac{|\zeta-w|^2}{|\zeta-a|}wz(\zeta),
$$
which implies that
$$
w\in\T\setminus\{\zeta\}, \quad \frac{(\zeta-w)^2}{\zeta-a}\in\R^\mp \quad \Leftrightarrow \quad w = \pm \overline{z(\zeta)}.
$$

\subsection{The envelopes of $\{b_{a,\zeta}^\pm\}_{\zeta\in\Sigma_a}$} \label{SS:Env}

\begin{figure}
\includegraphics[width=8cm]{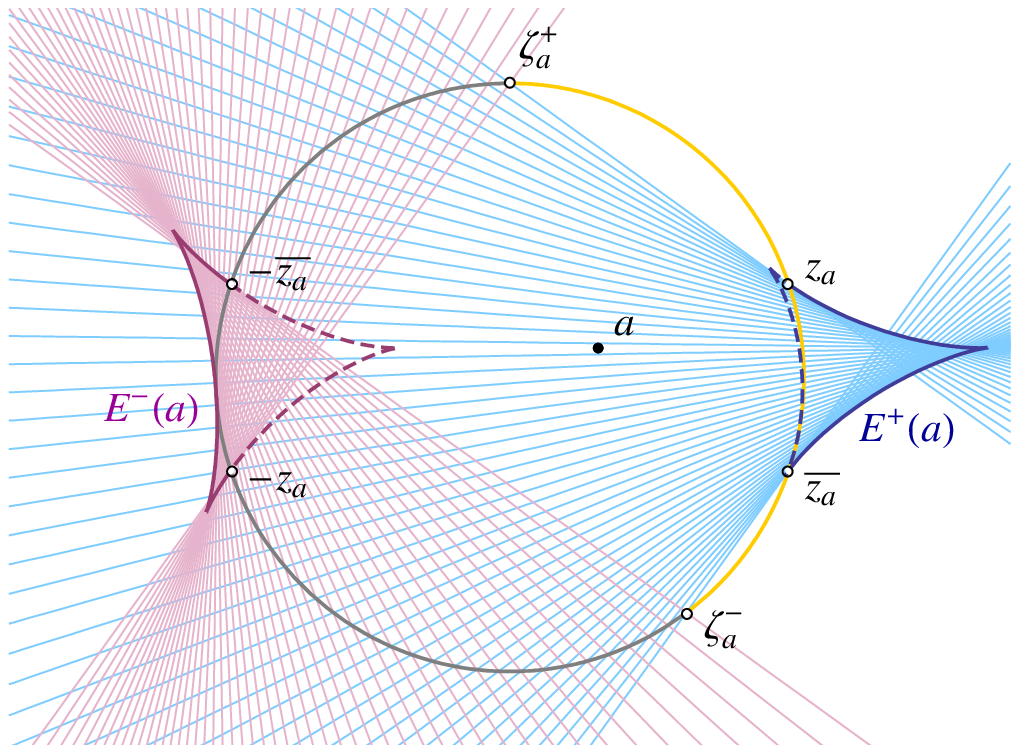}

\vspace{5pt}

\includegraphics[width=8cm]{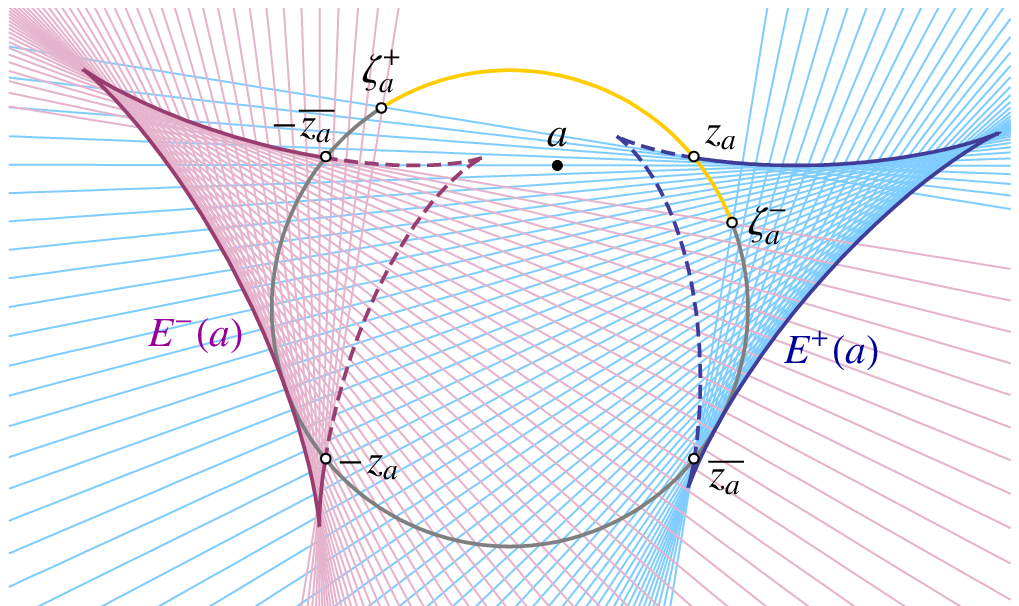}

\includegraphics[width=8cm]{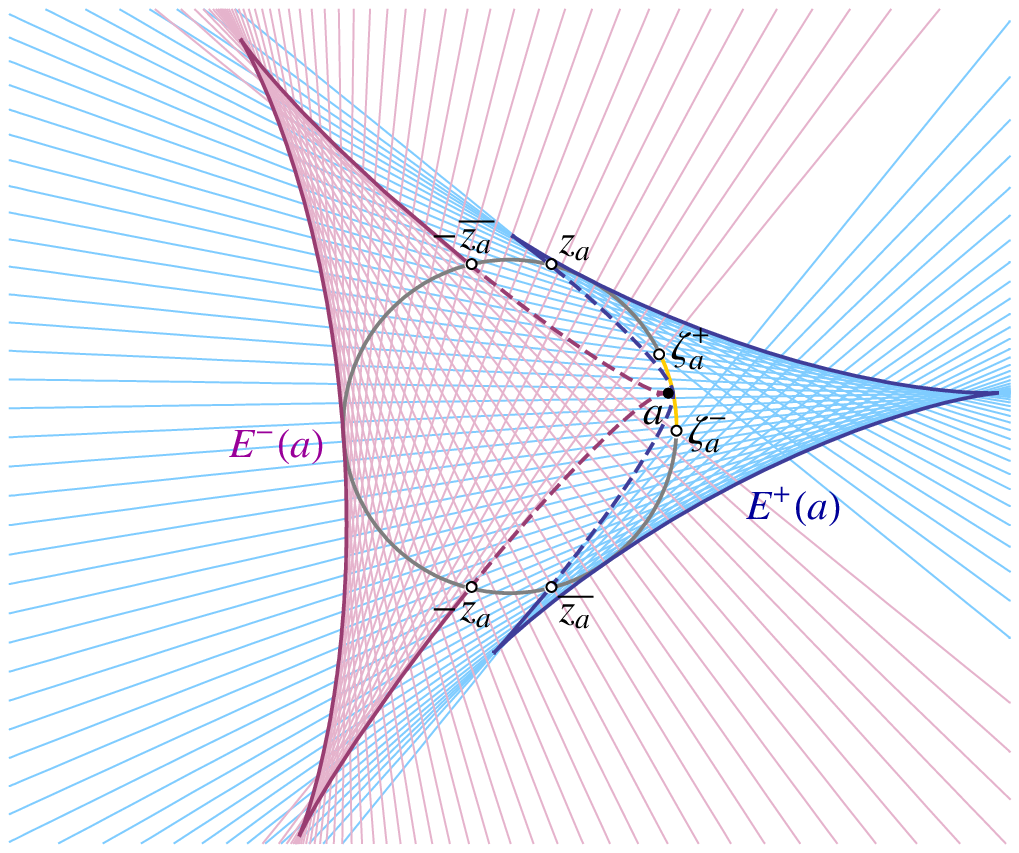}
\caption{For different choices of $a$, the families of straight lines $\{b_{a,\zeta}^+\}_{\zeta\in\Sigma_a}$ and
$\{b_{a,\zeta}^-\}_{\zeta\in\Sigma_a}$ in blue and purple color respectively. The corresponding envelopes $E_e^\pm(a)$ are the continuous curves in
the same dark color lying on the exterior of $\T$. The arc $\Sigma_a$ appears in grey color and $\T\setminus\Sigma_a$ in yellow. The dashed curves in
the interior of $\T$ are the envelopes $E_i^\pm(a)$ of the families $\{b_{a,\zeta}^\pm\}_{\zeta\in\T\setminus\Sigma_a}$ (which are not depicted
here). The tangent points to $\T$ split $E_e(a)$ into 3, 4 and 5 connected components respectively from the upper to the lower figure. The straight
lines corresponding to such components sweep different sectors of $\D$. The subset of $\D$ swept by any of these sectors is the region of values of
$b$ giving localization for the choice of $a$. When several sectors overlap, the corresponding values of $b$ yield as many mass points as overlapping
sectors cover $b$.} \label{F:haz}
\end{figure}

The envelopes of the two families $\{b_{a,\zeta}^\pm\}_{\zeta\in\Sigma_a}$ of straight lines can help us to determine the points of $\D$ swept by
them. For the computation of the envelopes it is convenient to rewrite the equations for the families in a different way. Denoting $A=\zeta-a$ and
$B=\zeta-b$,
$$
\frac{(\zeta-b)^2}{\zeta-a}\in\R^\mp \; \Leftrightarrow \; B \, \| \, \sqrt{\mp A} \; \Leftrightarrow \; B \, \bot \, \sqrt{\pm A} \; \Leftrightarrow
\; B \, \bot \, (|A| \pm A).
$$
Therefore, the equation for the family $b_{a,\zeta}^\pm$ can be written as
$$
\re[(A\pm|A|)\overline{B}]=0.
$$

Remember that $\zeta\in\Sigma_a$ is given by $\zeta=\frac{a}{|a|}e^{it}$, $\cos t <|a|$, so $A$ and $B$ can be considered functions of $t$ which
parametrizes the lines of the two families. Then, the envelope of each family is given parametrically with respect to $t$ by the equation of the
family together with its derivative with respect to $t$. This leads to the equations
\begin{equation} \label{E:env1}
\re[X_\pm \overline{B}]=0, \qquad \im[Y_\pm B + X_\pm \overline{\zeta}]=0,
\end{equation}
for the envelopes of $b_{a,\zeta}^\pm$ respectively, where
$$
X_\pm = A\pm|A|, \qquad Y_\pm = i\frac{d}{dt}\overline{X}_\pm = \overline{\zeta} \pm i\frac{\im(\overline{a}\zeta)}{|A|},
$$
only depend on $a$ and $t$, but not on $b$.

The system (\ref{E:env1}) can be solved in $B(t)$, thus in $b(t)=\zeta(t)-B(t)$, giving the envelopes $b^\pm_a(t)$ of the two families
$b_{a,\zeta}^\pm$,
\begin{equation} \label{E:env2}
b^\pm_a = \zeta + i \frac{\im(X_\pm\overline{\zeta})}{\re(X_\pm Y_\pm)}X_\pm.
\end{equation}

When we let $t\in[0,2\pi]$, then $\zeta$ runs over the whole unit circle and the two envelopes $b^\pm_a(t)$ obviously describe a closed curve because
$t$ enters in $b^\pm_a(t)$ only through $\zeta$. We will refer to
$$
E^\pm(a)=\{b_a^\pm(t):t\in[0,2\pi]\}
$$
as the full envelopes, to distinguish them from the original ones
$$
E^\pm_e(a)=\{b_a^\pm(t):\cos t<|a|\},
$$
in which $\zeta$ runs over $\Sigma_a$. The closure $\overline{E}^\pm_e(a)=\{b_a^\pm(t):\cos t\leq|a|\}$ allows $\zeta$ to run over the closed arc
$\overline\Sigma_a$, that is, it only adds to $E^\pm_e(a)$ the two limit points in $\partial E^\pm_e(a) = \{b_a^\pm(t):\cos t=|a|\}$. Apart from
being useful in some reasonings, the rest of the envelope $E^\pm_i(a)=\{b_a^\pm(t):\cos t>|a|\}$ has no interest for us because it comes from points
$\zeta\in\T\setminus\overline{\Sigma}_a$. When referring to the set of two $\pm$ envelopes we will use the notation $E(a) = E^+(a) \cup E^-(a)$,
$E_e(a) = E^+_e(a) \cup E^-_e(a)$ and so forth.

\smallskip

The following properties of the envelopes follow from (\ref{E:env2}) (see figure \ref{F:haz}):

\begin{itemize}

\item The full envelopes $E^\pm(a)$ are deformed deltoides, i.e., deformed triangles with concave curve sides joining at three cusps. Also,
$E^+(a) \cap E^-(a) = \emptyset$.

\item $E_e(a)\subset\C\setminus\D$, $E_i(a)\subset\overline{\D}$ and $\partial E_e(a)\subset\T$. Indeed,
$\partial E^\pm_e(a) = \partial\Gamma_a^\pm$. Hence, we will call $E_e(a)$ the exterior envelope, $E_i(a)$ the interior envelope and $\partial
E_e(a)$ the limit points of the envelope.

\item Two contiguous cusps lie on the closed exterior envelope $\overline{E}_e(a)$ if and only if there is a tangent point to $\T$ in the side
joining such cusps. In particular, a cusp lies on $\T$ if and only if it is a tangent point to $\T$, and this is also equivalent to stating that the
cusp is a limit point of the envelope.

\end{itemize}

Consider a given value of $a$.

If the 3 cusps of one of the two full envelopes $E^\pm(a)$ lie on the corresponding closed exterior $\overline{E}^\pm_e(a)$, then
$\overline{E}^\pm_e(a)$ becomes tangent to $\T$ at 2 points (see the lower image in figure \ref{F:haz}). In such a case the straight lines
corresponding to the points of $E^\pm_e(a)$ between the two tangent points sweep the whole unit disk $\D$ and, thus, any value $b\in\D$ gives a mass
point for $\mu_{a,b}$ .

In general, if $T(a)$ are the tangent points to $\T$ of the exterior envelope $E_e(a)$, then $E_e(a) \setminus T(a)$ splits into connected components
with a single cusp inside each component (see figure \ref{F:haz}). The straight lines corresponding to a given connected component do not intersect
in $\D$, so a value $b\in\D$ yields as many mass points as connected components have a straight line passing through $b$. The region of $\D$ swept by
the straight lines corresponding to a connected component is the open sector limited by the tangents to the two extreme points of the connected
component. If one of the extreme points of a connected components is tangent to $\T$ and the other one is a limit point of the envelope, the only
straight line which limits the related sector of $\D$ is that one tangent to the envelope at the limit point in question. If both extremes of the
connected component are limit points, then the related sector is limited by the two straight lines associated with such limit points. If a limit
point is simultaneously a tangent point to $\T$ of the closed exterior envelope $\overline{E}_e(a)$, then the corresponding straight line does not
provide any restriction for the related sector of values $b\in\D$.

In other words, each connected component of $E_e(a) \setminus T(a)$ has an associated sector of values $b\in\D$ to which it gives a mass point for
$\mu_{a,b}$. The elements deciding the sectors are the tangent points to $\T$ of $\overline{E}_e(a)$ and the straight lines corresponding to the
limit points of the envelope, which we will call the limit lines of the envelope. Let us have a closer look at such elements.

\subsubsection{Limit lines of the envelopes} \label{SSS:LimitLines}

The limit points $\partial E_e(a)$ of the envelopes correspond to setting $\zeta=\zeta_a^\pm=\frac{a}{|a|}(|a| \pm i\rho_a)$. The related straight
lines $b^-_{a,\zeta_a^\pm}$, $b^+_{a,\zeta_a^\pm}$ pass through such points in the directions $\sqrt{\zeta-a}$ and $i\sqrt{\zeta-a}$ respectively,
which are the orthogonal directions $\sqrt{\pm ia}$ and $\sqrt{\mp ia}$ respectively (notice that $\sqrt{\pm ia}$ points in the same direction as $1
\pm i\frac{a}{|a|}$).

More precisely, the limit lines of $E^+_e(a)$ are the orthogonal lines $b_{a,\zeta_a^\pm}^+$ and the limit lines of $E^-_e(a)$ are the orthogonal
lines $b_{a,\zeta_a^\mp}^-$, which are parallel to the previous ones. Thus the limit lines of the envelope are two pairs of parallel lines which are
orthogonal between themselves.

On the other hand, we know that $b_{a,\zeta}^\pm$ crosses the unit circle at $\zeta$ and $\pm\overline{z(\zeta)}$. Therefore, the limit lines are:
\begin{itemize}
\item $b_{a,\zeta_a^+}^\pm$, joining $\zeta_a^+$ to $\pm \overline{z(\zeta_a^+)} = \pm z_a$ respectively.
\item $b_{a,\zeta_a^-}^\pm$, joining $\zeta_a^-$ to $\pm \overline{z(\zeta_a^-)} = \pm\overline{z}_a$ respectively.
\end{itemize}

\subsubsection{Tangent points to $\T$ of the full envelopes} \label{SSS:TangentFull}

If the full envelope $E^\pm(a)$ is tangent to $\T$ at a point $\zeta$, then $b_{a,\zeta}^\pm$ passes through $\zeta$ in the direction orthogonal to
$\zeta$. This means that $\zeta \, \bot \, i\sqrt{\zeta-a}$ or $\zeta \, \bot \, \sqrt{\zeta-a}$, which is equivalent to $\zeta \, \parallel \,
\sqrt{\zeta-a}$ or $\zeta \, \parallel \, i\sqrt{\zeta-a}$. Thus, the tangency condition can be expressed as
\begin{equation} \label{E:azeta}
\frac{\zeta^2}{\zeta-a}\in\R.
\end{equation}
This condition is satisfied by the values of $a$ lying on a straight line $a_\zeta$ passing through $\zeta$ in the direction $\zeta^2$. The line
$a_\zeta$ picks up the values of $a$ with $E(a)$ having $\zeta$ as a common tangent point to $\T$.

The envelope of the family of lines $\{a_\zeta\}_{\zeta\in\T}$ will help us in counting the number of tangent points to $\T$ of $E(a)$ for any value
of $a$. Setting $\zeta=e^{it}$, $t\in[0,2\pi]$, this envelope is given by equation (\ref{E:azeta}) together with its derivative with respect to $t$,
$$
\im(\zeta-\overline{a}\zeta^2)=0, \qquad \re(\zeta-2\overline{a}\zeta^2)=0.
$$
The solution
$$
\ts a(t) = (\re(\overline{a}\zeta^2)-i\im(\overline{a}\zeta^2))\zeta^2 = (\frac{1}{2}\re\zeta - i\im\zeta) \zeta^2
= \frac{3}{4}e^{it}-\frac{1}{4}e^{3it}
$$
is the envelope of $\{a_\zeta\}_{\zeta\in\T}$, which is an epicycloid inscribed in the unit circle with two cusps at the points $\pm\frac{1}{2}$
(see figure \ref{F:cacahuete}).

\begin{figure}
\includegraphics[width=6.5cm]{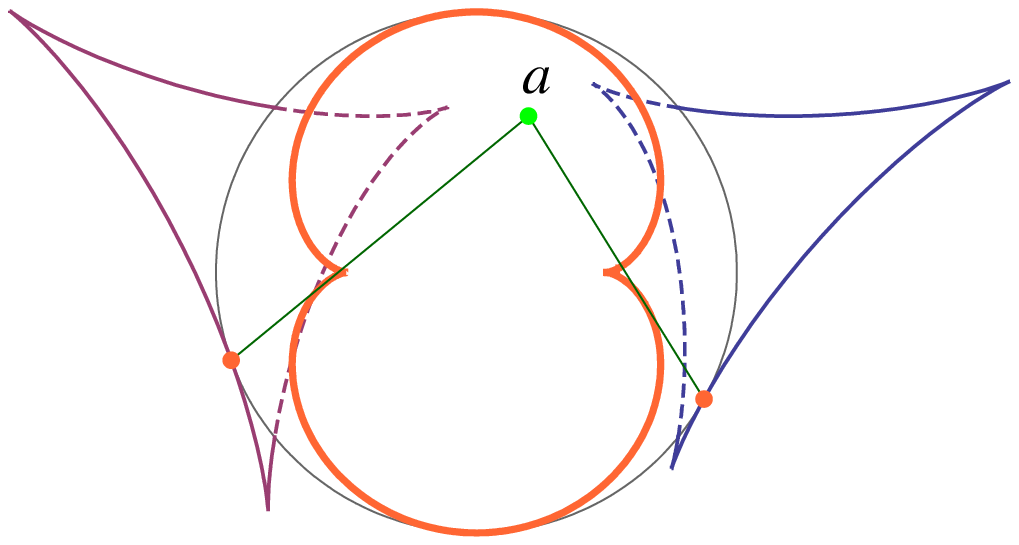} \includegraphics[width=6cm]{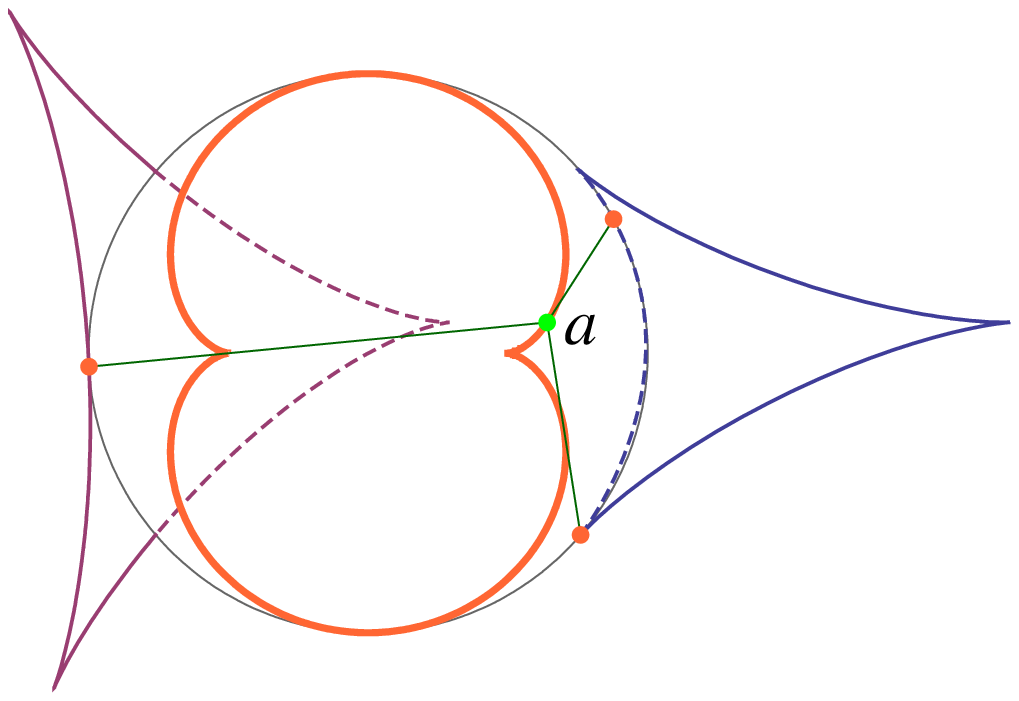} \includegraphics[width=7cm]{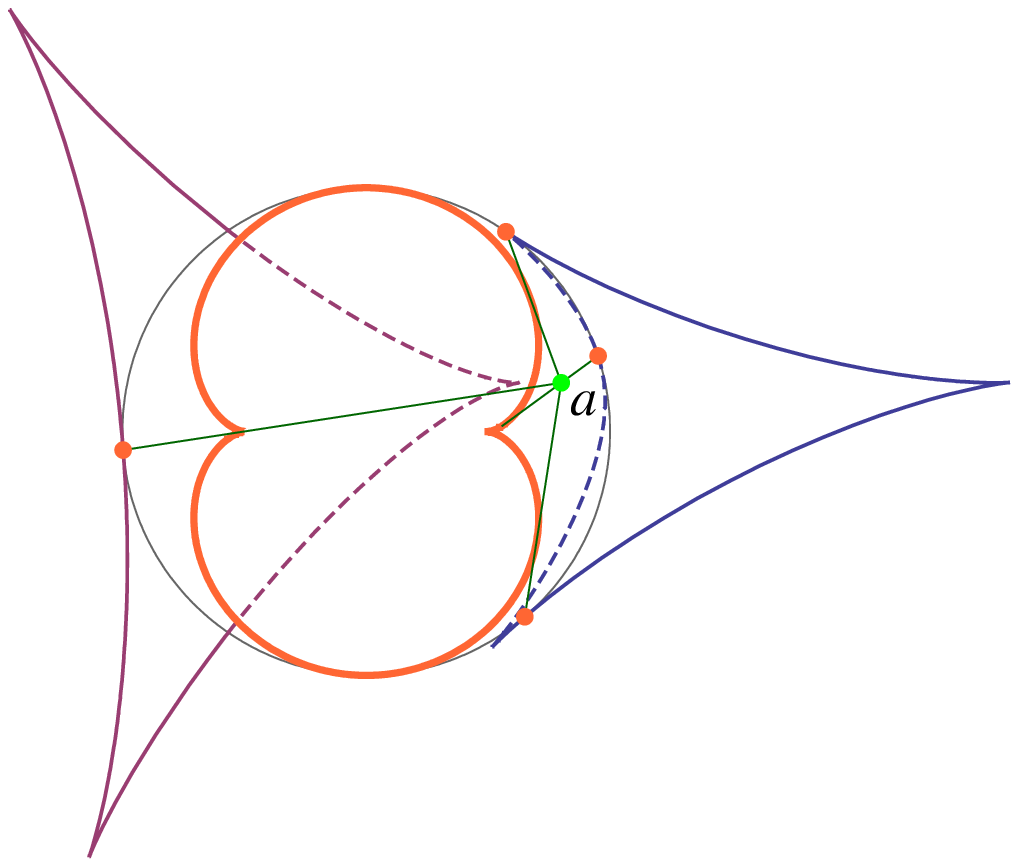}
\caption{The curve in orange is the epicycloid which delimits the values of $a$ with a different number of tangent points of $E(a)$ to $\T$. The
cusps are at $\pm\frac{1}{2}$. The tangents to the epicycloid passing through a given $a$ cross the unit circle exactly at the tangent points of
$E(a)$ to $\T$. Due to the shape of the epicycloid there are three possibilities according to the relative position of $a$ and the epicycloid. These
possibilities are shown in the figures above, where the orange points are the tangencies between $E(a)$ and $\T$.} \label{F:cacahuete}
\end{figure}

Given $a\in\D$, the number of tangents to the epycicloid passing through $a$ counts the number of points in $E(a)$ which are tangent to $\T$.
Therefore, due to the shape of the epycicloid we have the following possibilities for $a\in\D$ (see figure \ref{F:cacahuete}):
\begin{itemize}
\item[($\mathbf{T}^2$)] If $a$ lies inside the epycicloid, $E(a)$ has 2 tangent points to $\T$.
\item[($\mathbf{T}^3$)] If $a$ lies on the epycicloid, $E(a)$ has 3 tangent points to $\T$, except at the cusps $a=\pm\frac{1}{2}$,
where $E(a)$ has 2 tangent points to $\T$.
\item[($\mathbf{T}^4$)] If $a$ lies outside the epycicloid, $E(a)$ has 4 tangent points to $\T$.
\end{itemize}

\subsubsection{Tangent points to $\T$ of the closed exterior envelopes} \label{SSS:TangentExt}

To complete the picture of the mass points of $\mu_{a,b}$ we need to know the tangent points to $\T$ of $\overline{E}_e(a)$.

We know that every line $a_\zeta$, $\zeta\in\T$, includes all the values $a\in\D$ with $\zeta$ as a common tangent point of $E(a)$ to $\T$. However,
$\zeta$ lies on $E_e(a)$ if and only if $\zeta\in\Sigma_a$, i.e., $\re(\overline{a}\zeta)<|a|^2$. Bearing in mind that the parametric equation of
$a_\zeta$ is
$$
a_\zeta(\lambda) = \zeta + \lambda\zeta^2, \qquad \lambda\in\R,
$$
we find that the curves separating the points of the lines $\{a_\zeta(\lambda)\}_{\zeta\in\T}$ lying on the exterior and the interior envelopes are
given by
$$
\re(\overline{a_\zeta(\lambda)}\zeta)=|a_\zeta(\lambda)|^2 \; \Leftrightarrow \;
\begin{cases} \lambda=0, \\ \lambda=-\re\zeta. \end{cases}
$$
The curve corresponding to $\lambda=0$ is the unit circle, so it does not impose any limitation to the values $a\in\D$. Writing $\zeta=e^{it}$, the
remaining curve is given by $\lambda=-\cos t$, so it has the form
$$
\hat{a}(t) = e^{it} - e^{2it}\cos t = \frac{1}{2}e^{it} - \frac{1}{2}e^{3it},
$$
which is an epitrochoid inscribed on the unit circle with two loops and two self-intersections at the points $\pm\frac{1}{\sqrt{2}}$ (see figure
\ref{F:CC-CCC}).

The set $\hat{A}=\{\hat{a}(t):t\in[0,2\pi]\}$ is formed by all the values $a\in\D$ with $\overline{E}_e(a)$ tangent to $\T$ at some point of
$\partial E_e(a)$, i.e., at some limit point. Given $a\in\hat{A}$, the number of tangencies at the limit points $\partial E_e(a)$ is equal to the
number of lines $a_\zeta$, $\zeta\in\T$, passing through $a$, i.e., the number of tangent lines to the epitrochoid at the given point $a$. Therefore,
all the values of $a$ lying on the epitrochoid have a single limit point where $\overline{E}_e(a)$ is tangent to $\T$, except the self-intersections
$a=\pm\frac{1}{\sqrt{2}}$, in which case $\overline{E}_e(a)$ is tangent to $\T$ at two limit points (see figure \ref{F:CC-CCC} and the right column
of figure \ref{F:Z+ab}).

For any value $a\in\D\setminus\hat{A}$, $\overline{E}_e(a)$ has no tangent point to $\T$ on $\partial E_e(a)$, thus the tangent points to $\T$ of
$\overline{E}_e(a)$ must lie on $E_e(a)$. The set $\D\setminus\hat{A}$ splits into 6 open connected regions. By continuity, the number of points of
$\overline{E}_e(a)$ which are tangent to $\T$ must be constant on each of these connected regions. Also, continuity arguments together with the fact
that the envelopes $E^+(a)$ and $E^-(a)$ do not intersect, ensure that the distribution of tangent points between $\overline{E}^+_e(a)$ and
$\overline{E}^-_e(a)$ must be the same inside each of the above connected regions. Hence, the picture for the number of tangent points to $\T$ of the
closed exterior envelopes $\overline{E}^\pm_e(a)$ can be completed by simply calculating such number for one value of $a$ in each connected region.
This gives the following results (see figures \ref{F:haz}, \ref{F:CC-CCC} and the left column of figure \ref{F:Z+ab}):
\begin{itemize}
\item[($\mathbf{T}_e^{0+1}$)] If $a$ lies inside the loops of the epitrochoid, $\overline{E}_e(a)$ has 1 tangent point to $\T$
and it lies on one of the open envelopes $E^\pm_e(a)$.
\item[($\mathbf{T}_e^{1+1}$)] If $a$ lies inside the epitrochoid but outside the loops, $\overline{E}_e(a)$ has 2 tangent points to $\T$,
one in each open envelope $E^\pm_e(a)$.
\item[($\mathbf{T}_e^{1+2}$)] If $a$ lies outside the epitrochoid, $\overline{E}_e(a)$ has 3 tangent points to $\T$,
two of them in one of the open envelopes $E^\pm_e(a)$ and the other one in the remaining open envelope.
\end{itemize}

When $a\in\hat{A}$ we know that $\overline{E}_e(a)$ has one or two limit points tangent to $\T$, depending whether $a\neq\pm\frac{1}{\sqrt{2}}$ or
$a=\pm\frac{1}{\sqrt{2}}$. The complete picture for $a$ lying on the epitrochoid can be inferred by continuity from the previous results (see figure
\ref{F:CC-CCC} and the right column of figure \ref{F:Z+ab}):
\begin{itemize}
\item[$(\mathbf{T}_e^{1+\overline{1}})$] If $a\neq\pm\frac{1}{\sqrt{2}}$ lies on the loops of the epitrochoid, $\overline{E}_e(a)$ has 2 tangent points
to $\T$,
one in an open envelope $E^\pm_e(a)$ while the other one is a limit point of the remaining envelope.
\item[($\mathbf{T}_e^{1+\overline{2}}$)] If $a=\pm\frac{1}{\sqrt{2}}$, $\overline{E}_e(a)$ has 3 tangent points to $\T$,
one in an open envelope $E^\pm_e(a)$ while the other two are the limit points of the remaining envelope.
\item[($\mathbf{T}_e^{1+1\overline{1}}$)] If $a$ lies on the epitrochoid but outside the loops, $\overline{E}_e(a)$ has 3 tangent points to $\T$,
two of them in different open envelopes $E^\pm_e(a)$ while the other one is a limit point.
\end{itemize}

\subsection{Localization pictures on $\Z_+$: dependence on $a$ and $b$} \label{SS:Loc-ab-Z+}

The previous section shows that $\mathbf{T}_e^2 \equiv (\mathbf{T}_e^{1+2} \text{ or } \mathbf{T}_e^{1+\overline{2}} \text{ or }
\mathbf{T}_e^{1+1\overline{1}})$ picks up the values $a\in\D$ such that one of the closed exterior envelopes $\overline{E}^\pm_e(a)$ has 2 tangent
points to $\T$. The sector associated with the connected component of $E_e(a) \setminus T(a)$ ending in such tangent points fulfills $\D$ (see the
cases $a_3$, $a_5$ and $a_6$ in figures \ref{F:CC-CCC} and \ref{F:Z+ab}). Therefore, the region covered by $\mathbf{T}^2_e$, which is the closed
exterior of the epitrochoid $\hat{A}$, gives all the points $a$ with $\mu_{a,b}$ having mass points for any $b\in\D$.

Consider now a value $a\in\D$ satisfying $\mathbf{T}_e^{1+1}$. In this case $E_e(a) \setminus T(a)$ has 4 connected components, each of them ending
in a tangent point to $\T$ and a limit point. The associated sectors of $\D$ are all included in one of them (see the case $a_2$ in figures
\ref{F:CC-CCC} and \ref{F:Z+ab}). Therefore, such a dominant sector provides the values $b\in\D$ with $\mu_{a,b}$ having mass points. In the case
$\mathbf{T}_e^{1+\overline{1}}$ we have a similar conclusion, although only 3 connected components appear in $E_e(a) \setminus T(a)$ (see the case
$a_4$ in figures \ref{F:CC-CCC} and \ref{F:Z+ab}). Summarizing, the condition $\mathbf{T}_e^1 \equiv (\mathbf{T}_e^{1+1} \text{ or }
\mathbf{T}_e^{1+\overline{1}})$, which means that $a$ lies on the interior of $\hat{A}$ but not on the interior of the loops, ensures that
localization holds for all the values of $b$ bounded by a single limit line.

Finally, assume $\mathbf{T}_e^{0+1}$ for $a\in\D$. Then, there are 3 connected components in $E_e(a) \setminus T(a)$, whose related sectors are
included in that one bounded by two limit lines (see the case $a_1$ in figures \ref{F:CC-CCC} and \ref{F:Z+ab}). The points of this dominant sector
are the values of $b$ giving mass points for $\mu_{a,b}$.

\begin{figure}
\includegraphics[width=5.5cm]{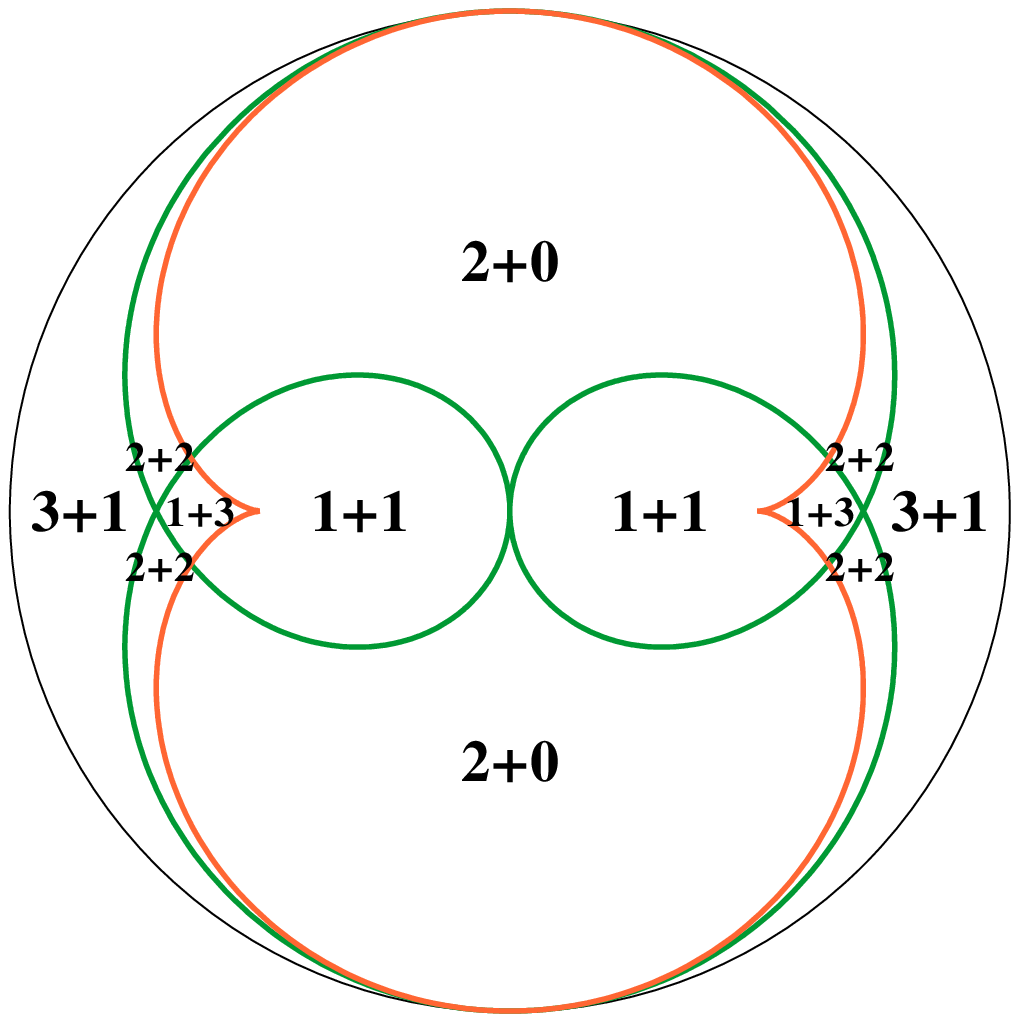} \hspace{.5cm} \includegraphics[width=5.5cm]{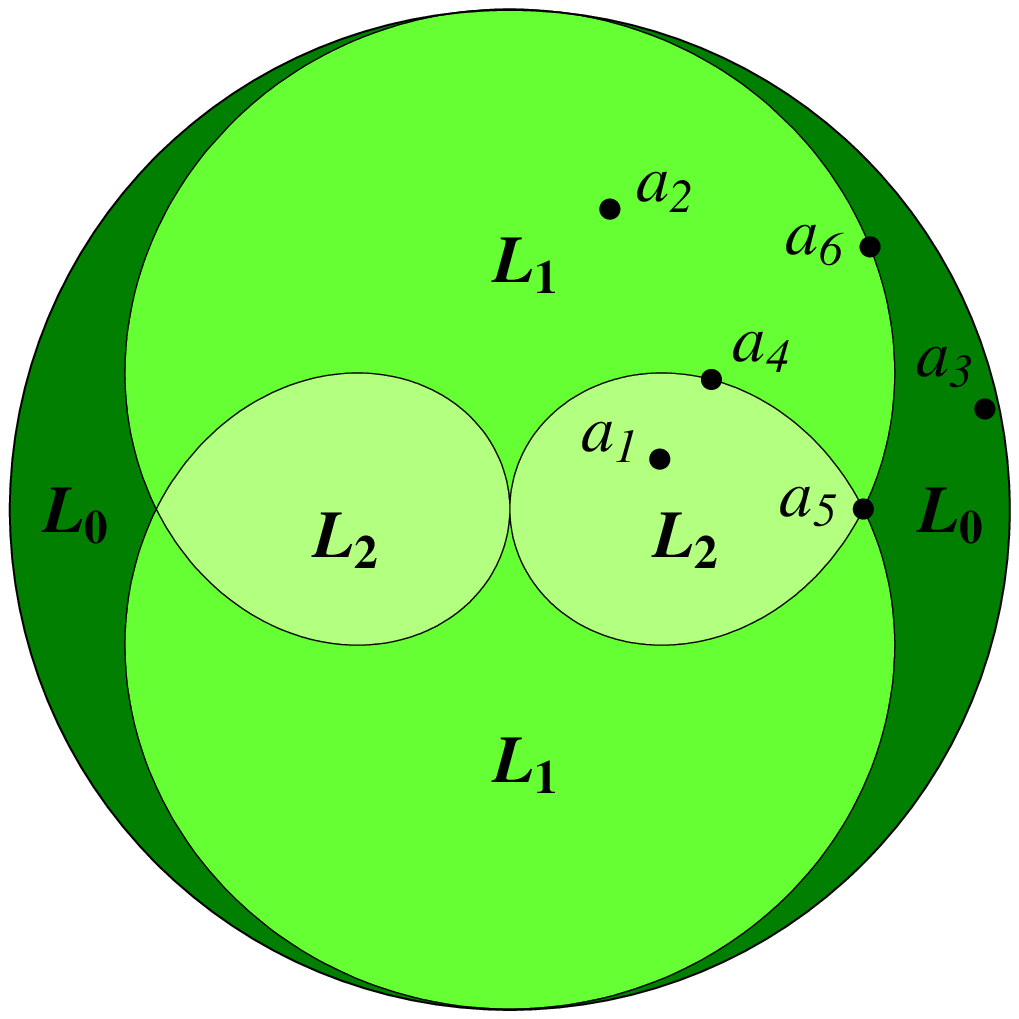}
\caption{The left figure represents the epicycloid (in orange) and the epitrochoid $\hat{A}$ (in green). The self-intersections of $\hat{A}$ are at
$\pm\frac{1}{\sqrt{2}}$. For all the values of $a$ in a given region enclosed by these cycloids the number of tangent points to $\T$ is constant for
both $E_e(a)$ (left number) and $E_i(a)$ (right number). Each crossing of the epicycloid from the exterior to the interior reduces the total number
of tangencies to $\T$ of $E_i(a)$ in 2, keeping invariant the number of tangencies of $E_e(a)$. On the other hand, any crossing of $\hat{A}$ from the
exterior to the interior changes a tangency of $E_e(a)$ into a tangency of $E_i(a)$. The number of tangencies to $\T$ of $E_e(a)$ is determined
exclusively by $\hat{A}$. This number, going from 1 to 3 as $a$ runs over the unit disk, is indicated by the darkness of the green color in the right
figure.} \label{F:CC-CCC}
\end{figure}

\begin{figure}
\hspace{-1cm} \includegraphics[width=5.1cm]{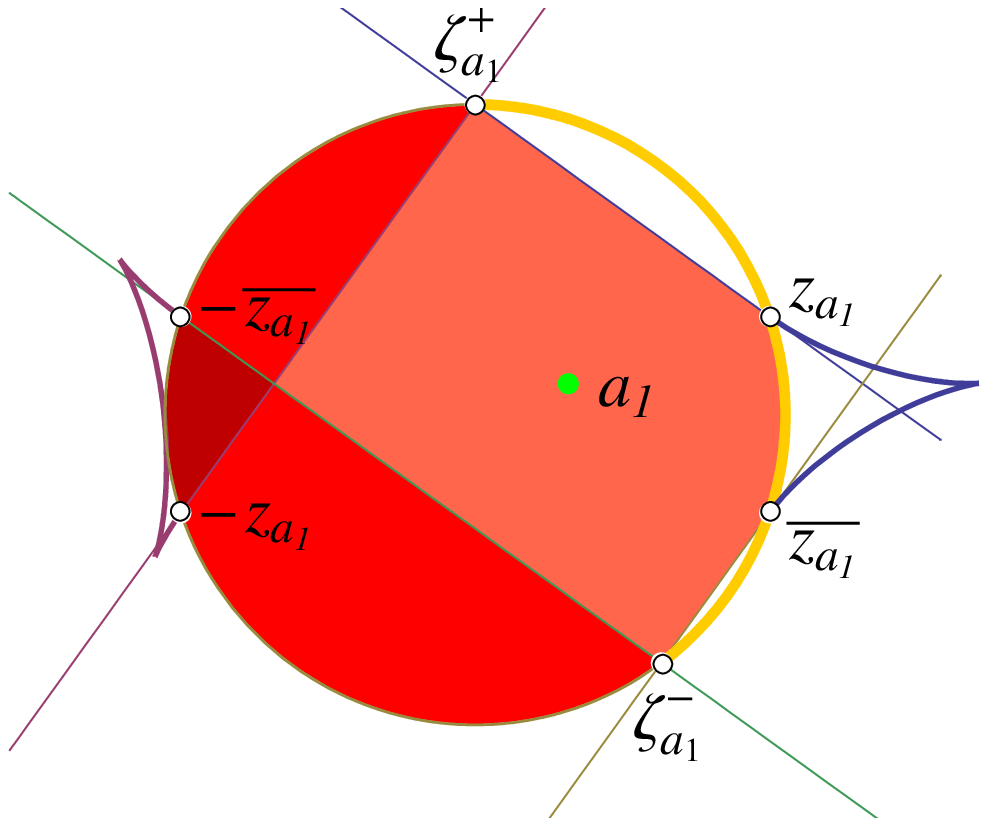} \hspace{.3cm} \includegraphics[width=5.6cm]{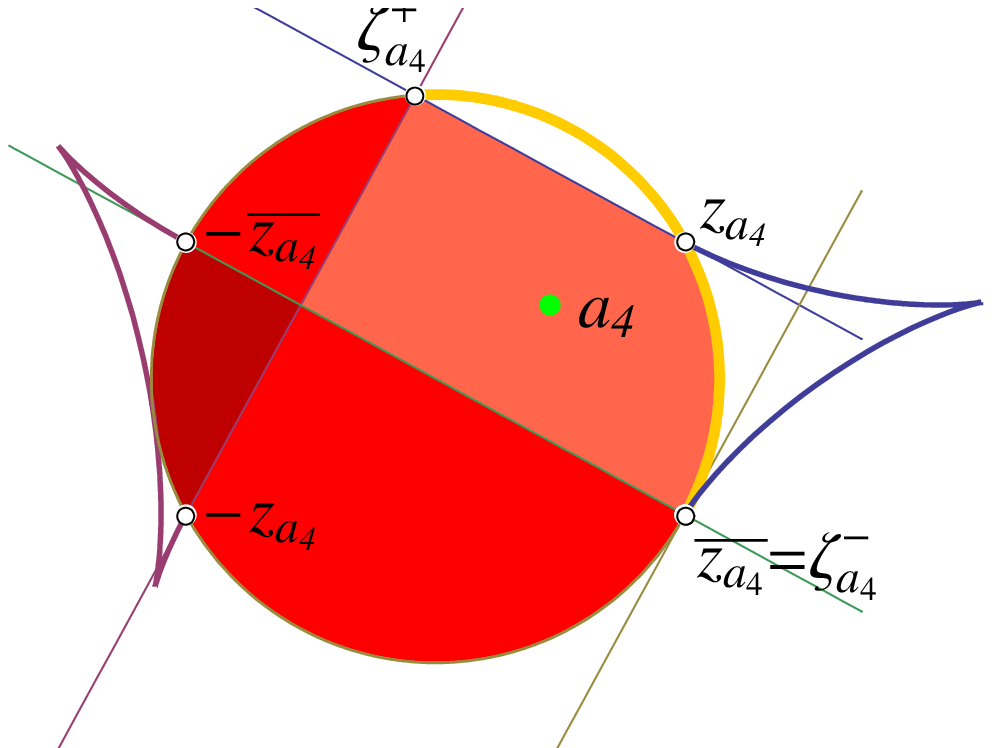}

\vspace{-.3cm}

\hspace{.4cm} \includegraphics[width=6.1cm]{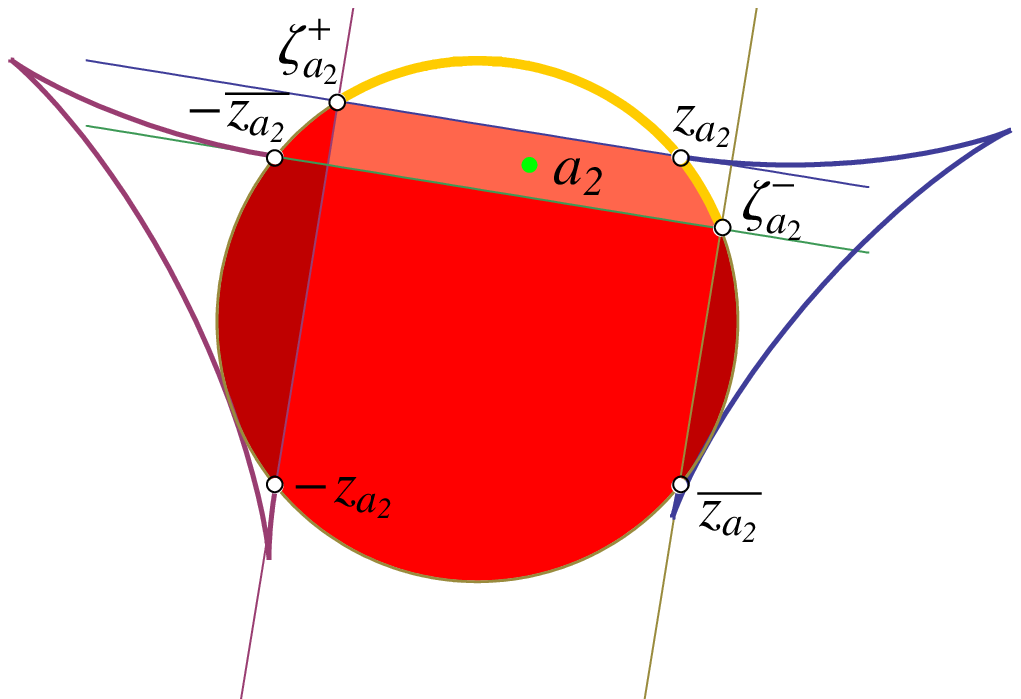} \includegraphics[width=7.3cm]{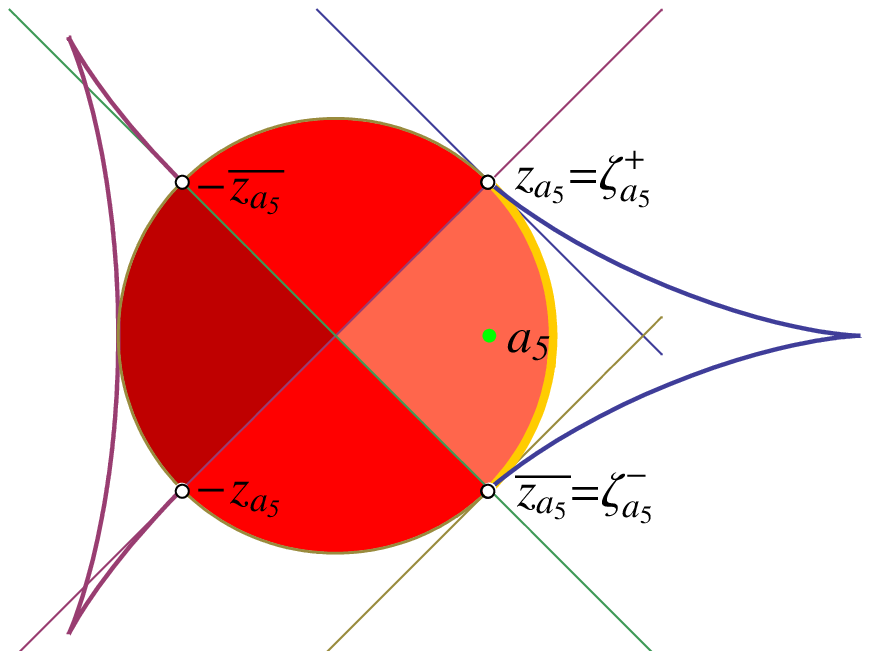}

\vspace{-.2cm}

\hspace{.7cm} \includegraphics[width=7.3cm]{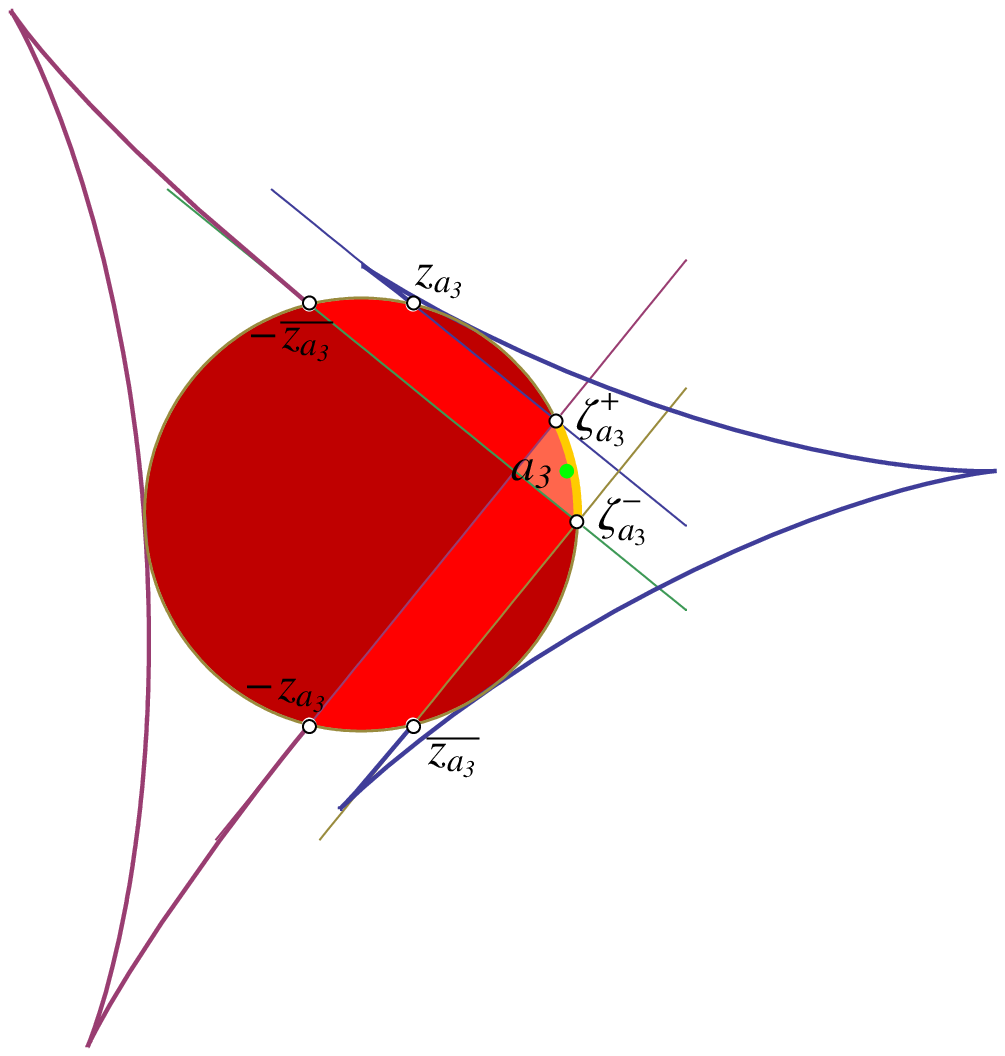} \hspace{-1.1cm} \includegraphics[width=7.6cm]{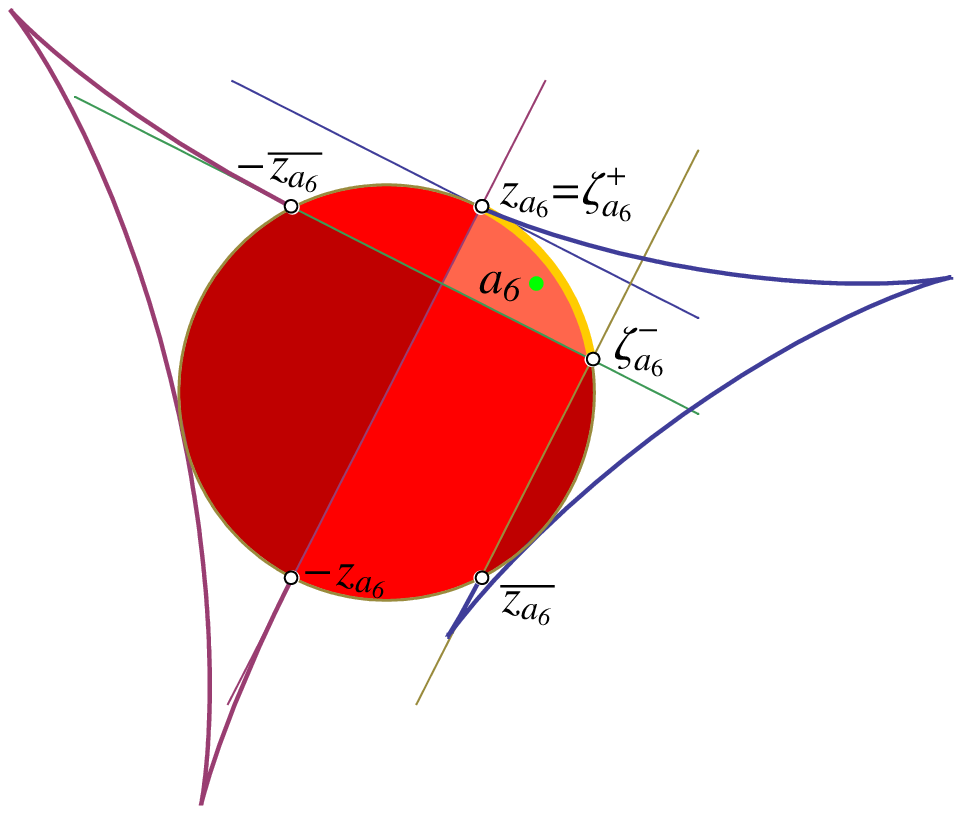}

\caption{{\bf Localization for one defect on $\Z_+$ (from $a$ to $b$).} In red the values of $b$ giving localization for each of the values of $a$ in
the right hand side of figure \ref{F:CC-CCC}. The values of $a$ in the first column are those already represented in figure \ref{F:haz}, and the
darkness of the color shows the number of overlapping sectors swept by the straight lines corresponding to the connected components of $E_e(a)
\setminus T(a)$, so it indicates the number of mass points (running from 1 to 3) for the related value of $b$. The right column represents similar
figures for different borderline situations corresponding to values of $a$ lying on the epitrochoid $\hat{A}$. The value $a_1$ leaves two
disconnected bands for $b$ without localization (condition $\mathbf{L}_2$). For $a_2$ and $a_4$ a unique band of values of $b$ is free of
localization (condition $\mathbf{L}_1$). The values $a_3$, $a_5$ and $a_6$ yield localization for any $b$ (condition $\mathbf{L}_0$). Given $a$, the
bands of $b$ which are localization free are bounded by the limit lines crossing the arc $\T\setminus\Sigma_a$ (in yellow).} \label{F:Z+ab}
\end{figure}

This provides the following localization criteria, which are the analogue of ($\mathbf{A}$) and ($\mathbf{B}_\pm$) for the case of $\Z_+$ at the end
of section \ref{S:Loc-defect-Z} (see figures \ref{F:CC-CCC} and \ref{F:Z+ab}): localization holds for the values $b\in\D$ bounded by
\begin{itemize}
\item[($\mathbf{L}_0$)] No limit line $\kern5pt \Leftrightarrow \kern5pt a$ lies on the closed exterior of $\hat{A}$.
\item[($\mathbf{L}_1$)] One limit line $\kern5pt \Leftrightarrow \kern5pt a$ lies on the interior of $\hat{A}$,
but not on the interior of the loops.
\item[($\mathbf{L}_2$)] Two limit lines $\kern5pt \Leftrightarrow \kern5pt a$ lies on the interior of the loops of $\hat{A}$.
\end{itemize}

An independent argument proves that, like in the case of $\Z_+$, for any $a\in\D$, there exists localization for all the values $b$ lying on the open
set $S(a)$ defined in (\ref{E:S}), limited by the arc $\Sigma_a$ and the straight line joining $\zeta_a^+$ and $\zeta_a^-$ (see figure
\ref{F:z-zeta}). The reason for this is that, when $\zeta$ runs over $\overline{\Sigma}_a$, the phase of $\zeta-a$ performs a rotation of an angle
$\pi$, so the orthogonal straight lines $b_{a,\zeta}^\pm$ rotate by an angle $\pi/2$. Then, geometric arguments show that the pair of orthogonal
families $\{b_{a,\zeta}^\pm\}_{\zeta\in\Sigma_a}$ sweep a region of $\D$ limited by at most two of the limit lines and that this region includes the
set $S(a)$.

As a consequence, there are at most two limit lines crossing the open arc $\T\setminus\overline{\Sigma}_a$, and these limit lines define in any case
the sector of values $b\in\D$ giving localization. Hence, the number of limit lines crossing $\T\setminus\overline{\Sigma}_a$ is $k$ in the case
$\mathbf{L}_k$ (see figures \ref{F:CC-CCC} and \ref{F:Z+ab}). In any of the cases $\mathbf{L}_k$ the measure $\mu_{a,b}$ can have 1, 2 or 3 mass
points, but only in the cases $\mathbf{L}_1$ and $\mathbf{L}_2$ it can have no mass points.

In the borderline case $a\in\hat{A}$ the envelope $E_e(a)$ is tangent to $\T$ at a limit point, so that a limit line is tangent to $\T$ at such a
point. Then $\partial\Sigma_a$ and $\partial\Gamma_a$ must have in common such a limit point (see figures \ref{F:CC-CCC} and \ref{F:Z+ab}).

When $a$ is imaginary the limit lines are parallel to the real and imaginary axes because $\zeta_a^\pm-a\in\R$, and the two limit lines which are
parallel to the real axis coincide because $\im\zeta_a^+=\im\zeta_a^-$. In such a case the region of values of $b$ giving localization becomes
exactly $S(a)$. Therefore, as in the case of $\Z$, among the values of $a$ with the same modulus the imaginary ones provide the largest region for
values of $b$ without localization.

\section{Asymptotic return probabilities: one defect on $\Z_+$} \label{S:Ret-defect-Z+}

Like in the case of $\Z$, we can perform the computation of the asymptotics of $p_{\alpha,\beta}^{(k)}(n)$ using (\ref{E:ARP}) and the canonical
representative $d\hat{\mu}=d\mu_{a,b}$ instead of the measure $d\mu(z)=d\hat{\mu}(e^{-i\vartheta}z)$ of the QW. From (\ref{E:psi}) and
(\ref{E:rot-Z+}) we find that the corresponding functions for the state $|\Psi_{\alpha,\beta}^{(k)}\>$ are related by
$\bs{\psi}_{\alpha,\beta}^{(k)}(z) = \hat{\bs{\psi}}_{\hat\alpha,\hat\beta}^{(k)}(e^{-i\vartheta}z)$ with $\hat\alpha=\hat\lambda_{2k}\alpha$ and
$\hat\beta=\hat\lambda_{2k+1}\beta$. Hence, $\bs{\psi}_{1,0}^{(k)}(z) = \hat\lambda_{2k} \hat{\bs{\psi}}_{1,0}^{(k)}(e^{-i\vartheta}z)$,
$\bs{\psi}_{0,1}^{(k)}(z) = \hat\lambda_{2k+1} \hat{\bs{\psi}}_{0,1}^{(k)}(e^{-i\vartheta}z)$ and (\ref{E:ARP}) can be expressed as
\begin{equation} \label{E:ARP-hat2}
\kern-5pt p_{\alpha,\beta}^{(k)}(n) \underset{n}{\sim} \text{\SMALL
$\left|\sum_{z\in\T} z^n\hat{\bs{\psi}}_{\hat\alpha,\hat\beta}^{(k)}(z)\hat{\mu}(\{z\})\overline{\hat{\bs{\psi}}_{1,0}^{(k)}(z)}\right|^2 +
\left|\sum_{z\in\T} z^n\hat{\bs{\psi}}_{\hat\alpha,\hat\beta}^{(k)}(z)\hat{\mu}(\{z\})\overline{\hat{\bs{\psi}}_{0,1}^{(k)}(z)}\right|^2$}.
\end{equation}
Again, we will omit the hat on $\hat\alpha,\hat\beta$ making the substitution $\alpha,\beta\to\hat\alpha,\hat\beta$ at the end of the calculations.

\subsection{Masses of $\hat{\mu}=\mu_{a,b}$} \label{SS:masses-Z+}

Any mass point of $\mu_{a,b}$ has the form
\begin{equation} \label{E:+-z-gen}
z_0 = \pm\frac{1-\overline{a}\zeta_0}{|1-\overline{a}\zeta_0|}\in\Gamma_a^\pm,
\qquad \lambda_0=\mp\frac{(\zeta_0-b)^2}{\zeta_0-a}>0,
\qquad \zeta_0\in\Sigma_a.
\end{equation}
Since $h_{a,b}(z_0)=z_0f_{a,b}(z_0)=1$, the corresponding mass is given by
\begin{equation} \label{E:mu-z0-Z+}
\hat{\mu}(\{z_0\}) = \frac{1}{z_0h'_{a,b}(z_0)} = \frac{1}{1+z_0^2f'_{a,b}(z_0)}.
\end{equation}

Performing the change of variables $\zeta(z)=-z^2f_a(z)$ we obtain
\begin{equation} \label{E:der-fab}
f_{a,b} = -\frac{\zeta-b}{1-\overline{b}\zeta},
\qquad
f'_{a,b} = -\zeta' \frac{\rho_b^2}{(1-\overline{b}\zeta)^2}.
\end{equation}
The expression of $\zeta'(z_0)$ remains as in (\ref{E:der-zeta}) for $z_0\in\Gamma_a^+$, but has opposite sign for $z_0\in\Gamma_a^-$. Therefore,
$$
f'_{a,b}(z_0) = \mp
\frac{\rho_b^2|1-\overline{a}\zeta_0|}{|a|^2-\re(\overline{a}\zeta_0)}\frac{1-a\overline{\zeta}_0}{(1-\overline{b}\zeta_0)^2}\zeta_0, \qquad
z_0\in\Gamma_a^\pm,
$$
which finally yields
\begin{equation} \label{E:m-z0-Z+}
\begin{aligned}
\hat{\mu}(\{z_0\}) & = \frac{1}{1 + \ds \frac{\rho_b^2|1-\overline{a}\zeta_0|}{|a|^2-\re(\overline{a}\zeta_0)}\frac{1}{\lambda_0}} =
\frac{1}{1 + \ds \frac{|\zeta_0-a|^2}{|a|^2-\re(\overline{a}\zeta_0)} \frac{\rho_b^2}{|\zeta_0-b|^2}}
\\
& = \frac{1}{1 + \ds 2 \frac{\frac{\rho_b^2}{|\zeta_0-b|^2}}{1-\frac{\rho_a^2}{|\zeta_0-a|^2}}}.
\end{aligned}
\end{equation}

\subsection{Asymptotics of $p_{\alpha,\beta}^{(0)}(n)$ on $\Z_+$} \label{SS:ARP-0-Z+}

We need the function $\hat{\bs{\psi}}_{\alpha,\beta}^{(0)}=\alpha\hat{X}_0+\beta\hat{X}_1$, where $\hat{X}_0(z)=1$ and
$\hat{X}_1(z)=(z^{-1}-b)/\rho_b$ follows from the the first two equations of $\hat{\cC}\hat{X}(z)=z\hat{X}(z)$. Hence,
$$
\hat{\bs{\psi}}_{\alpha,\beta}^{(0)}(z)=\alpha+\frac{\beta}{\rho_b}(z^{-1}-b).
$$

Any mass point $z_0$ satisfies $z_0f_{a,b}(z_0)=1$, which using (\ref{E:der-fab}) gives
$$
\overline{z_0}-b = -\frac{\rho_b^2}{\overline{\zeta_0}-\overline{b}}.
$$
Therefore,
$$
\begin{aligned}
& \hat{\bs{\psi}}_{\alpha,\beta}^{(0)}(z_0) \hat{\mu}(\{z_0\}) \overline{\hat{\bs{\psi}}_{1,0}^{(0)}(z_0)} =
\hat{\mu}(\{z_0\}) \left(\alpha-\beta\frac{\rho_b}{\overline{\zeta_0}-\overline{b}}\right),
\\
& \hat{\bs{\psi}}_{\alpha,\beta}^{(0)}(z_0) \hat{\mu}(\{z_0\}) \overline{\hat{\bs{\psi}}_{0,1}^{(0)}(z_0)} =
- \hat{\mu}(\{z_0\}) \frac{\rho_b}{\zeta_0-b} \left(\alpha-\beta\frac{\rho_b}{\overline{\zeta_0}-\overline{b}}\right).
\end{aligned}
$$

The cases with more than one mass point give in general a non convergent return probability $p^{(0)}_{\alpha,\beta}(n)$ due to the different factors
$z^n$ appearing in (\ref{E:ARP-hat2}). Nevertheless, the case with only one mass point $z_0$ yields
$$
\begin{aligned}
\lim_{n\to\infty} p^{(0)}_{\alpha,\beta}(n) & =
\hat{\mu}(\{z_0\})^2 \left(1+\frac{\rho_b^2}{|\zeta_0-b|^2}\right) \left|\hat\alpha-\hat\beta\frac{\rho_b}{\overline{\zeta_0}-\overline{b}}\right|^2
\\
& = \frac{1+\frac{\rho_b^2}{|\zeta_0-b|^2}}{\left(1 + 2 \frac{\frac{\rho_b^2}{|\zeta_0-b|^2}}{1-\frac{\rho_a^2}{|\zeta_0-a|^2}}\right)^2}
\left|\hat\alpha-\hat\beta\frac{\rho_b}{\overline{\zeta_0}-\overline{b}}\right|^2.
\end{aligned}
$$
Then, all the states at the origin exhibit localization except that one defined by
$$
\hat\beta = \hat\alpha \frac{\overline{\zeta_0}-\overline{b}}{\rho_b}.
$$


\begin{thebibliography}{99}

\bibitem{AM}
M. Aizenman and S. Molchanov, {\em Localization at large disorder and at extreme energies: An elemantary derivation}, Commun. Math. Phys.
\textbf{157} (1993) 245--278.

\bibitem{ABNVW}
A. Ambainis, E. Bach, A. Nayak, A. Vishwanath, J. Watrous, {\em One-dimensional quantum walks}, in Proceedings of the 33rd Annual ACM Symposium on
Theory of Computing, 2001, pp. 37--49.

\bibitem{And}
P. W. Anderson, {\em Absence of Diffusion in Certain Random Lattices}, Phys. Rev. \textbf{109} (1958) 1492--1505.


\bibitem{CMV}
M. J. Cantero, L. Moral, L. Vel\'azquez, {\em Five-diagonal matrices and zeros of orthogonal polynomials on the unit circle}, Linear Algebra Appl.
\textbf{362} (2003) 29--56.

\bibitem{CGMV}
M. J. Cantero, F. A. Gr\"unbaum, L. Moral and L. Vel\'azquez, {\em Matrix valued Szeg\H o polynomials and quantum random walks}, Commun. Pure Applied
Math. \textbf{58} (2010) 464--507.

\bibitem{C}
A. Clayton, {\em Quasi-birth-and-death processes and matrix-valued orthogonal polynomials}, SIAM J. Matrix Anal. Appl. \textbf{31} (5) (2010)
2239--2260.

\bibitem{DaKiSi}
D. Damanik, R. Killip, B. Simon, {\em Perturbations of orthogonal polynomials with periodic recursion coefficients}, Ann. of Math. \textbf{171}
(2010) 1931--2010.

\bibitem{DaPuSi}
D. Damanik, A. Pushnitski, B. Simon, {\em The analytic theory of matrix orthogonal polynomials}, Surveys in Approximation Theory \textbf{4} (2008)
1--85.

\bibitem{DRSZ}
H. Dette, B. Reuther, W. Studden, M. Zygmunt, {\em Matrix measures and random walks with a block tridiagonal transition matrix}, SIAM J. Matrix Anal.
Appl. \textbf{29} (1) (2006) 117--142.

\bibitem{Fa}
K. Falconer, {\em Fractal Geometry. Mathematical Foundations and Applications}, John Wiley \& Sons, Chichester, New York, 1990.


\bibitem{Go}
L. Golinskii, {\em Absolutely continuous measures on the unit circle with sparse Verblunsky coefficients}, Mat. Fiz. Anal. Geom.  \textbf{11} (4)
(2004) 408--420.

\bibitem{G1}
F. A. Gr\"unbaum, {\em Random walks and orthogonal polynomials: some challenges}, in Probability, Geometry and Integrable systems, Mark Pinsky and
Bjorn Birnir editors, MSRI publication vol. 55, 2007, pp. 241--260, see also arXiv math PR/0703375.

\bibitem{G2}
F. A. Gr\"unbaum, {\em QBD processes and matrix valued orthogonal polynomials: some new explicit examples}, Dagstuhl Seminar Proceedings 07461,
Numerical Methods in structured Markov Chains, 2007, D. Bini editor.

\bibitem{G3}
F. A. Gr\"unbaum, {\em Block tridiagonal matrices and a beefed-up version of the Ehrenfest urn model}, iin Operator Theory: Advances and
Applications, vol. 190, 2009, pp. 267--277.

\bibitem{G4}
F. A. Gr\"unbaum, {\em The Karlin-McGregor formula for a variant of a discrete version of Walsh's spider}, J. Phys. A: Math. Theor. \textbf{42}
(2009) 454010.

\bibitem{G5}
F. A. Gr\"unbaum, {\em A spectral weight matrix for a discrete version of Walsh's spider}, in Operator Theory: Advances and Applications, vol. 202,
2010, pp. 253--264.

\bibitem{GdI}
F. A. Gr\"unbaum, M. D. de la Iglesia, {\em Matrix valued orthogonal polynomials arising from group representation theory and a family of
quasi-birth-and-death processes}, Siam J. Matrix Anal. Appl. \textbf{30} (2) (2008) 741--763.

\bibitem{HJS}
E. Hamza, A. Joye, G. Stolz, {\em Dynamical Localization for unitary Anderson models}, Math. Phys. Anal. Geom. \textbf{12} (2009) 381--444.

\bibitem{IK}
N. Inui, N. Konno, {\em Localization of multi-state quantum walk in one dimension}, Phys. A \textbf{353} (2005) 133--144.

\bibitem{IKS}
N. Inui, N. Konno, E. Segawa, {\em One-dimensional three-state quantum walk}, Phys. Rev. E \textbf{72} (2005) 056112.

\bibitem{IKK}
N. Inui, Y. Konishi, N. Konno, {\em Localization of two-dimensional quantum walks}, Phys. Rev. A \textbf{69} (2004) 052323.

\bibitem{JM}
A. Joye, M. Merkli, {\em Dynamical Localization of Quantum Walks in Random Environments}, J. Stat. Phys. \textbf{140} (6) (2010) 1--29.

\bibitem{KMcG}
S. Karlin, J. McGregor, {\em Random walks}, Illinois J. Math. \textbf{3} (1959) 66--81.

\bibitem{Kn}
O. Knill, {\em A remark on quantum dynamics}, Helv. Phys. Acta \textbf{71} (1998) 233--241.

\bibitem{Ko}
N. Konno, {\em One-dimensional discrete-time quantum walks on random environments}, Quantum Inf. Proc. \textbf{8} (2009) 387--399.

\bibitem{K}
N. Konno, {\em Localization of an inhomogeneous discrete-time quantum walk on the line}, Quantum Information Processing \textbf{9} (2010) 405--418.

\bibitem{KS}
N. Konno, E. Segawa, {\em Localization of discrete time quantum walks on a half line via the CGMV method}, arXiv 1008.5109v1 [quant-ph] 30 Aug 2010.

\bibitem{K1}
M. G. Krein, {\em Infinite $J$-matrices and a matrix moment problem}, Dokl. Akad. Nauk SSSR \textbf{69} (2) (1949) 125--128.

\bibitem{K2}
M. G. Krein, {\em  Fundamental aspects of the representation theory of hermitian operators with deficiency index $(m,m)$}, AMS Translations, Series
2, vol. 97, Providence, Rhode Island (1971), pp. 75--143.

\bibitem{L}
Y. Last, {\em Quantum dynamics and decompositions of singular continuous spectra}, J. Funct. Analysis \textbf{142} (1996) 406--445.

\bibitem{Po}
G. Pólya, {\em  \"Uber eine Aufgabe der Wahrscheinlichkeitsrechnung betreffend die Irrfahrt im Strassennetz}, Mathematische Annalen \textbf{84}
(1921) 149--160.

\bibitem{RIESZ}
F. Riesz, {\em Über die Randwerte einer analytischen Funktion}, Math. Z. \textbf{18} (1923) 87--95.

\bibitem{Si1}
B. Simon, {\em Orthogonal Polynomials on the Unit Circle, Part 1: Classical Theory}, AMS Colloq. Publ., vol. 54.1, AMS, Providence, RI, 2005.

\bibitem{Si2}
B. Simon, {\em Orthogonal Polynomials on the Unit Circle, Part 2: Spectral Theory}, AMS Colloq. Publ., vol. 54.2, AMS, Providence, RI, 2005.

\bibitem{SJK1}
M. \v{S}tefa\v{n}ák, I. Jex, T. Kiss, {\em Recurrence and Pólya number of quantum walks}, Phys. Rev. Lett. \textbf{100} (2008) 020501.

\bibitem{SJK2}
M. \v{S}tefa\v{n}ák, T. Kiss, I. Jex, {\em Recurrence properties of unbiased coined quantum walks on infinite d-dimensional lattices}, Phys. Rev. A
\textbf{78} (2008) 032306.

\bibitem{SJK3}
M. \v{S}tefa\v{n}ák, T. Kiss, I. Jex, {\em Recurrence of biased quantum walks on a line}, New. J. Phys. \textbf{11} (2009) 043027.

\bibitem{WKKK}
K. Watabe, N. Kobayashi, M. Katori, N. Konno, {\em Limit distributions of two-dimensional quantum walks}, Phys. Rev. A \textbf{77} (2008) 062331.

\end{thebibliography}
\end{document}